\documentclass[10pt]{article}

\usepackage[margin=1in]{geometry}
\usepackage{graphicx}
\usepackage{authblk}
\usepackage{amsmath,amssymb,amsthm,mathtools}
\usepackage{aliascnt}
\usepackage{bm}
\usepackage{booktabs}
\usepackage{array}
\usepackage{enumitem}
\usepackage{makecell}
\IfFileExists{algorithm.sty}{
    \usepackage{algorithm}
    \usepackage[noend]{algpseudocode}
}{
    \usepackage{float}
    \newfloat{algorithm}{tbp}{loa}
    \floatname{algorithm}{Algorithm}
    \newcounter{algline}[algorithm]
    \newcounter{algdepth}
    \newlength{\algindentwidth}
    \setlength{\algindentwidth}{1.45em}
    \newcommand{\AlgIndent}{%
        \hspace*{\dimexpr\value{algdepth}\algindentwidth\relax}%
    }
    \newenvironment{algorithmic}[1][]{
        \setcounter{algline}{0}
        \setcounter{algdepth}{0}
        \begin{list}{}{
            \setlength{\leftmargin}{2.6em}
            \setlength{\labelwidth}{2.0em}
            \setlength{\labelsep}{0.6em}
            \setlength{\itemsep}{0.15em}
            \setlength{\parsep}{0pt}
        }
    }{
        \end{list}
    }
    \newcommand{\State}{%
        \refstepcounter{algline}%
        \item[\thealgline.]%
        \AlgIndent%
    }
    \newcommand{\Statex}{%
        \item[]%
        \AlgIndent%
    }
    \newcommand{\Require}{\State\textbf{Input:}\ }
    \newcommand{\For}[1]{%
        \State\textbf{for }##1\textbf{ do}%
        \stepcounter{algdepth}%
    }
    \newcommand{\EndFor}{%
        \addtocounter{algdepth}{-1}%
    }
    \newcommand{\Return}{\textbf{return}\ }
}
\usepackage{microtype}

\usepackage{xcolor}
\usepackage[hidelinks, colorlinks = true]{hyperref}
\hypersetup{
    pageanchor=false,
    colorlinks=true,
    linkcolor=magenta,
    filecolor=magenta,
    urlcolor=cyan,
    citecolor=cyan
}

\usepackage[nameinlink,capitalise]{cleveref}
\makeatletter
\@ifundefined{ALG@line}{}{%
    \renewcommand{\theALG@line}{\thealgorithm.\arabic{ALG@line}}%
    \providecommand{\theHALG@line}{}%
    \renewcommand{\theHALG@line}{\thealgorithm.\arabic{ALG@line}}%
}
\makeatother
\usepackage[table]{xcolor}
\usepackage{multirow}
\usepackage{hhline}
\usepackage{array}

\newtheorem{theorem}{Theorem}[section]

\newaliascnt{lemma}{theorem}
\newtheorem{lemma}[lemma]{Lemma}
\aliascntresetthe{lemma}

\newaliascnt{corollary}{theorem}
\newtheorem{corollary}[corollary]{Corollary}
\aliascntresetthe{corollary}

\newaliascnt{proposition}{theorem}
\newtheorem{proposition}[proposition]{Proposition}
\aliascntresetthe{proposition}

\newaliascnt{fact}{theorem}
\newtheorem{fact}[fact]{Fact}
\aliascntresetthe{fact}

\theoremstyle{definition}
\newaliascnt{definition}{theorem}
\newtheorem{definition}[definition]{Definition}
\aliascntresetthe{definition}

\newaliascnt{assumption}{theorem}

\aliascntresetthe{assumption}

\theoremstyle{remark}
\newaliascnt{remark}{theorem}

\aliascntresetthe{remark}

\crefname{theorem}{theorem}{theorems}
\Crefname{theorem}{Theorem}{Theorems}
\crefname{lemma}{lemma}{lemmas}
\Crefname{lemma}{Lemma}{Lemmas}
\crefname{corollary}{corollary}{corollaries}
\Crefname{corollary}{Corollary}{Corollaries}
\crefname{proposition}{proposition}{propositions}
\Crefname{proposition}{Proposition}{Propositions}
\crefname{fact}{fact}{facts}
\Crefname{fact}{Fact}{Facts}
\crefname{definition}{definition}{definitions}
\Crefname{definition}{Definition}{Definitions}
\crefname{remark}{remark}{remarks}
\Crefname{remark}{Remark}{Remarks}
\crefname{algorithm}{algorithm}{algorithms}
\Crefname{algorithm}{Algorithm}{Algorithms}

\newcommand{\R}{\mathbb{R}}
\newcommand{\C}{\mathbb{C}}
\newcommand{\E}{\mathbb{E}}
\newcommand{\Prb}{\mathbb{P}}
\newcommand{\Var}{\operatorname{Var}}
\newcommand{\Tr}{\operatorname{Tr}}
\newcommand{\rank}{\operatorname{rank}}
\newcommand{\softmax}{\operatorname{softmax}}
\newcommand{\Gap}{\operatorname{Gap}}

\newcommand{\clip}{\operatorname{clip}}
\newcommand{\supp}{\operatorname{supp}}

\newcommand{\inner}[2]{\left\langle #1,#2\right\rangle}
\newcommand{\norm}[1]{\left\lVert #1\right\rVert}
\newcommand{\abs}[1]{\left\lvert #1\right\rvert}

\usepackage{ifdraft}
\newbool{AuthNotes} 
\booltrue{AuthNotes} 
\ifbool{AuthNotes}{\newcommand{\authnote}[3]{{\color{#3}({\bf  #1:} #2)}}}{\newcommand{\authnote}[3]{}}

\newcommand{\dn}{n}
\newcommand{\mc}{m}
\newcommand{\sr}{s}
\newcommand{\Rx}{R}
\newcommand{\Ry}{r}

\newcommand{\om}{\omega}
\newcommand{\eps}{\varepsilon}
\newcommand{\del}{\delta}
\newcommand{\TT}{T}
\newcommand{\gam}{\gamma}
\newcommand{\Lam}{\Lambda}
\newcommand{\LamL}{\Lambda_{\mathrm L}}

\newcommand{\lrX}{\eta_X}
\newcommand{\lrY}{\eta_y}

\newcommand{\bA}{\mathbf{A}}
\newcommand{\bB}{\mathbf{B}}
\newcommand{\bC}{\mathbf{C}}
\newcommand{\bH}{\mathbf{H}}
\newcommand{\bI}{\mathbf{I}}
\newcommand{\bK}{\mathbf{K}}
\newcommand{\bM}{\mathbf{M}}
\newcommand{\bN}{\mathbf{N}}

\newcommand{\bG}{\mathbf{G}}
\newcommand{\bX}{\mathbf{X}}
\newcommand{\bY}{\mathbf{Y}}
\newcommand{\bBmat}{\mathbf{B}}

\newcommand{\bD}{\mathbf{D}}

\newcommand{\bp}{\mathbf{p}}
\newcommand{\by}{\mathbf{y}}
\newcommand{\bz}{\mathbf{z}}
\newcommand{\bc}{\mathbf{c}}
\newcommand{\bell}{\bm{\ell}}
\newcommand{\be}{\mathbf{e}}
\newcommand{\bu}{\mathbf{u}}
\newcommand{\bv}{\mathbf{v}}
\newcommand{\bw}{\mathbf{w}}
\newcommand{\bx}{\mathbf{x}}
\newcommand{\bzero}{\mathbf{0}}
\newcommand{\bpsi}{\bm{\psi}}
\newcommand{\brho}{\bm{\rho}}
\newcommand{\bmu}{\bm{\mu}}
\newcommand{\bxi}{\bm{\xi}}
\newcommand{\bbias}{\mathbf{b}}

\newcommand{\cA}{\mathcal{A}}
\newcommand{\cE}{\mathcal{E}}
\newcommand{\cF}{\mathcal{F}}
\newcommand{\cP}{\mathcal{P}}
\newcommand{\cG}{\mathcal{G}}
\newcommand{\cS}{\mathcal{S}}
\newcommand{\A}[1]{\bA_{#1}}
\newcommand{\Ht}[1]{\bH_{#1}}
\newcommand{\Gt}[1]{\bG_{#1}}
\newcommand{\Xtilde}[1]{\widetilde{\bX}_{#1}}
\newcommand{\rtilde}[1]{\widetilde{\brho}_{#1}}
\newcommand{\psitilde}[1]{\widetilde{\bpsi}_{#1}}
\newcommand{\phat}[1]{\widehat{\bp}_{#1}}
\newcommand{\zt}[1]{\bz_{#1}}
\newcommand{\ctilde}[1]{\widetilde{\bc}_{#1}}
\newcommand{\chat}[1]{\widehat{\bc}_{#1}}
\newcommand{\lhat}[1]{\widehat{\bell}_{#1}}
\newcommand{\jt}[1]{J_{#1}}
\newcommand{\isamp}[1]{I_{#1}}
\newcommand{\ut}[1]{\bu_{#1}}
\newcommand{\vtilde}[1]{\widetilde{\bv}_{#1}}
\newcommand{\vt}[1]{\bv_{#1}}
\newcommand{\psit}[1]{\bpsi_{#1}}
\newcommand{\rt}[1]{\brho_{#1}}
\newcommand{\Xt}[1]{\bX_{#1}}
\newcommand{\Fil}[1]{\cF_{#1}}

\newcommand{\xbar}[1]{\overline{\bX}_{#1}}

\newcommand{\Gibbs}[1]{\mathbf{\Omega}_{#1}}

\newcommand{\Xout}{\bX^{\mathrm{out}}}
\newcommand{\zout}{\bz^{\mathrm{out}}}

\newcommand{\MatSide}{\mathcal E_X(\TT)}
\newcommand{\VecSide}{\mathcal E_y(\TT)}

\newcommand{\GoodL}{\mathcal G_{\mathrm L}}

\newcommand{\Zest}{Z}
\newcommand{\Bloc}{B}
\newcommand{\LanczosExp}{\operatorname{L\acute{a}nczosExp}}

\newcommand{\Rfy}{\Phi}
\newcommand{\Cfy}{\Psi}
\newcommand{\Osparse}{\mathcal O_{\mathrm{sparse}}}

\usepackage{authblk}

\title{Solving Sparse SDPs in Sublinear Time:\\A Classical Algorithm Inspired by the Quantum OR Lemma}
\author[1,2]{Fernando G.S.L.~Brandão\thanks{\url{fbrandao@amazon.com}}}
\author[1]{Alexander M.~Dalzell\thanks{\url{dalzel@amazon.com}}}
\author[3]{András Gilyén\thanks{\url{gilyen.andras@renyi.hu}}}
\author[1,4]{Francisca Vasconcelos\thanks{\url{francisca@berkeley.edu}}}

\affil[1]{AWS Center for Quantum Computing}
\affil[2]{California Institute of Technology}
\affil[3]{Alfréd Rényi Institute of Mathematics}
\affil[4]{University of California, Berkeley}

\date{}

\begin{document}

\maketitle

\begin{abstract}
We give the first sublinear-time classical solvers for sparse semidefinite programs in the bounded-radius regime, without low-rank assumptions or Frobenius norm dependence on the constraint matrices. For constant precision and bounded primal and dual radii, prior quantum algorithms of Brand\~ao et al. (2019) and van Apeldoorn and Gily\'en (2019) achieved $\widetilde{O}(\sqrt{n}+\sqrt{m})$ dependence on matrix dimension $n$ and constraint number $m$. Compared with the $\widetilde{O}(mn)$ runtime of existing classical methods, this suggests a quartic quantum speedup when $m \approx n$. Beyond a usual Grover speedup, this separation relies on the Quantum OR lemma, whose sample-reuse mechanism decouples the cost of Gibbs-state preparation from constraint search. We show that this reuse mechanism is classically realizable for sparse SDPs.

Our main technical contribution is a classical procedure for simultaneously estimating many expectation values with respect to a sparse Hamiltonian's Gibbs state. This combines randomized L\'anczos filtering with an efficient sampling-based estimator. We also introduce a stochastic online-learning framework for SDP solving, substantially improving accuracy-dependence over standard oracle-based MMWU approaches. Let $s$ denote the the input matrix sparsity and $\gamma:=Rr/\varepsilon$ capture dependence on the primal $(R)$ and dual $(r)$ radii as well as target accuracy $(\varepsilon)$. When $\gamma^2\leq\min\{m,n/s\}$, our solver runs in time $\widetilde{O}\left(ns\gamma^{4.5}+ms\gamma^2\right)$. For $\gamma=O(1)$, this is $\widetilde{O}\left((n+m)s\right)$ and sublinear in the $O(mns)$ input size. Similar to the quantum algorithms, this matches known lower bounds with respect to $m$ and $n$, up to logarithmic factors. This implies that, with respect to dimensions $m$ and $n$, there is no super-quadratic quantum advantage for generic sparse SDP solving.
\end{abstract}

\thispagestyle{empty}
\clearpage
\begingroup
\setcounter{tocdepth}{2}
\tableofcontents
\thispagestyle{empty}
\endgroup
\clearpage
\pagenumbering{arabic}

\section{Introduction}
\label{sec:introduction}

Semidefinite programs (SDPs) constitute a central class of convex
optimization problems.  They underlie approximation algorithms such as the
Goemans--Williamson algorithm for \textsc{Max-Cut}~\cite{GW95},
polynomial optimization through the sum-of-squares hierarchy
\cite{Lasserre2001,Parrilo2000}, and basic optimization tasks in quantum
information, including optimal measurement design~\cite{Eldar2003} and
bounds on nonlocal games~\cite{CHTW2004}.  We consider SDPs with primal
form
\begin{align}
\operatorname{OPT}
&=
\max_{\bX\succeq0}
\left\{
\Tr(\bC\bX):
\Tr(\A{j}\bX)\leq b_j,
~~
\forall~j\in[\mc]
\right\}
\label{eqn:intro-standard-sdp-primal}
\end{align}
and dual form
\begin{align}
\operatorname{OPT}
&=
\min_{\by\in\R_{\geq0}^{\mc}}
\left\{
\mathbf b^\top\by:
\sum_{j=1}^{\mc}y_j\A{j}\succeq\bC
\right\},
\label{eqn:intro-standard-sdp-dual}
\end{align}
where
$\bC,\A{1},\ldots,\A{\mc}\in\C^{\dn\times\dn}$ are
$\sr$-row-sparse Hermitian matrices with operator norm at most one.  We
assume strong duality, with optimal primal and dual solutions satisfying
$\Tr(\bX^\star)\leq\Rx$ and
$\norm{\by^\star}_1\leq\Ry$, respectively.  We take the supplied bounds $\Rx,\Ry\geq1$, enlarging them if necessary. Throughout, by solving the SDP to additive error $\eps$, we mean
returning (an implicit representation of) a matrix $\widehat{\bX}\succeq\bzero$
 such that $\Tr(\widehat{\bX})\leq\Rx$, 
$\Tr(\bC\widehat{\bX})\geq\operatorname{OPT}-\eps$, and
$\Tr(\A{j}\widehat{\bX})\leq b_j+\eps/\Ry$, for every $j\in[\mc]$. For additive error $\eps$, write
$\gam:=\Rx\Ry/\eps$ for the effective inverse-accuracy parameter, which
combines the target precision with the scales of the optimal primal and
dual solutions. For the purpose of this work, we view $\gamma$ as a parameter independent of the dimensions $n$ and $m$ and often quote the dimension-specific scaling obtained by setting $\gamma = O(1)$.

We work in a classical sparse-row oracle model, analogous to the sparse
access used by prior quantum SDP solvers.  Given a matrix identifier, a row
$i$, and a position $k$, the oracle returns the location and value of the
$k$-th nonzero entry in that row, while each threshold $b_j$ is available
by random access.  Hence one row can be read in $O(\sr)$ time without
scanning the full $\Theta(\mc\dn\sr)$-entry sparse representation.  Since
a sublinear algorithm also cannot materialize a generic dense primal
matrix, our solvers return implicit primal certificates and sparse dual
descriptions.  Formal definitions, output conventions, and the reduction
to the normalized game appear in
\Cref{sec:preliminaries}.

For this class of SDPs, quantum algorithms appear to achieve a
\emph{quartic} speedup over state-of-the-art classical solvers for input dimensions $\dn$ and $\mc$.  In particular, the best known generic
classical solver has runtime dimension dependence
$\widetilde O(\dn\mc)$~\cite{CDST19}, whereas the best quantum
solvers achieve complexity
$\widetilde O(\sqrt{\dn}+\sqrt{\mc})$~\cite{BKL19,AG19}.  Thus, in
the balanced regime $m=O(\dn)$, the known bounds scale as
$\widetilde O(\dn^2)$ and $\widetilde O(\sqrt \dn)$, respectively.  Such a
super-quadratic separation is particularly significant because the prevailing conventional wisdom~\cite{Babbush2021} suggests that quadratic
speedups may be insufficient to overcome the overheads of fault-tolerant
quantum computation, whereas quartic speedups are considerably more
promising.  It is therefore important to determine
whether this apparent quartic speed-up offers a genuinely novel quantum advantage or suggests removable bottlenecks in classical SDP
algorithms.

The apparent quartic quantum speedup for SDPs can be decomposed into two conceptually
distinct quadratic improvements. The first is the familiar Grover speedup,
which replaces linear dependence on $\dn$ and $\mc$ by square-root dependence
and is known to be optimal in the black-box model. Indeed, the original
quantum SDP solver of Brandão and Svore~\cite{BrandaoSvore} combined
quantum Gibbs-state preparation with Grover-accelerated constraint search,
yielding $\widetilde O(\sqrt{\dn\mc})$ dependence on the two input
dimensions. The second improvement, which reduces this multiplicative
dependence to $\widetilde O(\sqrt{\dn}+\sqrt{\mc})$, comes from the Fast
Quantum OR lemma of Brandão et al.~\cite{BKL19}. Building on the
Quantum OR bound of Harrow, Lin, and Montanaro~\cite{HLM17} as well as the fast
QMA-amplification techniques of Nagaj, Wocjan, and Zhang~\cite{pavel},
the Fast Quantum OR reuses expensive Gibbs-state samples
across many constraint tests, preserving them through gentle
measurement. This reuse decouples the cost of preparing the Gibbs state
from the cost of searching over constraints, replacing the product of the
$\dn$- and $\mc$-dependent terms by a sum. For general SDPs, it has remained
unclear whether this reuse mechanism provides an additional fundamental
source of quantum advantage beyond Grover search. This motivates our first
question of study:
\begin{quote}
\emph{Does the Quantum OR provide a novel, fundamental source of quadratic
quantum advantage?}
\end{quote}
%
%
We answer this question in the \textit{negative}, at least for the case when the Quantum OR is applied to sparse observables and Gibbs states of sparse Hamiltonians, as in the SDP application.
Concretely, we dequantize the structured Quantum OR
procedure for SDP solving by developing a classical algorithm with additive
runtime in $\dn$ and $\mc$. The key technical contribution is a classical
algorithm for simultaneously estimating many Gibbs expectation values using only a few classically computed Gibbs samples. It combines
randomized Lánczos filtering with an estimator inspired by variational Monte
Carlo~\cite{FoulkesEtAl2001}, reusing each sampled coordinate across all
observables. This yields additive dependence on $\dn$ and $\mc$ and
sample complexity controlled by operator norm, unlike  standard dequantization methods, e.g. \cite{Tang2019}, which are controlled by Frobenius norm. Substituting this procedure into the
quantum-style matrix-multiplicative-weights framework yields a classical
SDP solver with runtime
\begin{align}
\widetilde O\left(
\sr\dn\gam^{6.5}
+
\sr\mc\gam^4
\right).
\end{align}
When compared with the quantum~\cite{BKL19,AG19} running
time
\begin{align}
\widetilde O\left(
\sr\sqrt{\dn}\gam^5
+
\sr\sqrt{\mc}\gam^4
\right),
\end{align}
our result demonstrates that the efficient decoupling in $\dn$ and $\mc$ achieved by the Quantum OR is not inherent to quantum algorithms.  With this dequantization, the only remaining dimensional gap is the standard Grover-type quadratic speedup. Indeed, it was shown \cite[Section IIC]{BrandaoSvore} by reduction from the Grover search problem that solving sparse SDPs with $\gamma = O(1)$ requires $\Omega(n+m)$ classical queries or $\Omega(\sqrt{n}+\sqrt{m})$ quantum queries, meaning our $\widetilde O(n+m)$ complexity matches the optimal classical dimension-dependence up to logarithmic factors. 

However, this direct dequantization of state-of-the-art quantum SDP solvers
does not yet yield our most efficient classical SDP solver.  Its dependence
on the accuracy parameter $\gam$ is suboptimal because it retains the
oracle-based MMWU architecture of the quantum algorithms, in which every
round requires a fresh, high-accuracy Gibbs-expectation computation.  In
contrast, Carmon et al.~\cite{CDST19} obtain improved accuracy dependence
through a \emph{stochastic optimization framework}.  Rather than computing an
accurate matrix response on each round, they sample randomized rank-one
actions and control their error only over the full optimization trajectory.
This suggests moving beyond the traditional oracle-based MMWU framework altogether
and motivates our second main question:
\begin{quote}
\centering
\emph{Can our Quantum OR dequantization be combined with a
\textbf{stochastic} SDP framework to obtain a faster classical solver?}
\end{quote}
We answer this question in the \emph{affirmative}. To do so, we model the SDP optimization problem as a zero-sum game. In this context, we build on the
one-sided stochastic rank-one framework of Carmon et al.~\cite{CDST19}, but modify \emph{both} players to exploit the sampling ideas developed in our Quantum OR dequantization.  A sampled coordinate of the approximate rank-one
matrix response supplies simultaneous stochastic payoff estimates to
the constraint learner, while a sampled constraint gives a sparse
representation of the Hamiltonian update.  These changes eliminate the
need for a separate high-accuracy constraint search in each iteration, but
introduce stochastic feedback, clipping bias, approximation bias, and
coupled fluctuations that are absent from the original analysis of \cite{CDST19}.  We
control these effects through a new high-probability online-learning analysis that
combines second-order vector-regret bounds, same-history coupling for the
rank-one matrix response, and martingale concentration.  The resulting
analysis shows that the errors average out over the trajectory and that
the averaged primal and dual iterates converge to an approximate saddle
point.

Overall, our final algorithm has running time
\begin{align}
\widetilde O\left(
\dn\sr
\min\left\{
\gam^2,
\mc,
\frac{\dn}{\sr}
\right\}
\gam^{2.5}
+ \mc\sr\gam^2
\right).
\end{align}
In the regime
$\gam^2\leq\min\{\mc,\dn/\sr\}$, this simplifies to
\begin{align}
    \widetilde O\left(\dn\sr\gam^{4.5}+\mc\sr\gam^2\right)
\end{align}
The resulting solver therefore retains the additive dependence on $\dn$ and $\mc$
obtained by the direct dequantization algorithm, while substantially improving its
dependence on $\gam$. From the runtime bounds summarized in \Cref{tab:intro-runtime-comparison}, one can identify regimes of $m$, $n$, and $\varepsilon$ in which our algorithm outperforms all prior state-of-the-art classical algorithms and even the best known quantum algorithms. Taken together, our results show that the Quantum OR sample-reuse mechanism admits an efficient classical realization and that coupling this classical reuse with stochastic rank-one updates for both players in the SDP game yields an even faster classical solver.

\subsection{Comparison with Prior Works}
\label{sec:intro-runtime-comparison}

\begin{table*}[ht!]
\centering
\scriptsize
\renewcommand{\arraystretch}{1}
\setlength{\tabcolsep}{3pt}
\newcommand{\headone}[1]{\raisebox{0.55ex}{#1}}
\newcommand{\headtwo}[2]{\shortstack[c]{#1\\[-0.55ex]#2}}
\newcommand{\runtimepad}{\rule[-4.0ex]{0pt}{10.0ex}}
\newcommand{\twoline}[2]{\raisebox{\dimexpr-0.5\height+0.5\ht\strutbox\relax}{\shortstack[c]{#1\\#2}}}
\begin{tabular}{|c|c|c|c|c|c|}
\hline
\rowcolor{gray!10}
\headone{Framework} & \headone{Model} & \headone{Paper} & \headone{Runtime}
& \headtwo{Dimension}{Dependence} & \headtwo{Accuracy}{Dependence} \\
\hline
\hline
\rowcolor{blue!7}
\cellcolor{white} & & Brand\~ao et al.\ (2019)~\cite{BKL19} &
\runtimepad $\displaystyle
\widetilde O\!\left(\sr^2\left(\frac{\sqrt\dn}{\eps^{12}}+
\frac{\sqrt\mc}{\eps^{10}}\right)\right)$ &
$\widetilde O(\sqrt\mc+\sqrt\dn)$ & $\widetilde O(\eps^{-12})$ \\
\hhline{|~|~|----|}
\rowcolor{blue!7}
\cellcolor{white} & \multirow[c]{-2}{*}[4ex]{Quantum} &
\twoline{van Apeldoorn--Gily\'en}{ (2019)~\cite{AG19}} &
\runtimepad $\displaystyle
\widetilde O\!\left(\sr\left(\frac{\sqrt\dn}{\eps^5}+
\frac{\sqrt\mc}{\eps^4}\right)\right)$ &
$\widetilde O(\sqrt\mc+\sqrt\dn)$ & $\widetilde O(\eps^{-5})$ \\
\hhline{|~|-----|}
& & \twoline{Arora--Kale (2007)~\cite{AroraKale2016}}{ see \cite[Sec.~2.4]{AGGW17}}  &
\runtimepad $\displaystyle
\widetilde O\!\left(\frac{\mc\dn\sr}{\eps^4}+
\frac{\dn\sr}{\eps^7}\right)$ &
$\widetilde O(\mc\dn)$ & $\widetilde O(\eps^{-7})$ \\
\hhline{|~|~|----|}
\multirow[c]{-4}{*}[10ex]{MMWU} &
\multirow[c]{-2}{*}[4ex]{Classical} &
This Work (\Cref{thm:intro-direct-solver}) &
\runtimepad $\displaystyle
\widetilde O\!\left(\frac{\dn\sr}{\eps^{6.5}}+
\frac{\mc\sr}{\eps^4}\right)$ &
\cellcolor{green!15}$\widetilde O(\mc+\dn)$ &
\cellcolor{green!15}$\widetilde O(\eps^{-6.5})$ \\
\hline
\hline
& & Garber--Hazan (2011)~\cite{GarberHazan2016} &
\runtimepad $\displaystyle
\widetilde O\!\left(\frac{\mc F}{\eps^2}+
\frac{\dn\sr}{\eps^4}+
\min\left\{\frac{\dn\sr}{\eps^{4.5}},
\frac{\dn^2}{\eps^{2.5}}\right\}\right)$ &
$\widetilde O(\mc \dn)$ & $\widetilde O(\eps^{-4.5})$ \\
\hhline{|~|~|----|}
& & Carmon et al.\ (2019)~\cite{CDST19} &
\runtimepad $\displaystyle
\widetilde O\!\left(\frac{\mc\dn\sr}{\eps^2}+
\frac{1}{\eps^{2.5}}\min\{\mc\dn\sr,\dn^2\}\right)$ &
$\widetilde O(\mc\dn)$ &
\cellcolor{green!15}$\widetilde O(\eps^{-2.5})$ \\
\hhline{|~|~|----|}
\multirow[c]{-3}{*}[6.5ex]{Stochastic} &
\multirow[c]{-3}{*}[6.5ex]{Classical} &
This Work (\Cref{thm:intro-regret-solver}) &
\runtimepad $\displaystyle
\widetilde O\!\left(
\frac{\dn\sr}{\eps^{2.5}}
\min\left\{\frac1{\eps^2},\mc,\frac\dn\sr\right\}+
\frac{\mc\sr}{\eps^2}\right)$ &
\cellcolor{green!15}$\widetilde O(\mc+\dn)$ &
$\widetilde O(\eps^{-4.5})$ \\
\hline
\hline
\rowcolor{orange!6}
\twoline{Low-Rank}{Dequantization} &
Classical &
\twoline{Chia et al.\ (2020)}{\cite[Corollary 6.25]{chia2020}} &
\runtimepad $\displaystyle
O\!\left(\frac{F^{11}\log^{23}\dn}{\eps^{46}}+\mc\frac{F^{7}\log^{13}\dn}{\eps^{28}}\right)\textsuperscript{$\dagger$}$ &
$O(\mc\dn^{7}+\dn^{11})$ &
$O(\eps^{-46})$ \\
\hline
\rowcolor{orange!6}
Interior Point &
Classical &
Huang et al.\ (2022)~\cite{HJSTZ22} &
\runtimepad $\displaystyle
O^*\!\left(
\bigl(\mc^{\tau}+\dn^{2\tau}\bigr)
\log\frac1\eps
\right)$ &
$O^*\!\left(\mc^{\tau}+\dn^{2\tau}\right)$ &
$O^*\!\left(\log\frac1\eps\right)$ \\
\hline
\end{tabular}
\caption{
Comparison of representative SDP runtimes under different input-access and structural assumptions. We assume instances are normalized to $\Rx=\Ry=\om=1$, so that $\gam=1/\eps$. We write $F:=\max_j\norm{\A{j}}_F^2$ and let $\tau$ denote an admissible matrix-multiplication exponent (which currently satisfies $\tau \geq 2.37$). The dimension column fixes accuracy and sparsity, and applies the worst-case substitution $F = \Theta(n)$, while the accuracy column reports the accuracy-dominated regime for bounds containing a minimum.  Blue rows denote quantum algorithms, orange rows denote methods with different structural or input-access assumptions, and green cells highlight the strongest classical dependence within the corresponding sparse-oracle framework. The notation $\widetilde O$ suppresses polylogarithmic factors, while $O^*$ additionally suppresses subpolynomial factors. \textsuperscript{$\dagger$}Assumes sample-and-query access to the input, which in the standard RAM model can be achieved in $O(1)$ time after some preprocessing that takes linear time in the number of non-zero matrix elements~\cite[Remark~3.11]{chia2020}.
}
\label{tab:intro-runtime-comparison}
\end{table*}

\Cref{tab:intro-runtime-comparison} places our results alongside prior
sparse-oracle algorithms, low-rank dequantizations, and full-input
high-accuracy methods. These results operate under different access models
and structural assumptions, so the comparison should not be interpreted as
holding all parameters and input models fixed. For the sparse-oracle
algorithms, we use the common unit-scale normalization
$\Rx=\Ry=\om=1$, where $\om:=\max\{\norm{\bC},\norm{\A{1}},\ldots,\norm{\A{\mc}}\}$, such that $\gam=1/\eps$. 
This is without loss of generality since, given the task of solving the general SDP in \Cref{eqn:intro-standard-sdp-primal,eqn:intro-standard-sdp-dual} to error $\eps$, one may  divide $\bC$ by $\omega r$, divide  $\A{j}$ by $\omega$, and divide $b_j$ by $\omega R$, and then equivalently solve the resulting unit-scale SDP to precision $\varepsilon/(\omega R r )$.
Under the change of variables $\bX'=\bX/\Rx$, solving the normalized
instance to accuracy $\eps/(\om\Rx\Ry)$ gives objective error at most
$\eps$ and constraint violations at most $\eps/\Ry$ in the original
units.
For a given family of instances, $\Rx$, $\Ry$, and $\om$ may or may not themselves
depend on $\dn$ and $\mc$. However, in the dimension column of the table we isolate only the explicit
dependence on $\dn$ and $\mc$, with remaining parameters fixed. The
low-rank dequantization and interior-point entries are included under their
own access and accuracy conventions, which we discuss below.

The main dimensional feature of both of our solvers is that their dependence
on $\dn$ and $\mc$ is additive. For fixed $\gam$ and sparsity $\sr$, the
oracle-based solver (\Cref{thm:intro-direct-solver}) and stochastic solver
(\Cref{thm:intro-regret-solver}) run in
$\widetilde O((\dn+\mc)\sr)$ time, compared with
$\Theta(\mc\dn\sr)$ entries in a saturated sparse-row input. Thus, in the
joint large-$\dn$, large-$\mc$ regime, they can run in time sublinear in
the input size while returning succinct solutions, rather than explicitly materializing a dense
$\dn\times\dn$ primal matrix. The stochastic solver further improves the
accuracy dependence while preserving this additive dimension dependence.

Prior sublinear classical algorithms obtain different tradeoffs by imposing
additional structure or stronger input access. Garber--Hazan
~\cite{GarberHazan2016} use an estimator whose second moment is controlled
by $F = \max_j \norm{\A{j}}_F^2$, so their sublinear guarantee is strongest when the
constraint Frobenius norms are small. For generic full-rank constraints,
however, $\norm{\A{j}}_F^2$ can be $\Theta(\dn)$ even when operator norm
$\norm{\A{j}}\leq1$ and sparsity $s = O(1)$ (consider the simple example that $\A{j} = \bI$), and as a result, their complexity returns to $O(mns)$. Similarly, Chia et al.~\cite{Chia2019,chia2020} obtain Tang-style
dequantizations with only polylogarithmic dependence on $\dn$ for low-rank or approximately low-rank matrices. Concretely, the normalized feasibility algorithm in Corollary 6.25 of Ref.~\cite{chia2020} has runtime
$O(\mc F^{7}\log^{13}\dn/\eps^{28} + F^{11}\log^{23}\dn/\varepsilon^{46})$, where $F\leq n$ is the same as defined above. In addition to requiring small $F$ to be competitive, this runtime requires the stronger sampling access model, which may hide a pre-processing cost linear in the input size. Our
algorithms instead assume only standard sparse-row access and require neither
low-rank constraints nor favorable Frobenius-norm bounds. Any vectors from
which our algorithms sample are explicitly constructed, with this cost included in the runtime.

Interior-point and cutting-plane methods occupy a complementary
high-accuracy regime~\cite{JLSW20,JKLPS20,HJSTZ22}. These methods achieve
polylogarithmic dependence on $1/\eps$, substantially better than the
polynomial accuracy dependence of the sparse-oracle algorithms in the
table. However, this improvement comes with a different computational tradeoff.
They operate in a full-input model and incur larger polynomial dependence
on $\dn$ and $\mc$, rather than seeking sublinear input dependence. For example, interior-point methods explicitly construct and solve $n \times n$ linear systems in each iteration, which can be dense even when the original SDP inputs are sparse, making sublinear dependence fundamentally out of reach.  We
include the interior-point method of Huang et al.~\cite{HJSTZ22} as a
representative result from this line of work. Thus, these methods can be
preferable in the very-high-accuracy regime, while our focus is on reducing
the dependence on the input dimensions when $\gam$ is bounded or
grows moderately.

Among sparse-oracle methods, our improved accuracy dependence builds most
directly on the randomized rank-one framework of Carmon et
al.~\cite{CDST19}. Their approach improves the accuracy scaling of the
matrix updates, but a generic sparse-row implementation still retains
multiplicative dependence on $\dn$ and $\mc$. Our stochastic solver combines
this rank-one approach with constraint sampling and stochastic payoff
estimation, preserving its favorable accuracy behavior while replacing the
multiplicative dimension dependence by an additive one. We develop this
connection in detail below.

Finally, the prior quantum solvers~\cite{BKL19,AG19} retain a square-root
advantage in both $\dn$ and $\mc$. Our algorithms classically realize the
Quantum OR sample-reuse mechanism responsible for separating the two
dimension-dependent costs, but they do not reproduce the Grover-type
square-root speedup in those costs. Consequently, which algorithm is
preferable depends on the dimensions, target accuracy, scale parameters,
and available input access. 

We also note a relevant related work by Franco Garrido et al.~\cite{franco2025socp}, which studies second-order cone programs (SOCPs), a special case of SDPs, and gives both quantum and classical
multiplicative-weights-based algorithms. They introduce a framework of sample reuse that achieves additive error for both their quantum and their classical algorithm.  However, their framework exploits the special low-rank structure of the Gibbs states arising
in SOCPs, and as a result, even in their quantum algorithm, the reused samples are unentangled classical strings. This can be viewed as sidestepping the Quantum OR Lemma in a special case, rather than actually simulating it---after all, the Quantum OR Lemma is nontrivial only in the case the samples are entangled, uncloneable states requiring processing via gentle measurement.   Consequently, the strategy does not clearly generalize to classical algorithms for general, high-rank SDPs, where it had remained open whether the same additive dependence
could be achieved, and the possibility of quartic speedup had persisted.

\subsection{Main Result 1: Dequantizing the Quantum OR Lemma for SDPs}
\label{sec:intro-fast-or}

We first isolate the sample-reuse mechanism underlying the Quantum OR and
show that it admits an efficient classical algorithm, in the structured
SDP setting. Here, the quantum state is not an arbitrary
unknown state. Rather, it is the Gibbs state of a sparse Hamiltonian, and each test is
determined by the expectation of a sparse observable. We reduce the resulting
search problem to the task of estimating many Gibbs expectations
simultaneously and give a classical estimator whose cost depends only additively on
$\dn$ and $\mc$. In
Part~2 (\Cref{sec:intro-regret-result}), we will discuss how this primitive results in efficient end-to-end SDP solvers.

Importantly, we do not attempt to directly classically simulate the gentle-measurement
procedure used by the Fast Quantum OR lemma of \cite{BKL19}. Instead, we
exploit the particular structure of its SDP application. The state is a Gibbs
state $\rho$ specified by a sparse Hamiltonian, while the tests are sparse
observables $\A{1},\ldots,\A{\mc}$, given explicitly as classical matrices. This
structure allows us to replace the quantum search procedure with the direct
classical estimation of the corresponding Gibbs expectations
$\Tr(\rho \A{j})$. Our approach also differs from the sample-and-query dequantization framework
popularized by Tang~\cite{Tang2019}. Rather than assuming specialized sampling access to the input, we work in the standard sparse-access model for SDP solving and use a Lánczos approximation procedure to estimate the relevant Gibbs expectations. In particular, our algorithm explicitly constructs and stores the length-$n$ and length-$m$ vectors it uses, while never materializing dense $n\times n$ matrices. Moreover, our bounds depend on the operator norms
of the constraint matrices rather than their Frobenius norms and, therefore,
apply to generic full-rank constraints. 

The rest of Part~1 develops the ingredients of the algorithm. Namely,
\Cref{sec:intro-fast-or-reduction} explains how the Fast Quantum OR search
problem reduces, in the SDP setting, to simultaneous Gibbs expectation
estimation, while \Cref{sec:intro-gibbs-estimator} develops an efficient
classical procedure for carrying out this estimation.

\subsubsection{The Fast Quantum OR Lemma and Its Reduction to Expectation Estimation}
\label{sec:intro-fast-or-reduction}

Let $\brho$ be a quantum state and let
$\mathbf{\Pi}_1,\ldots,\mathbf{\Pi}_{\mc}$ specify two-outcome measurements.
For each $j$, the measurement has outcomes ``accept'' and ``reject,'' with
corresponding POVM elements $\mathbf{\Pi}_j$ and $\bI-\mathbf{\Pi}_j$, where
$0\preceq\mathbf{\Pi}_j\preceq \bI$.  When this measurement is applied to
$\brho$, its acceptance probability is
$q_j:=\Tr(\mathbf{\Pi}_j\brho)$.  The abstract Fast Quantum OR problem asks to determine
whether there exists a $j$ for which $q_j$ is large, while consuming only a small number of copies of the
potentially expensive state $\brho$. Applying the measurements
sequentially on the same copy is  generally not possible because even a rejecting measurement
can disturb the state.  On the other hand, using independent copies of $\brho$ to probe $q_j$ for different $j$ \textit{multiplies} the high cost of preparing $\brho$ by the factor of $m$ (or $\sqrt{m}$ with Grover) cost of searching over values of $j$.  

The Fast Quantum OR lemma avoids this multiplication under an asymmetric
promise.  If either $\max_j q_j\geq1-\delta_1$ or
$\frac{1}{\mc}\sum_jq_j\leq\delta_2$, with
$\delta_2\leq(1-\delta_1)^2/(12\mc)$, then a single copy of $\brho$ and
$\widetilde O(\sqrt{\mc})$ coherent test applications suffice to
distinguish the two cases with constant bias~\cite{BKL19,AG19}.  The
procedure consumes the input copy and returns only an existence bit.
Gentle Quantum Search, or the equivalent Two-Phase Quantum Search
reduction, can then recover an accepting index using only
polylogarithmically (in $m$) many copies~\cite{BKL19,AG19}.  For SDP solving, the
key observation is that the same small collection of Gibbs states can support
the search over many constraints, rather than preparing a fresh Gibbs
state for each search query.

We aim to achieve an analogous classical sample-reuse mechanism. Rather than simulating the Fast Quantum OR as a black-box procedure for an unknown state and arbitrary measurements, however, we exploit the additional structure of the SDP setting and solve a stronger estimation problem tailored to this application. Specifically, in this setting, the state of interest is specified
implicitly as
$\Gibbs{\bH}:=e^{-\bH}/\Tr(e^{-\bH})$,
for a sparse Hamiltonian $\bH$, while the tests are determined by
$\sr$-row-sparse Hermitian observables
$\A{1},\ldots,\A{\mc}$.  The relevant quantities are collected in
the Gibbs expectation vector
\begin{align}
\bmu(\bH)
&:=
\bigl(
\Tr(\A{j}\Gibbs{\bH})
\bigr)_{j\in[\mc]}
\end{align}
and we wish to identify a $j$ for which $\mu_j(\bH)$ exceeds $b_j$ by at least a certain fixed margin---that is, we are searching for a violated constraint.
When $0\preceq\A{j}\preceq \bI$, we may take $\A{j}$ itself as the
``accept'' POVM element of a two-outcome measurement, whose acceptance
probability on a state $\brho$ is $\Tr(\A{j}\brho)$.  More generally, any
Hermitian observable satisfying $\norm{\A{j}}\leq1$ can be converted into a
valid two-outcome measurement by setting
$\mathbf{\Pi}_j=(\bI+\A{j})/2$.  Indeed,
$0\preceq\mathbf{\Pi}_j\preceq \bI$, and its acceptance probability is
\begin{align}
\Tr(\mathbf{\Pi}_j\brho)
&=
\frac{1+\Tr(\A{j}\brho)}{2}.
\end{align}
Thresholding this quantity against $(1+b_j+2\alpha)/2$, where $\alpha$ is a tolerance margin related to $\varepsilon$, and amplifying across the resulting acceptance-probability gap then produces the promise required by the Fast Quantum OR. This amplification uses additional copies of the Gibbs state, contributing to the accuracy-dependent overhead of the quantum procedure.

Instead of only deciding whether such a violating $j$ exists, as in Fast
Quantum OR, we consider the stronger task of producing an estimate
$\widehat{\bmu}$ of the entire Gibbs expectation vector $\bmu(\bH)$ with
entrywise additive accuracy. This task is not attempted by quantum methods
aiming for $\widetilde O(\sqrt{\mc})$ dependence, since even writing down
$\widehat{\bmu}$ requires $\Omega(\mc)$ time. If
\begin{align}
\max_j
\left|
\widehat\mu_j- \Tr(\A{j}\Gibbs{\bH})
\right|
&\leq
\alpha,
\end{align}
then thresholding the estimates with appropriate slack either identifies a
constraint violated by an additive $\Omega(\alpha)$ margin or certifies that
all constraints are satisfied up to additive $O(\alpha)$ margin. Thus,
simultaneous Gibbs expectation estimation directly implements the decision
and search tasks needed in the SDP application of the Fast Quantum OR, while
allowing the expensive Hamiltonian-dependent computation to be reused across
all observables.

\subsubsection{A Classical Simultaneous Gibbs Expectation Estimator}
\label{sec:intro-gibbs-estimator}

Our first main technical contribution solves this simultaneous Gibbs expectation estimation task using only a classical sparse description of $\bH$ and sparse-row access to the observables $\A{j}$. The algorithm neither prepares nor receives copies of a Gibbs state and it does not access the dense matrix $e^{-\bH}$ entrywise. Instead, it approximates the action of $e^{-\bH/2}$ on a collection of random vectors and reuses each resulting vector to estimate the expectations of all $\mc$ observables.

\begin{theorem}[Informal simultaneous Gibbs expectation estimator]
\label{thm:intro-gibbs-moment}
Let $\bH=\sum_{\ell=1}^{h}y_\ell\bB_\ell$ satisfy
$\norm{\by}_1\leq B$, where every $\bB_\ell$ is $\sr$-row-sparse and has
operator norm at most one.  For $\alpha,\del\in(0,1)$, a randomized
classical algorithm simultaneously estimates
$\Tr(\A{j}\Gibbs{\bH})$ for all $j\in[\mc]$ to additive error $\alpha$
with probability at least $1-\del$ in time
\begin{align}
    \widetilde O\!\left(
        \frac{\sr}{\alpha^2}
        \left[
            \dn h\min\{\dn,\sqrt{1+B}\}+\mc
        \right]
    \right).
    \label{eqn:intro-gibbs-moment-runtime}
\end{align}
\end{theorem}
\noindent \Cref{alg:intro-gibbs-estimator} summarizes the estimator.  We use
$K=\Theta(\alpha^{-2})$ samples in each of
$R=\Theta(\log(\mc/\del))$ independent groups.  Every sample performs one
Lánczos filtering step and one shared-coordinate readout, after which a
median-of-ratios aggregation produces the final estimates. We now
describe the key ingredients underlying these steps.

\begin{algorithm}[t]
\small
\caption{Simultaneous Gibbs Expectation Estimation}
\label{alg:intro-gibbs-estimator}
\begin{algorithmic}[1]
\Require Sparse Hamiltonian $\bH$, observables $(\A{j})_{j=1}^{\mc}$, accuracy $\alpha$, failure probability $\del$
\State \textbf{Parameters:}\quad $K\gets\Theta(\alpha^{-2})$, \quad $R\gets\Theta(\log(\mc/\del))$
\For{$r=1,\ldots,R$}
    \For{$k=1,\ldots,K$}
        \State Draw a Haar-random unit vector $\bu^{(r,k)}$
        \State Compute $\widetilde{\bv}^{(r,k)}\approx e^{-\bH/2}\bu^{(r,k)}$ by Lánczos
        \State $W^{(r,k)}\gets\|\widetilde{\bv}^{(r,k)}\|_2^2$, \quad
$\bpsi^{(r,k)}\gets\widetilde{\bv}^{(r,k)}/\|\widetilde{\bv}^{(r,k)}\|_2$
\State Sample $I^{(r,k)}\sim(|[\bpsi^{(r,k)}]_i|^2)_{i=1}^{\dn}$ in $O(\dn)$ time
        \State For every $j\in[\mc]$, set
        $X_j^{(r,k)}\gets
        W^{(r,k)}\operatorname{Re}\big\{
        (\A{j}\bpsi^{(r,k)})_{I^{(r,k)}}/(\bpsi^{(r,k)})_{I^{(r,k)}}
        \big\}$
    \EndFor
    \State For every $j\in[\mc]$, set
    $\widehat\mu_j^{(r)}\gets
    \big(\sum_{k=1}^{K}X_j^{(r,k)}\big)/
    \big(\sum_{k=1}^{K}W^{(r,k)}\big)$
\EndFor
\State \Return
$\widehat\mu_j\gets\operatorname{median}_{r\in[R]}\widehat\mu_j^{(r)}$
for every $j\in[\mc]$
\end{algorithmic}
\end{algorithm}

The first ingredient is a randomized rank-one representation of the
unnormalized Gibbs operator.  Draw a Haar-random unit vector $\bu$ and define
$\bv^{(\bu)}:=e^{-\bH/2}\bu$, together with
$W^{(\bu)}:=\norm{\bv^{(\bu)}}_2^2$ and
$\bpsi^{(\bu)}:=\bv^{(\bu)}/\|\bv^{(\bu)}\|_2$.  Haar isotropy gives
$\E_{\bu}[\bv^{(\bu)}(\bv^{(\bu)})^\dagger]=e^{-\bH}/\dn$, and hence
\begin{align}
    \Tr(\A{j}\Gibbs{\bH})
    &=
    \frac{
        \E_{\bu}\!\left[
            W^{(\bu)}
            \langle\bpsi^{(\bu)},\A{j}\bpsi^{(\bu)}\rangle
        \right]
    }{
        \E_{\bu}[W^{(\bu)}]
    }.
    \label{eqn:intro-weighted-gibbs-identity}
\end{align}
The weight $W^{(\bu)}$ restores the contribution lost when the filtered
vector is normalized.  In particular, the algorithm samples ordinary Haar starts and
keeps the norm as a statistical weight, rather than using rejection sampling
to draw directly from a Gibbs-dependent ensemble.

The remaining challenge is to estimate all $\mc$ expectations without applying every observable to the full $\dn$-dimensional vector. Once $\bpsi^{(\bu)}$ has been computed and stored, we can sample an index directly from its squared amplitudes in $O(\dn)$ time. Specifically, draw $I^{(\bu)}\sim\bigl(|[\bpsi^{(\bu)}]_i|^2\bigr)_{i\in[\dn]}$
and define the estimator
\begin{align} 
    Z_j^{(\bu)} &:= \operatorname{Re}\!\left( \frac{(\A{j}\bpsi^{(\bu)})_{I^{(\bu)}}} {[\bpsi^{(\bu)}]_{I^{(\bu)}}} \right). 
\end{align} 
This estimator mirrors the ``local energy'' sampling identity used in variational Monte Carlo~\cite[Section~III.C]{FoulkesEtAl2001}. It is unbiased, $\E_{I^{(\bu)}}[Z_j^{(\bu)}\mid\bpsi^{(\bu)}] =\langle\bpsi^{(\bu)},\A{j}\bpsi^{(\bu)}\rangle$, and its second moment is bounded by $\norm{\A{j}}^2$. Crucially, the sampling distribution depends only on $\bpsi^{(\bu)}$, so the same sampled coordinate can be reused for every observable. Evaluating $(\A{j}\bpsi^{(\bu)})_{I^{(\bu)}}$ then requires only row $I^{(\bu)}$ of $\A{j}$ and costs $O(\sr)$ time, allowing one coordinate sample to estimate all $\mc$ constraint energies in $O(\mc\sr)$ time.

This shared-coordinate reuse would not be useful if obtaining reliable
estimates required a dimension-dependent number of samples.  The
operator-norm second-moment bound prevents exactly this.  Although the
local-energy ratio can be large when the sampled amplitude is small, the
same coordinate is sampled with probability proportional to the square of
that amplitude.  These factors cancel in the second moment, giving
\begin{align}
    \E_{I^{(\bu)}}\big[
        (Z_j^{(\bu)})^2
        \mid\bpsi^{(\bu)}
    \big]
    &\leq
    \big\|\A{j}\bpsi^{(\bu)}\big\|_2^2
    \leq
    \norm{\A{j}}^2.
\end{align}
Hence, under $\norm{\A{j}}\leq1$, the sample complexity for fixed additive
accuracy is independent of $\dn$.  This contrasts with
Frobenius-controlled estimators, for which
$\norm{\A{j}}_F^2$ may be $\Theta(\dn)$ for the full-rank constraints
considered here.  Thus, the same mechanism that permits a shared coordinate
across all observables also avoids reintroducing a hidden dimension factor
through statistical averaging.

With this dimension-independent second-moment control in hand, it remains to
aggregate the weighted samples into simultaneous high-probability estimates.
Within each group, we average the numerator and denominator separately and
form their ratio. Taking the coordinatewise median of these independent
group-level ratio estimates then boosts the guarantee simultaneously over all
$\mc$ observables. This is analogous to the usual median-of-means
amplification, but applied to ratio estimates rather than directly to sample
means. Overall, $O(\alpha^{-2}\log(\mc/\del))$ random starts suffice. The
estimates across different observables are correlated by design because they
reuse the same filtered vectors and sampled coordinates, while independence
is required only across random starts.

The remaining issue is computational. The exact filtered vector $e^{-\bH/2}\bu$ is dense and cannot be formed by explicitly exponentiating $\bH$. We instead approximate it using randomized Lánczos with $\widetilde O(\min\{\dn,\sqrt{1+B}\})$ sparse Hamiltonian--vector products. After writing the resulting $\dn$-dimensional vector to memory, the additional $O(\dn)$ cost of normalizing it and constructing its $\ell_2$-sampling data structure is included in the filtering term of \Cref{thm:intro-gibbs-moment}. The random-start analysis of Carmon et al.~\cite{CDST19} provides the required relative approximation guarantee, while a stability argument propagates the resulting error through normalization, coordinate sampling, weighting, and the final ratio estimator. Consequently, each
random start incurs one $\dn$-dependent filtering cost and a separate
$\mc$-dependent readout cost, yielding the additive runtime in
\Cref{thm:intro-gibbs-moment}. The formal result,
\Cref{thm:classical-gibbs-moments}, combines the estimator moments from
\Cref{sec:energy-estimator}, the complex-Haar Lánczos guarantee from
\Cref{sec:lanczos}, and the robust-ratio and perturbation arguments from
\Cref{sec:robust-simultaneous-ratios,sec:gibbs-perturbation}.

\paragraph{Why This is a Nonstandard Dequantization.}
Unlike standard sample-and-query dequantization~\cite{Tang2019}, our algorithm assumes only sparse-row access to the SDP input.  The vectors produced by Lánczos are explicitly computed and stored, and the required samples are drawn directly from these intermediate vectors in $O(\dn)$ time, with this cost included in the runtime.  Our guarantees also differ from earlier sublinear SDP methods such as Garber--Hazan~\cite{GarberHazan2016}, whose estimator variance is controlled by $\norm{\A{j}}_F^2$.  Our estimator instead achieves operator-norm control and therefore remains efficient for generic full-rank constraints even when $\norm{\A{j}}_F^2=\Theta(\dn)$.  The speedup thus comes from exploiting the structured Gibbs states generated by the optimization algorithm and reusing each sample across all observables, rather than from stronger access assumptions or low-rank structure.  Accordingly, we classically realize the Quantum OR sample-reuse mechanism in this structured SDP setting, rather than simulating arbitrary gentle measurements on an unknown quantum state.

\subsection{Main Result 2: Quantum-Inspired Classical Sublinear-Time SDP Solvers}
\label{sec:intro-regret-result}

The simultaneous Gibbs expectation estimator provides a classical realization
of the sample-reuse mechanism underlying Quantum OR. We now use it to develop
two distinct quantum-inspired classical SDP solvers.

Our first solver inserts this estimator into the standard MMWU-based
primal-oracle framework, closely paralleling the architecture of the prior
quantum SDP solvers. On each round, the algorithm constructs a Gibbs-state
iterate, estimates its values against all constraints, and identifies a
constraint that is guaranteed to be violated whenever the current iterate is
not approximately feasible. As in the quantum setting, reusing the expensive
Gibbs computation across all constraints decouples the matrix dimension $\dn$
from the number of constraints $\mc$. For fixed accuracy and sparsity, the
resulting dependence on $\dn$ and $\mc$ is additive rather than
multiplicative, yielding a classical solver with sublinear input dependence.
Its main limitation is the accuracy dependence. Each of the
$\widetilde O(\eps^{-2})$ MMWU rounds requires
$\widetilde O(\eps^{-2})$ random starts to obtain sufficiently accurate
constraint estimates before choosing the next update, contributing to an
overall $\widetilde O(\eps^{-6.5})$ dependence, compared with
$\widetilde O(\eps^{-5})$ for the corresponding quantum solver.

Our second algorithm builds on the framework of Carmon et
al.~\cite{CDST19}, which replaces the full matrix response used on each
round by a randomized rank-one approximation whose error need only average
out over time. We push this stochastic viewpoint further, replacing the
accurate per-round constraint selection and payoff computation of the first
solver with randomized counterparts as well. Consequently, individual
rounds need not produce an accurate primal response or a certified violated
constraint. Correctness instead emerges from the aggregate behavior of the
optimization trajectory. This fully stochastic formulation retains additive
dependence on $\dn$ and $\mc$ while improving the normalized accuracy
dependence to $\widetilde O(\eps^{-4.5})$, better than even the
$\widetilde O(\eps^{-5})$ dependence of the best corresponding quantum
solver.

\subsubsection{SDP Feasibility as a Two-Player Game}
We begin by expressing SDP feasibility as a two-player zero-sum game.
This viewpoint provides a common framework for both of our algorithms and
will later allow us to interpret the second algorithm directly through
online-learning regret.

To begin, fix a target objective value $g\in[-\Rx,\Rx]$ and treat the requirement
$\inner{\bC}{\bX}\geq g$ as an additional feasibility constraint.
After introducing one extra coordinate to account for unused trace, each
requirement can be expressed as the nonnegativity of a linear feasibility
margin. Normalizing these margins gives matrices
$\mathbf D_1,\ldots,\mathbf D_{\mc}$ satisfying
$\norm{\mathbf D_j}\leq1$. The reduction preserves sparse-row access up
to constant factors. For simplicity, we continue to write $\dn$ and
$\mc$ for the lifted dimension and number of feasibility requirements.

The matrix player chooses a density matrix
$\brho\in\mathcal S_{\dn}$ representing a candidate primal solution,
while the constraint player chooses a distribution
$\bp\in\Delta_{\mc}$ over the feasibility requirements. Their payoff is
\begin{align}
\mathcal L(\brho,\bp)
&=
\left\langle
\sum_{j=1}^{\mc}p_j\mathbf D_j,\brho
\right\rangle.
\end{align}
The matrix player seeks to maximize this payoff, while the constraint
player seeks to minimize it. Furthermore, for a fixed matrix action $\brho$, minimizing over $\bp$ places all
weight on a least-satisfied feasibility requirement. The value of the game is, therefore, by standard minimax,
\begin{align}
\mathfrak s
&=
\min_{\bp\in\Delta_{\mc}}
\max_{\brho\in\mathcal S_{\dn}}
\mathcal L(\brho,\bp)=
\max_{\brho\in\mathcal S_{\dn}}
\min_{j\in[\mc]}
\inner{\mathbf D_j}{\brho}.
\label{eqn:intro-regret-game}
\end{align}
The target value $g$ is, thus, attainable exactly when
$\mathfrak s\geq0$.
For a pair of strategies $(\brho,\bp)$, the corresponding saddle-point
gap is
\begin{align}
\Gap_{\mathcal D}(\brho,\bp)
&=
\lambda_{\max}\!\left(
\sum_{j=1}^{\mc}p_j\mathbf D_j
\right)
-
\min_{j\in[\mc]}
\inner{\mathbf D_j}{\brho}.
\label{eqn:intro-regret-gap}
\end{align}
The first term is the best payoff available to the matrix player against
$\bp$, while the second is the best payoff available to the constraint
player against $\brho$. These quantities give upper and lower bounds on
the game value, so a small gap certifies that the two strategies are
close to equilibrium in payoff.

Solving these feasibility games to normalized accuracy
$\Theta(1/\gam)$, computing suitable game-value certificates, and
performing a gapped binary search over $g$ yields an additive-$\eps$
solution to the original SDP, where $\gam:=\Rx\Ry/\eps$ and we take
$\Rx,\Ry\geq1$. The dual-radius bound $\Ry$ determines the accuracy
required of the feasibility games rather than the radius of the
constraint player's simplex. The complete reduction and certificate
construction are given in \Cref{sec:standard-sdp-to-game}.

\paragraph{Oracle-Based versus Two-Sided Stochastic Realizations.}
The primal-feasibility routine of the oracle-based MMWU framework treats
the two sides asymmetrically. The matrix player follows a Gibbs update,
while a separation oracle either certifies approximate feasibility or
returns a violated constraint that determines the next update. Prior
quantum SDP solvers accelerate this search using the Fast Quantum OR and
Grover search. Our first classical solver retains this oracle-based
architecture and replaces the quantum constraint search with simultaneous
classical Gibbs expectation estimation.

Carmon, Duchi, Sidford, and Tian~\cite{CDST19} move toward a stochastic
realization by replacing the matrix player's deterministic response with
a randomized rank-one action. However, they still compute the full payoff
vector and update the matrix learner via a dense mixture of
constraints. Our second algorithm makes the interaction stochastic on
\emph{both} sides. It samples a single constraint for the matrix update and uses
a single sampled coordinate of the rank-one matrix action to obtain
stochastic payoff estimates for all constraints. Neither side, therefore,
requires an accurate response on every round. Instead, convergence is
established by controlling the cumulative regret of the two players.

\subsubsection{Oracle-Based MMWU with Simultaneous Gibbs Estimation}
\label{sec:intro-direct-dequantization}

\begin{algorithm}[t]
\small
\caption{Primal Oracle-Based MMWU with Simultaneous Gibbs Estimation}
\label{alg:intro-direct}
\begin{algorithmic}[1]
\Require Normalized constraints $D(\mathbf D_j)_{j=1}^{\mc}$, effective inverse accuracy $\gam$, failure probability $\del$
\State \textbf{Parameters:} $\TT\gets\widetilde O(\gam^2)$ and
$\alpha,\eta\gets\Theta(1/\gam)$
\State \textbf{Initialize} $\bH_1\gets\bzero$
\For{$t=1,\ldots,\TT$}
\State Let $\brho_t:=e^{-\bH_t}/\Tr(e^{-\bH_t})$ denote the Gibbs state
implicitly represented by $\bH_t$
\State Compute $\widehat e_t$ by \Cref{thm:intro-gibbs-moment} with
\(    \Pr\!\left[
        \max_j
        \left|
            \widehat e_{t,j}
            -
            \Tr(\mathbf D_j\brho_t)
        \right|
        \geq\alpha/2
    \right]
    \leq\del/\TT
   \)
\If{$\widehat e_t$ certifies approximate feasibility}
\State \Return $\bH_t$ as a succinct representation of the primal
solution $\brho_t$
\EndIf
\State Choose a certified violated constraint $j_t$ and set
$\bH_{t+1}\gets\bH_t-\eta\mathbf D_{j_t}$
\EndFor
\State $\overline{\bp}\gets\TT^{-1}\sum_{t=1}^{\TT}\be_{j_t}$
\State \Return $\overline{\bp}$ as a dual certificate
\end{algorithmic}
\end{algorithm}

We begin with our first classical SDP-solving algorithm, with pseudocode given in \Cref{alg:intro-direct}. It follows the MMWU-based SDP primal-oracle framework of Brandão et al.~\cite{BKL19} and van Apeldoorn--Gilyén~\cite{AG19}. At each round, the algorithm forms a Gibbs-state candidate and invokes an oracle to search for a violated constraint. If one is found, it updates the Hamiltonian in the direction of that constraint. The quantum and classical algorithms share these outer dynamics. The quantum algorithms implement the search using Fast Quantum OR and Grover-type procedures, whereas our classical solver uses the simultaneous Gibbs expectation estimator from Part~1. We apply this estimator with the normalized constraint matrices $\mathbf D_1,\ldots,\mathbf D_{\mc}$ as the observables. As illustrated in \Cref{fig:mmwu-game}, the resulting trajectory repeatedly updates the Gibbs candidate in response to a constraint returned by the oracle.

At round $t$, the normalized losses returned so far define a Hamiltonian
$\bH_t$, and the matrix action is
$\Gibbs{\bH_t}=e^{-\bH_t}/\Tr(e^{-\bH_t})$.  The oracle either certifies
approximate feasibility or returns a violated constraint, whose loss is
added to $\bH_t$.  MMWU controls the resulting sequence in aggregate, so
either a Gibbs iterate becomes approximately feasible or the average
returned loss yields a dual certificate.  The key asymmetry is that the
constraint side must still complete a sufficiently accurate search before
every update.

The Fast Quantum OR lemma decouples this search from Gibbs-state preparation by
reusing a shared state across many constraint tests, while Grover search
provides the square-root dependence on $\mc$.  For fixed accuracy and
sparsity, the quantum cost becomes
$\widetilde O(\sqrt{\dn}+\sqrt{\mc})$ rather than
$\widetilde O(\sqrt{\dn\mc})$.  Our simultaneous estimator reproduces the
same decoupling classically by estimating all energies
$(\Tr(\mathbf D_j\Gibbs{\bH_t}))_{j=1}^{\mc}$ from one shared sample collection.
The outer MMWU procedure is otherwise unchanged.

Because this solver retains the oracle-based MMWU architecture, the
simultaneous Gibbs estimates must be sufficiently accurate within every
round.  The MMWU procedure runs for
$\widetilde O(\gam^2)$ rounds with per-round estimation accuracy
$\Theta(1/\gam)$, so each round requires
$\widetilde O(\gam^2)$ random starts.  At round $t$, the Hamiltonian has
the form required by \Cref{thm:intro-gibbs-moment} with
$h=O(t)$ sparse terms and coefficient norm
$B=\widetilde O(t/\gam)$.  The corresponding Lánczos degree is
$\widetilde O(\sqrt\gam)$ before the optional $\dn$-step cap.  This gives
the following end-to-end guarantee.

\begin{theorem}[Informal oracle-based classical SDP solver]
\label{thm:intro-direct-solver}
For a sparse normalized SDP instance and a fixed candidate objective value
$g$, a randomized classical sparse-oracle algorithm determines whether there
exists a feasible solution achieving objective value at least $g$, up to the
prescribed additive accuracy, with probability at least $1-\del$.  It returns
either a succinct Gibbs-form primal certificate or a sparse dual certificate.
A standard logarithmic search over $g$ then approximates the SDP optimum to
additive error $\eps$ in time
$\widetilde O(\dn\sr\gam^{6.5}+\mc\sr\gam^4)$.
\end{theorem}

The formal algorithm and runtime analysis appear in
\Cref{sec:direct-gibbs-mmwu}.  The resulting solver has a particularly
simple dimensional interpretation.  The succinct primal representation
avoids the $\Omega(\dn^2)$ cost of materializing a dense matrix.  For fixed
accuracy and sparsity, the quantum solver scales as
$\widetilde O(\sqrt{\dn}+\sqrt{\mc})$, while our classical solver scales as
$\widetilde O(\dn+\mc)$.  It therefore recovers classically the additive
dimension dependence arising from the Quantum OR, without reproducing the
square-root search speedup.  In the bounded-radius, constant-accuracy
regime, this gives sublinear runtime relative to the
$\Theta(\mc\dn\sr)$-size sparse input representation.  In particular, the
algorithm can approximate the optimum without reading the full problem
specification. Its main limitation is the accuracy dependence.  The nested
$\widetilde O(\gam^2)$ MMWU rounds and
$\widetilde O(\gam^2)$ random starts per round contribute to the
$\widetilde O(\eps^{-6.5})$ dependence, compared with
$\widetilde O(\eps^{-5})$ for the corresponding quantum solver.  This
motivates removing the requirement for an accurate separation oracle on
every round.

\begin{figure}[t]
    \centering
    \includegraphics[width=0.5\linewidth]{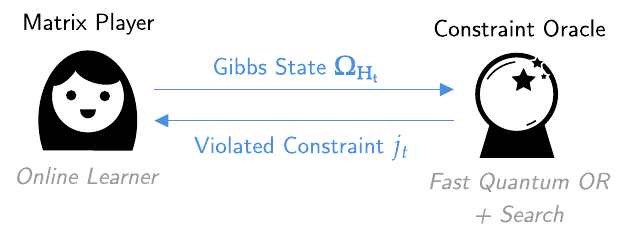}
    \caption{Oracle-based realization of the SDP game.  The matrix player
    evolves through an online Gibbs update, while a separation oracle
    returns a sufficiently violated constraint in each round.  Prior
    quantum SDP solvers implement this oracle using Fast Quantum OR and
    Grover search.}
    \label{fig:mmwu-game}
\end{figure}

\subsubsection{From One-Sided to Two-Sided Stochastic SDP Solving}
\label{sec:intro-two-sided-stochastic}

An online learner need not choose a best response on every round.  Instead,
its performance is measured over the entire sequence by its
\emph{regret}, which compares the cumulative payoff of the actions it
played with that of the best fixed action chosen in hindsight.

In the normalized SDP game derived above, the payoff on round $t$ is
$\mathcal L(\brho_t,\bp_t)
=\left\langle\sum_{j=1}^{\mc}p_{t,j}\mathbf D_j,\brho_t\right\rangle$.
The matrix player chooses $\brho_t\in\mathcal S_{\dn}$ and seeks to
maximize this payoff.  The constraint player chooses
$\bp_t\in\Delta_{\mc}$ and seeks to minimize it by placing weight on
constraints that are poorly satisfied by the matrix actions.  Their
cumulative regrets are
\begin{align}
    \operatorname{Reg}_{\rho}(\TT)
    &=
    \max_{\brho\in\mathcal S_{\dn}}
    \sum_{t=1}^{\TT}\mathcal L(\brho,\bp_t)
    -
    \sum_{t=1}^{\TT}\mathcal L(\brho_t,\bp_t),
    \notag\\
    \operatorname{Reg}_{p}(\TT)
    &=
    \sum_{t=1}^{\TT}\mathcal L(\brho_t,\bp_t)
    -
    \min_{\bp\in\Delta_{\mc}}
    \sum_{t=1}^{\TT}\mathcal L(\brho_t,\bp).
    \label{eqn:intro-player-regrets}
\end{align}
The first quantity measures how much more the matrix player could have
obtained by using the best fixed matrix against the sequence
$\bp_1,\ldots,\bp_{\TT}$.  The second measures how much the constraint
player could have reduced the cumulative payoff by using the best fixed
distribution against $\brho_1,\ldots,\brho_{\TT}$.  Sublinear regret means
that these differences are $o(\TT)$, so the average error per round
vanishes.

This cumulative guarantee translates directly into an approximate solution
of the SDP game.  Define the average actions
$\bar{\brho}=\frac{1}{\TT}\sum_{t=1}^{\TT}\brho_t$ and
$\bar{\bp}=\frac{1}{\TT}\sum_{t=1}^{\TT}\bp_t$.  When the two regrets are added,
the realized interaction payoffs
$\sum_{t=1}^T\mathcal L(\brho_t,\bp_t)$ cancel.  Bilinearity then gives
\begin{align}
    \Gap_{\mathcal C}
    \bigl(\bar{\brho},\bar{\bp}\bigr)
    &=
    \frac{1}{\TT}
    \left[
        \max_{\brho\in\mathcal S_{\dn}}
        \sum_{t=1}^{\TT}\mathcal L(\brho,\bp_t)
        -
        \min_{\bp\in\Delta_{\mc}}
        \sum_{t=1}^{\TT}\mathcal L(\brho_t,\bp)
    \right]=
    \frac{
        \operatorname{Reg}_{\rho}(\TT)
        +
        \operatorname{Reg}_{p}(\TT)
    }{\TT}.
    \label{eqn:intro-regret-to-gap-ideal}
\end{align}
Thus neither player needs to respond accurately on an individual round.
It is enough to control the cumulative regret of each sequence, and the
algorithm returns the averages of the actions played over the full
trajectory.  In particular, regret
$\widetilde O(\sqrt{\TT})$ for each player yields saddle-point gap
$\widetilde O(\TT^{-1/2})$, so $\widetilde O(\gam^2)$ rounds suffice to
reach normalized accuracy $O(1/\gam)$.  The approaches below differ in how
they obtain these cumulative guarantees and in how much exact computation
they require on each round.

\paragraph{One-Sided Stochasticity and \cite{CDST19}.}
Carmon et al.~\cite{CDST19} treat both players in the SDP game as online learners, but introduce stochasticity only on the matrix side.  Given cumulative gain $\bH$, the matrix
player uses the rank-one response
\begin{align}
    \mathcal P_{\bu}(\bH)
    &:={}
    \frac{e^{\bH/2}\bu\bu^\dagger e^{\bH/2}}
         {\langle\bu,e^{\bH}\bu\rangle}.
    \label{eqn:intro-rank-one-response}
\end{align}
A single response can be far from the Gibbs matrix and, in fact, is not even an unbiased
Gibbs estimator.  Instead, their average-projection theorem shows that the mean response $\overline{\mathcal P}(\bH):=\E_{\bu}\!\left[\mathcal P_{\bu}(\bH)\right]$
has small cumulative regret.  On round $t$, after conditioning on the
current Hamiltonian $\bH_t$, the fresh rank-one response
$\mathcal P_{\bu_t}(\bH_t)$ has expectation
$\overline{\mathcal P}(\bH_t)$.  Hence, for the current gain matrix
$\bG_t$,
\begin{align}
    \E\!\left[
        \left\langle
            \bG_t,\,
            \mathcal P_{\bu_t}(\bH_t)
            -
            \overline{\mathcal P}(\bH_t)
        \right\rangle
        \,\middle|\, \bH_t
    \right]
    =0.
\end{align}
The resulting payoff errors are therefore conditionally mean-zero and
concentrate over the full trajectory.  Thus convergence depends on the
cumulative behavior of the sampled responses, rather than on any
individual rank-one response approximating the Gibbs state.

This yields the improved accuracy dependence of Carmon et al., but their
sparse-row implementation still incurs two dense per-round operations.
First, the constraint learner evaluates the full payoff vector
$\bigl(\inner{\mathbf D_j}{\brho_t}\bigr)_{j=1}^{\mc}$ exactly, at cost
$O(\mc\dn\sr)$.  Second, the matrix learner updates its Hamiltonian using the dense mixture
$\sum_j p_{t,j}\mathbf D_j$.  Approximating
$e^{\bH_t/2}\bu$ by Lánczos requires repeated Hamiltonian--vector products,
each costing $O(\min\{\mc\dn\sr,\dn^2\})$ with this dense mixture.  Thus, although the one-sided stochastic
analysis improves the accuracy dependence, the generic sparse-row runtime
still scales multiplicatively in $\mc$ and $\dn$ and is therefore not
sublinear in the natural sparse input representation.

\paragraph{Two-Sided Stochasticity in Our Solver.}
We remove the two dense operations in the Carmon et al. framework through
two separate sampling steps.  First, the dual player maintains a
distribution $\bp_t$ and samples $j_t\sim\bp_t$, so the primal player
updates with the sparse pure action $\mathbf D_{j_t}$ rather than the dense mixture
$\sum_j p_{t,j}\mathbf D_j$.  Second, to update the dual player, we reuse the
shared-coordinate sampling idea underlying our simultaneous Gibbs expectation
estimator from Part~1.  A single coordinate sampled from the primal
player's rank-one action provides simultaneous payoff estimates for all
constraints.  Unlike in the oracle-based solver, these estimates are not
averaged to high accuracy within each round, but are fed directly into the
dual update.  Before clipping, they are conditionally unbiased with
operator-norm-controlled second moment, and the same coordinate can be
reused across all constraints in $O(\mc\sr)$ time.  This two-sided
stochastic interaction is illustrated in \Cref{fig:intro-final-algo}.

\begin{figure*}[t]
    \centering
    \includegraphics[width=0.80\textwidth]{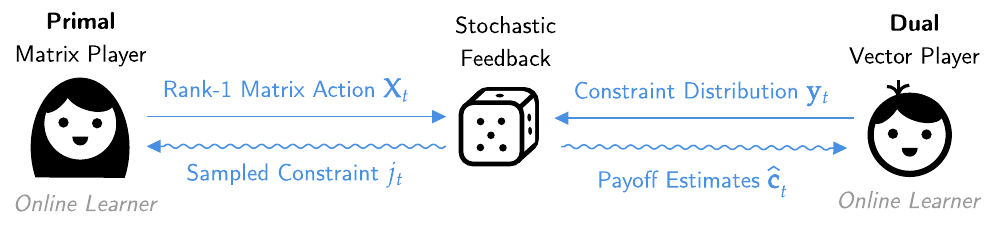}
    \caption{Two-sided stochastic realization of the SDP game.  The primal
    player sends a rank-one matrix action and the dual player sends a
    distribution over constraints to the stochastic feedback mechanism.
    The mechanism samples a constraint for the primal update and uses one
    sampled coordinate of the matrix action to form payoff estimates for
    the dual update.  Both players update online and the returned solution
    averages their actions over the full trajectory.}
    \label{fig:intro-final-algo}
\end{figure*}

Both players, therefore, update directly from sampled information. No dense
constraint mixture enters the Hamiltonian, and no full payoff vector is
computed exactly. Individual rounds may be poor approximations to best
responses, while the averages of the rank-one matrix actions and sampled
constraints converge over the trajectory. On the primal side, this is
precisely the perspective introduced by Carmon et al., who control the
fluctuations of randomized rank-one responses only in aggregate. Our dual analysis follows the same cumulative-regret framework, but stochastic payoff feedback introduces an additional challenge. Although unbiased with bounded second moment, it can still take large values on individual rounds.
Controlling these occasional large values is therefore the main challenge in the dual regret analysis.

We address this by clipping the stochastic payoff feedback before it is passed to
the Hedge update.  For
$\Lam=\Theta(\log(\dn\mc/\del))$, we set
$B=\Theta(\sqrt{\TT/\Lam})$ and define $\operatorname{clip}_{B}(x):=\max\{-B,\min\{x,B\}\}$.
The clipped feedback is uniformly bounded by $B$, while the
operator-norm second-moment bound implies that clipping introduces bias at
most $1/B$.  This choice of $B$ balances the cumulative clipping bias
against the range-dependent concentration term, preserving the
operator-norm advantage without introducing any dimension-dependent range
bound.

The pseudocode is summarized in \Cref{alg:intro-two-sided}, with the full algorithm developed
in \Cref{sec:algorithm}. The convergence proof appears in
\Cref{sec:analysis} and the runtime is analyzed in
\Cref{subsec:main-theorem-proof}.  Hermitian inputs are reduced to the
real-symmetric setting in \Cref{sec:realification}.

\begin{theorem}[Informal two-sided stochastic SDP solver]
\label{thm:intro-regret-solver}
Suppose $\max_j\norm{\mathbf D_j-c\mathbf I}\leq\om$ and set
$\gam=\Rx\Ry\om/\eps\geq1$.  With probability at least $1-\del$, the
algorithm returns normalized primal and dual actions which, after rescaling
by $\Rx$ and $\Ry$, have original-scale saddle-point gap at most $\eps$
after $\widetilde O(\gam^2)$ rounds.  Its sparse-oracle
running time is
\begin{align}
    \widetilde O\!\left(
        \dn\sr\min\!\left\{\gam^2,\mc,\frac{\dn}{\sr}\right\}\gam^{2.5}
        +\mc\sr\gam^2
    \right).
    \label{eqn:intro-regret-runtime}
\end{align}
When $\gam^2\leq\min\{\mc,\dn/\sr\}$, this becomes
$\widetilde O(\dn\sr\gam^{4.5}+\mc\sr\gam^2)$.
The primal output is represented by the average of the sampled rank-one
states and the dual output by the empirical distribution of sampled
constraints.
\end{theorem}

\begin{algorithm}[t]
\small
\caption{Two-Sided Stochastic SDP Solver}
\label{alg:intro-two-sided}
\begin{algorithmic}[1]
\Require Normalized constraints $(\mathbf C_j)_{j=1}^{\mc}$, effective inverse accuracy $\gam$, failure probability $\del$
\State \textbf{Parameters:} $\TT\gets\widetilde O(\gam^2)$, $\Lam\gets\Theta(\log(\dn\mc/\del))$, $\lrX\gets\Theta(\sqrt{\log(\dn)/\TT})$, $B\gets\Theta(\sqrt{\TT/\Lam})$, and $\lrY\gets\min\{1/B,\Theta(\sqrt{\log(\mc)/\TT})\}$
\State \textbf{Initialize} $\Ht{1}\gets\bzero$ and $\widehat{\bell}_1\gets\bzero$
\For{$t=1,\ldots,\TT$}
    \State $\phat{t}\gets\softmax(-\lrY\widehat{\bell}_t)$ and draw $\jt{t}\sim\phat{t}$
    \State Draw a Haar-random $\ut{t}$ and compute
    $\psitilde{t}\approx e^{\lrX\Ht{t}/2}\ut{t}/\norm{e^{\lrX\Ht{t}/2}\ut{t}}_2$ by Lánczos
    \State $\rtilde{t}\gets\psitilde{t}\psitilde{t}^{\dagger}$ and draw $\isamp{t}\sim\abs{\psitilde{t}}^2$
    \State For every $j\in[\mc]$, set
    $[\chat{t}]_j\gets\operatorname{clip}_{B}\!\left(\operatorname{Re} \left((\mathbf C_j\psitilde{t})_{\isamp{t}}/(\psitilde{t})_{\isamp{t}}\right)\right)$
    \State $\Ht{t+1}\gets\Ht{t}+\mathbf C_{\jt{t}}$ and $\widehat{\bell}_{t+1}\gets\widehat{\bell}_t+\chat{t}$
\EndFor
\State \Return $\brho^{\mathrm{out}}\gets\TT^{-1}\sum_{t=1}^{\TT}\rtilde{t}$ and $\bp^{\mathrm{out}}\gets\TT^{-1}\sum_{t=1}^{\TT}\be_{\jt{t}}$
\end{algorithmic}
\end{algorithm}

The analysis follows the adaptive trajectory generated by the algorithm
itself.  On the matrix side, the average-projection theorem controls the
mean rank-one response, while martingale concentration shows that sampling
one response per round introduces only a small cumulative error.  Sampling
a single constraint instead of using the full mixture is handled similarly.
On the constraint side, a variance-sensitive Hedge analysis uses the
operator-norm second-moment bound of the stochastic payoff estimator, with clipping
controlling the occasional large feedback values at the cost of a small
bias.  Finally, the error from finite-step Lánczos is bounded by comparing
each approximate response with the corresponding exact rank-one response
for the same realized Hamiltonian and random start.  Together, these bounds
give $\Gap_{\mathcal C}(\bar{\brho},\bar{\bp})=\widetilde O(1/\sqrt{\TT})$, so $\TT=\widetilde O(\gam^2)$ rounds suffice.

The same constraint sampling that enables the two-sided stochastic analysis
also keeps the Hamiltonian sparse enough for efficient Lánczos updates.
After $t$ rounds, the Hamiltonian contains at most
$\min\{t,\mc\}$ distinct constraint matrices, so applying it to a vector
costs $O(\dn\sr\min\{t,\mc,\dn/\sr\})$
where the final term reflects the $O(\dn^2)$ cost of dense
matrix--vector multiplication.  Since the total Krylov degree across all
rounds is $\widetilde O(\TT^{5/4})$, these Hamiltonian--vector products
determine the leading computational cost.  Each round additionally requires
$O(\dn)$ time to form the rank-one response and $O(\mc\sr)$ time to read
one sparse row from every constraint.  Substituting
$\TT=\widetilde O(\gam^2)$ yields
\Cref{eqn:intro-regret-runtime}.  For fixed $\gam$ and $\sr$, the resulting
runtime is $\widetilde O(\dn+\mc)$ even for full-rank constraints with
$\norm{\mathbf C_j}_F^2=\Theta(\dn)$.
\subsection{Discussion}
\label{sec:intro-discussion}

Our results isolate the component of the quantum SDP speedup that can be reproduced classically. Namely, we show that the sample-reuse mechanism underlying the Quantum OR lemma admits an efficient classical realization for the structured Gibbs states and sparse observables arising in SDP solving. Importantly, this does not constitute a general dequantization of the Quantum OR lemma for arbitrary states and measurements. Our algorithms instead rely on problem-specific classical access, using sparse Hamiltonian representations for randomized Lánczos and sparse-row queries to the constraint matrices for the energy estimator, which do not apply to an unknown quantum state and black-box measurements.

The remaining dimensional separation is the usual Grover advantage.
Our algorithms explicitly construct an $\dn$-dimensional vector and query
one sparse row of each constraint, leading to linear rather than
square-root dependence on $\dn$ and $\mc$.  This gap is fundamental in
the standard sparse-oracle model, where known lower bounds require
$\Omega(\dn+\mc)$ classical queries in the worst case, compared with
$\Omega(\sqrt{\dn}+\sqrt{\mc})$ quantum queries~\cite{BKL19}.  Thus, for
fixed accuracy and sparsity, our $\widetilde O(\dn+\mc)$ dependence is
essentially optimal classically.  After dequantizing the Quantum OR,
the dimensional advantage retained by the quantum algorithms is therefore
precisely the standard quadratic Grover speedup which manifests across many other applications. Given the widespread belief that superquadratic speedups are necessary to realize practical advantage for the foreseeable future \cite{Babbush2021}, our findings cast doubt on the hope that SDP solving will be a meaningful practical application of quantum computers.

It may seem surprising that SDPs can be solved without reading the entire input. After all, the optimum can in principle depend on every constraint and every matrix entry. The key is that an optimization algorithm need not determine the contribution of every piece of the input individually. Instead, iterative methods can use a small amount of adaptively chosen or randomly sampled information in each round to determine how to update the current solution. In the multiplicative-weights approach, this takes the particularly intuitive form of searching for a violated constraint and using it to guide the next update. Our regret-based approach instead relies on stochastic feedback, but similarly avoids ever constructing or inspecting the input in its entirety. In both cases, the algorithm touches only a small fraction of the input while still accumulating enough information over the course of the optimization to approximate the optimum.

The sublinear runtimes should also be interpreted together with our
implicit-output convention.  The oracle-based solver returns a succinct
Gibbs-form primal certificate, while the two-sided stochastic solver
returns an average of at most $\widetilde O(\gam^2)$ rank-one factors and
an empirical distribution over the sampled constraints.  Explicitly
materializing a generic dense primal matrix would already require
$\Omega(\dn^2)$ time.  The significance of the result is instead that the
SDP value and a succinct approximate solution can be obtained without
reading the full $\Theta(\mc\dn\sr)$-size sparse input.

Our results also suggest revisiting the optimization architecture of quantum SDP
solvers.  Existing quantum algorithms largely accelerate the oracle-based
Arora--Kale MMWU framework, retaining the requirement that every round
produce a sufficiently accurate constraint response.  By contrast, our
two-sided classical solver obtains its improved accuracy dependence by
allowing both players to operate on stochastic feedback and controlling
the resulting errors only over the full optimization trajectory. A closely related quantum precedent is provided by Bouland, Getachew, Jin,
Sidford, and Tian~\cite{bouland2023quantum} in the setting of (zero-sum game formulations of) linear programs (LPs).  Their approach combines stochastic optimization
with dynamic quantum Gibbs sampling for slowly changing distributions.
It is therefore natural to ask whether an analogous architecture can be
developed for SDPs.  In particular, can a quantum two-player SDP solver
retain the $\widetilde O(\sqrt{\dn}+\sqrt{\mc})$ dimension dependence of
the existing quantum algorithms while improving their dependence on
$\gam$ through a trajectory-level stochastic analysis?  Such an algorithm
would require matrix-valued analogues of the dynamic sampling and
cumulative-feedback ideas used for LPs, together with a way to control
matrix-response and constraint-feedback errors over the entire trajectory.
Our classical analysis suggests that online-learning and martingale
techniques may provide a useful starting point for this direction.

\paragraph{Remaining Paper Organization.}
\Cref{sec:problem} introduces the sparse-access model, the zero-sum game
framework, and the reduction from SDP optimization to approximate
saddle-point computation.  \Cref{sec:energy-estimator} develops the
coordinate-sampling estimator for sparse observables and
\Cref{sec:lanczos} establishes the randomized Lánczos guarantees used to
approximate matrix-exponential--vector products.  Building on these
ingredients, \Cref{sec:gibbs-or} develops the simultaneous Gibbs
expectation estimator (dequantization of the Fast Quantum OR).  \Cref{sec:direct-solver} uses this estimator to
dequantize Gibbs-state constraint search and obtain the direct
oracle-based Gibbs--MMWU solver.  \Cref{sec:algorithm} then develops our
two-sided stochastic SDP solver and gives its high-probability convergence
and sparse-oracle runtime analysis.  Finally, \Cref{sec:realification}
extends the real-symmetric analysis of the stochastic solver to Hermitian
SDPs.

\section*{Acknowledgments and AI Disclosure}
\label{sec:acknowledgments}

We thank Oskar Painter, James Hamilton, Nafea Bshara, Peter DeSantis, and Andy Jassy for their involvement and support of the research activities at the AWS Center for Quantum Computing. A.G. acknowledges funding from the AWS Center for Quantum Computing. F.V. conducted this work as a summer intern at the AWS Center for Quantum Computing.

The authors disclose the use of AI tools in the development and preparation
of this manuscript. The observation of an apparent quartic quantum speed-up
for SDP solving, together with the main ideas underlying the dequantization
of the Quantum OR lemma in the SDP setting, originated with the authors
in 2023 and 2024, without the use of AI tools. The authors formally wrote out
these ideas further in the summer of 2026, with assistance (but no major contributions) from GPT-5.6 Sol.
The subsequent stochastic-solver approach emerged through interactions
between the authors and GPT-5.6 Sol. It was then developed, formalized, and
verified by the authors with further assistance from GPT-5.6 Sol, largely in working out and verifying the concentration guarantees. GPT-5.6 Sol and GPT-6 Astra were also used as writing and editing tools, with all
AI-assisted text subject to substantial revision and modification by the authors.

\section{Preliminaries}
\label{sec:preliminaries}
\label{sec:problem}

This section fixes the matrix and game notation used throughout the paper,
formalizes the sparse-row oracle model, and recalls the standard reduction
from SDP optimization to a bilinear zero-sum game over a spectrahedron
and a simplex.  

\subsection{Matrix Settings and Notation}

All vectors and matrices are denoted in boldface, and vectors are column
vectors.  The symbol $^\dagger$ denotes conjugate transpose, while
$^\top$ denotes transpose for real-valued objects.  The Hilbert--Schmidt
inner product is
\begin{align}
\inner{\bA}{\bB}:=\Tr(\bA^\dagger\bB).
\end{align}
We write $\norm{\bA}$, $\norm{\bA}_F$, and
$\norm{\bA}_1:=\Tr\sqrt{\bA^\dagger\bA}$ for the operator, Frobenius,
and Schatten trace norms, respectively.  For a Hermitian matrix $\bA$,
we denote its largest and smallest eigenvalues by
$\lambda_{\max}(\bA)$ and $\lambda_{\min}(\bA)$, respectively.  The
normalized spectrahedron is denoted
\begin{align}
\cS_{\dn}
&:=
\{\bX\succeq0:\Tr(\bX)=1\},
\end{align}
and the probability simplex is denoted
\begin{align}
\Delta_{\mc}
&:=
\{\bp\in\R_{\geq0}^{\mc}:\norm{\bp}_1=1\}.
\end{align}
Finally, $\widetilde O(f)$ denotes $O(f)$ with polylogarithmic factors
suppressed.

The simultaneous Gibbs energy estimator is most naturally stated for Hermitian
matrices and complex Haar-random vectors.  The two-sided stochastic solver
is analyzed for real symmetric matrices and real spherical directions,
matching the average-projection theorem of Carmon, Duchi, Sidford, and
Tian~\cite{CDST19}.  The realification argument in
\Cref{sec:realification} transfers the latter result to Hermitian inputs
with only constant-factor overhead while preserving the game value and
saddle-point gap.

\subsection{Sparse-row oracle and arithmetic model}

The input matrices are accessed through the following classical analogue
of the sparse-matrix oracle used by quantum SDP solvers.  The model charges
only for entries and rows actually queried, rather than
an up-front cost for scanning the full sparse representation.

\begin{definition}[Sparse Matrix Oracle Access]
\label{def:sparse-access}
Let $\bM\in\C^{\dn\times\dn}$ be an $\sr$-row-sparse matrix. We take $1\leq\sr\leq\dn$, and use the same convention for all row-sparsity bounds. Sparse
oracle access to $\bM$ consists of the following two oracles:
\begin{enumerate}[leftmargin=2.4em,label=\textup{(\roman*)}]
    \item $\Osparse(\bM,i,k)$ returns the column index of the $k$-th
    nonzero entry in row $i$ of $\bM$, or a distinguished null symbol if
    row $i$ contains fewer than $k$ nonzero entries.
    \item $\mathcal O_{\bM}(i,\ell)$ returns the entry $[\bM]_{i,\ell}$.
\end{enumerate}
\end{definition}

A saturated instance contains $\Theta(\mc\dn\sr)$ scalar nonzeros, but an
oracle algorithm need not read all of them. We also assume $O(1)$-time random access to explicitly stored scalar input parameters. Throughout the paper, we work in an idealized real-arithmetic model and do
not consider finite-precision or bit-complexity issues.  

\subsection{Squared-Magnitude Sampling from Explicit Vectors}
\label{sec:squared-magnitude-sampling}

Throughout this work, we will occasionally need to sample a coordinate
of an explicitly constructed vector according to its squared magnitude.
For a nonzero vector $\bx\in\C^d$, define for each $i\in[d]$,
\begin{align}
p_{\bx}(i)
&:=
\frac{\abs{[\bx]_i}^2}{\norm{\bx}_2^2}.
\label{eqn:squared-magnitude-distribution}
\end{align}
Such a sample can be generated directly in $O(d)$ time.  Let
$W:=\norm{\bx}_2^2$, draw $r$ uniformly from $[0,W)$, and scan the
coordinates of $\bx$ while maintaining the cumulative weight
\begin{align}
C_t:=\sum_{i=1}^t \abs{[\bx]_i}^2.
\end{align}
We output the first index $t$ for which $C_t>r$.  Then
\begin{align}
\Prb[t=i]
&=
\frac{C_i-C_{i-1}}{W}
=
\frac{\abs{[\bx]_i}^2}{\norm{\bx}_2^2}
=
p_{\bx}(i),
\end{align}
so this procedure samples exactly from the desired distribution.

All vectors to which we apply this procedure are constructed and stored
explicitly by our algorithms.  Thus their coordinates can be queried in
$O(1)$ time, while a squared-magnitude sample costs $O(d)$ time and
requires only $O(1)$ additional space.  In particular, we do not assume
sample-and-query access to any input vector or constraint matrix.
\subsection{Properties of Haar-Random Vectors}
\label{sec:haar-random-prelims}

We will use random unit vectors drawn uniformly from the complex unit
sphere.  Formally, we write
$\bu\sim\operatorname{Haar}(\C^{\dn})$ for a Haar-random unit vector in
$\C^{\dn}$.  The Haar distribution is invariant under unitary
transformations, so no direction is preferred.  The following two moment
identities capture the symmetry properties that we will need.

\begin{fact}[Haar first moment]
\label{fact:haar-first-moment}
If $\bu\sim\operatorname{Haar}(\C^{\dn})$, then $\E_{\bu}[\bu\bu^\dagger]=\bI/\dn$.
\end{fact}
\noindent Thus, on average, a Haar-random rank-one projector is proportional to the
identity.  See, for example, \cite[Section~3]{Mele2024}.

\begin{fact}[Haar quadratic-form second moment]
\label{fact:haar-second-moment}
Let $\bu\sim\operatorname{Haar}(\C^{\dn})$.  For fixed matrices
$\bM,\bN\in\C^{\dn\times\dn}$,
\begin{align}
    \E_{\bu}\!\left[
        (\bu^\dagger\bM\bu)
        (\bu^\dagger\bN\bu)
    \right]
    &=
    \frac{
        \Tr(\bM)\Tr(\bN)+\Tr(\bM\bN)
    }{
        \dn(\dn+1)
    }.
    \label{eqn:haar-quadratic-form-second-moment}
\end{align}
\end{fact}
\noindent This identity controls the second moments of quadratic forms evaluated on
a Haar-random direction and will be used to bound the variance of our
random estimators.  See, for example, \cite[Section~3]{Mele2024}.

\subsection{The Zero-Sum Game Framework}

Our algorithms operate on a bilinear zero-sum game between a matrix
player and a constraint player.  The matrix player chooses
$\bX$ such that $\bX/\Rx \in\cS_{\dn}$ and seeks to maximize the payoff, while the
constraint player chooses $\by$ such that $\by/\Ry\in\Delta_{\mc}$ and seeks to
minimize it.  The radii $\Rx$ and $\Ry$ are kept explicit because they
determine the payoff scale, learning rates, and effective inverse
accuracy.

Let $\A{1},\ldots,\A{\mc}\in\R^{\dn\times \dn}$ be real symmetric
matrices.  We assume that
every $\A{j}$ is $\sr$-row-sparse, accessed according to
\Cref{def:sparse-access}, and that a common operator-norm bound
$\om$ is supplied such that
\begin{align}
    \max_{j\in[\mc]}
    \norm{\A{j}}
    &\leq
    \om.
    \label{eqn:omega}
\end{align}
Consequently, $\A{j}$ has spectral width
$\lambda_{\max}(\A{j})-\lambda_{\min}(\A{j})\leq 2\om$. For radii $\Rx,\Ry>0$, we will define the radius-$\Rx$ spectrahedron as
\begin{align}
    \Rx\cS_{\dn}
    &:={}
    \left\{\bX\succeq0:\Tr(\bX)=\Rx\right\}
    \label{eqn:primal-domain}
\end{align}
and the radius-$\Ry$ scaled probability simplex as
\begin{align}
    \Ry\Delta_{\mc}
    &:={}
    \left\{\by\in\R_{\geq0}^{\mc}:\norm{\by}_1=\Ry\right\}.
    \label{eqn:dual-domain}
\end{align}
The extreme points of $\Rx\cS_{\dn}$ are the rank-one matrices
$\Rx\bpsi\bpsi^\dagger$, where $\bpsi\in\C^{\dn}$ is a unit vector,
while the extreme points of $\Ry\Delta_{\mc}$ are the pure constraint
actions $\Ry\be_j$, where $\be_j\in\R^{\mc}$ denotes the $j$th standard
basis vector.  We refer to these extreme points as the \emph{pure
strategies} of the two players.  More generally, every
$\bX\in\Rx\cS_{\dn}$ can be written as a convex combination of rank-one
pure strategies via its spectral decomposition, and every
$\by\in\Ry\Delta_{\mc}$ can be written as a convex combination of the
pure constraint actions $\Ry\be_j$.  We therefore view arbitrary
$\bX$ and $\by$ as \emph{mixed strategies}.  Our ultimate stochastic solver will operate by sampling pure strategies whose averages realize these mixed actions.

Having specified the players' strategy spaces, we next define how their
actions interact.  A mixed strategy $\by$ of the constraint player induces
the corresponding weighted combination of the constraint matrices through the
adjoint map
\begin{align}
    \cA^\ast\by
    &:={}
    \sum_{j=1}^{\mc}y_j\A{j}.
    \label{eqn:adjoint-map}
\end{align}
Against a matrix strategy $\bX$, the resulting payoff is
\begin{align}
    \mathcal{L}(\bX,\by)
    &:={}
    \inner{\cA^\ast\by}{\bX} = \sum_{j=1}^{\mc}y_j\inner{\A{j}}{\bX},
    \label{eqn:zero-sum-payoff}
\end{align}
i.e. a weighted aggregate of the
constraint values $\inner{\A{j}}{\bX}$ selected by $\by$.
The matrix player seeks to maximize $\mathcal{L}$, while the constraint
player seeks to minimize it.  The value of the resulting zero-sum game is
\begin{align}
    \mathfrak{s}
    &:={}
    \min_{\by\in \Ry\Delta_{\mc}}
    \max_{\bX\in \Rx\cS_{\dn}}
    \mathcal{L}(\bX,\by)
    =
    \max_{\bX\in \Rx\cS_{\dn}}
    \min_{\by\in \Ry\Delta_{\mc}}
    \mathcal{L}(\bX,\by),
    \label{eqn:saddle-point}
\end{align}
where the equality follows from the finite-dimensional bilinear case of
Sion's minimax theorem~\cite[Theorem~3.4]{Sion1958}.

To measure how far a candidate pair $(\bX,\by)$ is from equilibrium, we
consider the improvement available to either player through a unilateral
best response.  This defines the game's so-called saddle-point gap
\begin{align}
    \Gap(\bX,\by)
    &:={}
    \max_{\bX'\in\Rx\cS_{\dn}}
    \mathcal{L}(\bX',\by)
    -
    \min_{\by'\in\Ry\Delta_{\mc}}
    \mathcal{L}(\bX,\by')
    \label{eqn:gap}
\end{align}
The matrix player's best response to $\by$ concentrates the trace-$\Rx$
matrix on a top eigenvector of $\cA^\ast\by$, while the constraint
player's best response to $\bX$ places all mass on a minimum payoff
constraint.  By the Rayleigh--Ritz variational
principle~\cite{Bhatia1997}, the gap can more explicitly be expressed as
\begin{align}
    \Gap(\bX,\by)
    &=
    \Rx\lambda_{\max}(\cA^\ast\by)
    -
    \Ry\min_{j\in[\mc]}\inner{\A{j}}{\bX}.
\end{align}
In particular, the gap is nonnegative, and the two
best-response values bound the possible game value, as
\begin{align}
    \min_{\by'\in\Ry\Delta_{\mc}}
    \mathcal{L}(\bX,\by')
    &\leq
    \mathfrak{s}
    \leq
    \max_{\bX'\in\Rx\cS_{\dn}}
    \mathcal{L}(\bX',\by).
    \label{eqn:gap-brackets-value}
\end{align}
Hence $\Gap(\bX,\by)\leq\eps$ certifies that $(\bX,\by)$ is an
additive-$\eps$ approximate saddle point.

Every payoff in this game has magnitude at most $\Rx\Ry\om$, since
\begin{align}
    \abs{
        \inner{\cA^\ast\by}{\bX}
    }
    &\leq
    \Rx\norm{\cA^\ast\by}
    \leq
    \Rx\Ry\om.
    \label{eqn:payoff-scale}
\end{align}
For target accuracy $\eps$ and failure probability $\del$, it is convenient
to define the effective inverse-accuracy parameter
\begin{align}
    \gam
    &:=
    \frac{\Rx\Ry\om}{\eps}
    \label{eqn:complexity-parameters}
\end{align}
which measures the desired additive error $\eps$ relative to the natural payoff scale
$\Rx\Ry\om$ of the game.   We work in the nontrivial accuracy regime where $0<\eps\leq \Rx\Ry\om$, such that $\gam\geq1$.

\subsection{Mapping SDPs to Zero-Sum Games}
\label{sec:standard-sdp-to-game}

We now connect the zero-sum game formulation to the standard primal--dual
SDP in \Cref{eqn:intro-standard-sdp-primal,eqn:intro-standard-sdp-dual}.  The reduction is standard in
matrix-multiplicative-weights approaches to SDP solving
\cite{AroraKale2016,AG19}, but we make the dependence on the primal and
dual radii explicit.  The basic idea is to fix a candidate objective value
$g$, encode the objective requirement together with the SDP constraints as
a feasibility game, and then recover an approximately optimal primal
solution by a gapped binary search over $g$.

A point that will be important below is that the original SDP dual variable
is not embedded directly into this game.  The constraint player chooses a
probability distribution and therefore has radius one.  The original
dual-radius bound $\Ry$ instead determines how accurately the game must be
solved in order to distinguish attainable objective values from values
above $\operatorname{OPT}$.

\subsubsection{Embedding the Trace-Bounded Primal Domain}
\label{subsec:trace-bounded-embedding}

Assume that there exists an optimal primal solution $\bX^\star$ with
$\Tr(\bX^\star)\leq\Rx$ and an optimal dual solution $\by^\star$ with
$\norm{\by^\star}_1\leq\Ry$.  We may assume $\Rx,\Ry\geq1$, enlarging the
supplied upper bounds if necessary.  Under the normalization of
\Cref{eqn:intro-standard-sdp-primal,eqn:intro-standard-sdp-dual}, we have $\norm{\bC}\leq1$ and
$\norm{\A{j}}\leq1$ for every $j\in[\mc]$.  Hence every
$\bX\succeq\bzero$ with $\Tr(\bX)\leq\Rx$ satisfies
\begin{align}
    \abs{\inner{\A{j}}{\bX}}
    &\leq
    \Rx.
\end{align}
A threshold $b_j<-\Rx$ is incompatible with the promised
trace-bounded feasible domain.  For $b_j>\Rx$, replace $b_j$ by $\Rx$
without deleting the constraint.  This leaves the trace-bounded primal
feasible set unchanged.  Moreover, complementary slackness for the
promised optimal primal--dual pair gives $y_j^\star=0$ whenever the
original $b_j>\Rx$, so the same dual optimum and its radius bound are
preserved.  This cap can be applied on each threshold query.  We may
consequently assume that, for every $j\in[\mc]$,
\begin{align}
    \abs{b_j}
    &\leq
    \Rx
    \label{eqn:threshold-range}
\end{align}

Our game solver uses a fixed-trace spectrahedron, whereas the SDP permits
$\Tr(\bX)<\Rx$.  We account for the unused trace by adding one coordinate.
For every $\bX\succeq\bzero$ with $\Tr(\bX)\leq\Rx$, define
\begin{align}
    \bX^\uparrow
    &:=
    \begin{pmatrix}
        \bX & \bzero \\
        \bzero^\dagger & \Rx-\Tr(\bX)
    \end{pmatrix}
    \in
    \Rx\cS_{\dn+1}.
    \label{eqn:standard-sdp-primal-lift}
\end{align}
Conversely, if $\bY\in\Rx\cS_{\dn+1}$ and $\bX$ is its leading
$\dn\times\dn$ principal block, then
$\bX\succeq\bzero$ and $\Tr(\bX)\leq\Rx$.  Thus the additional coordinate
simply records unused trace.

The two-sided stochastic solver is stated for real symmetric matrices.
For Hermitian SDP inputs, we first apply the gap-preserving realification
from \Cref{sec:realification}, which changes the dimension and row sparsity
by at most constant factors while preserving the relevant trace radii,
operator norms, and saddle-point guarantees.  To avoid additional
notation, we continue to write $\dn$ for the dimension after this
preprocessing.

\subsubsection{The Target-Value Feasibility Game}
\label{subsec:target-value-feasibility-game}

Fix a candidate objective value $g\in[-\Rx,\Rx]$.  We wish to determine
whether there exists $\bX\succeq\bzero$ with $\Tr(\bX)\leq\Rx$ satisfying
\begin{align}
    \inner{\bC}{\bX}
    &\geq
    g \qquad \text{and}, \qquad \forall~j \in [m], \quad
    \inner{\A{j}}{\bX}
    \leq
    b_j
    \label{eqn:target-feasibility-requirements}
\end{align}
We encode the objective requirement and the $\mc$ SDP constraints by the
payoff matrices
\begin{align}
    \bD_0(g)
    &:=
    \begin{pmatrix}
        \bC & \bzero \\
        \bzero & 0
    \end{pmatrix}
    -
    \frac{g}{\Rx}\bI_{\dn+1},
    \notag\\
    \bD_j(g)
    &:=
    \frac{b_j}{\Rx}\bI_{\dn+1}
    -
    \begin{pmatrix}
        \A{j} & \bzero \\
        \bzero & 0
    \end{pmatrix},
    \qquad
    \forall~j\in[\mc].
    \label{eqn:target-feasibility-matrices}
\end{align}
If $\bY\in\Rx\cS_{\dn+1}$ has leading principal block $\bX$, then
\begin{align}
    \inner{\bD_0(g)}{\bY}
    &=
    \inner{\bC}{\bX}-g\\
    \inner{\bD_j(g)}{\bY}
    &=
    b_j-\inner{\A{j}}{\bX}.
    \label{eqn:standard-sdp-feasibility-margins}
\end{align}
Thus the payoffs are precisely the feasibility margins of the objective
and constraint requirements.

Let the matrix player choose
$\bY\in\Rx\cS_{\dn+1}$ and the constraint player choose
$\bp\in\Delta_{\mc+1}$.  The resulting game value is
\begin{align}
    \mathfrak{s}_g
    &:=
    \min_{\bp\in\Delta_{\mc+1}}
    \max_{\bY\in\Rx\cS_{\dn+1}}
    \inner{
        \sum_{j=0}^{\mc}p_j\bD_j(g)
    }{
        \bY
    }=
    \max_{\bY\in\Rx\cS_{\dn+1}}
    \min_{0\leq j\leq\mc}
    \inner{\bD_j(g)}{\bY},
    \label{eqn:standard-sdp-target-game}
\end{align}
where the equality follows from the finite-dimensional bilinear minimax
theorem~\cite{Sion1958}.  Consequently,
\begin{align}
    \mathfrak{s}_g\geq0
    \quad\Longleftrightarrow\quad
    g
    \text{ is attainable}.
    \label{eqn:target-game-sign}
\end{align}

The reduction preserves sparse access up to constant factors.  By
\Cref{eqn:threshold-range}, every $\bD_j(g)$ is
$(\sr+1)$-row-sparse and
\begin{align}
    \norm{\bD_0(g)}
    &\leq
    \norm{\bC}+\frac{\abs g}{\Rx}
    \leq
    2\\
    \norm{\bD_j(g)}
    &\leq
    \norm{\A{j}}+\frac{\abs{b_j}}{\Rx}
    \leq
    2, \qquad \forall~j\in [\mc].
    \label{eqn:target-game-width}
\end{align}
A row of a payoff matrix is generated in $O(\sr)$ time by
reading the corresponding original row and merging the diagonal term.
All accesses to these derived matrices below are complete row scans or
matrix--vector products, so no constant-time oracle for the $k$-th
nonzero after cancellation is needed.
Thus every target-value game has a uniform operator-norm scale, independent
of the particular threshold $g$.  We use the supplied game-width upper bound
\begin{align}
    \om_D
    &:=
    2\geq\max_{0\leq j\leq\mc}\norm{\bD_j(g)}.
    \label{eqn:target-game-width-parameter}
\end{align}
This is the parameter that plays the role of $\om$ when we apply the
generic game solver and the Lánczos upper-certificate routine to the
target-value game.

\subsubsection{From Approximate Saddle Points to SDP Solutions}
\label{subsec:game-to-sdp-solution}

We now show how an approximate saddle point of the target-value game
yields an approximate solution of the original SDP.  The main point is
that the two returned strategies provide lower and upper bounds on the
game value.  These bounds let us distinguish attainable objective values
from values sufficiently far above $\operatorname{OPT}$, which is enough
to perform a gapped binary search over the objective value.

Suppose that the game solver applied to
\Cref{eqn:standard-sdp-target-game} returns
$(\bY^{\mathrm{out}},\bp^{\mathrm{out}})$.  Define
\begin{align}
    \ell_g
    &:=
    \min_{0\leq j\leq\mc}
    \inner{\bD_j(g)}{\bY^{\mathrm{out}}},
    \label{eqn:target-game-lower-certificate}
    \\
    u_g
    &:=
    \Rx
    \lambda_{\max}\left(
        \sum_{j=0}^{\mc}
        p_j^{\mathrm{out}}\bD_j(g)
    \right).
    \label{eqn:target-game-upper-certificate}
\end{align}
The quantity $\ell_g$ is the smallest feasibility margin achieved by the
returned matrix strategy.  The quantity $u_g$ is the largest payoff that
any matrix strategy could obtain against the returned constraint
distribution.  Thus $\ell_g$ and $u_g$ provide, respectively, lower and
upper certificates for the value of the target-value game.

\begin{proposition}[Optimization-to-Feasibility Reduction]
\label{prop:target-value-reduction}
Suppose that
$(\bY^{\mathrm{out}},\bp^{\mathrm{out}})
\in\Rx\cS_{\dn+1}\times\Delta_{\mc+1}$
has saddle-point gap at most $\xi$.  Then
\begin{align}
    \ell_g
    &\leq
    \mathfrak{s}_g
    \leq
    u_g \qquad \text{and} \qquad
    u_g-\ell_g
    \leq
    \xi.
    \label{eqn:target-game-value-bracket}
\end{align}
Moreover:
\begin{enumerate}
    \item \textbf{Attainable thresholds.}
    If $g$ is attainable, then
    $\mathfrak{s}_g\geq0$ and $u_g\geq0$.

    \item \textbf{Recovering a primal solution.}
    If $u_g\geq-\zeta$ for some $\zeta\geq0$, then the leading
    $\dn\times\dn$ principal block $\bX^{\mathrm{out}}$ of
    $\bY^{\mathrm{out}}$ satisfies
    \begin{align}
        \bX^{\mathrm{out}}
        &\succeq
        \bzero,
        \\
        \Tr(\bX^{\mathrm{out}})
        &\leq
        \Rx\\
        \inner{\bC}{\bX^{\mathrm{out}}}
        &\geq
        g-\xi-\zeta \\
        \inner{\A{j}}{\bX^{\mathrm{out}}}
        &\leq
        b_j+\xi+\zeta,
        \qquad
        \forall~j\in[\mc].
        \label{eqn:target-game-approximate-primal}
    \end{align}

    \item \textbf{Thresholds above the optimum.}
    For every $\tau\geq0$, if $g \geq \operatorname{OPT}+(1+\Ry)\tau$, then
    $\mathfrak{s}_g\leq-\tau$.
\end{enumerate}
\end{proposition}

\begin{proof}
By the minimax characterization of
\Cref{eqn:standard-sdp-target-game},
\begin{align}
    \ell_g
    &\leq
    \mathfrak{s}_g
    \leq
    u_g,
\end{align}
where the upper bound follows from Rayleigh--Ritz.  Moreover,
$u_g-\ell_g$ is exactly the saddle-point gap of the returned pair, so
$u_g-\ell_g\leq\xi$. If $g$ is attainable, some lifted primal strategy has every feasibility
margin nonnegative, and hence $\mathfrak{s}_g\geq0$.  The bound
$\mathfrak{s}_g\leq u_g$ then implies $u_g\geq0$. 

Next, suppose $u_g\geq-\zeta$.  Then,
\begin{align}
    \ell_g
    &\geq
    u_g-\xi
    \geq
    -(\xi+\zeta).
\end{align}
Thus every feasibility margin of $\bY^{\mathrm{out}}$ is at least
$-(\xi+\zeta)$.  Applying
\Cref{eqn:standard-sdp-feasibility-margins} to its leading principal
block gives \Cref{eqn:target-game-approximate-primal}.

Finally, let $\by^\star$ be an optimal dual solution and consider the
constraint-player distribution
\begin{align}
    p_0^\star
    &:=
    \frac{1}{1+\norm{\by^\star}_1}\\
    p_j^\star
    &:=
    \frac{y_j^\star}{1+\norm{\by^\star}_1},
    \qquad
    \forall~j\in[\mc].
    \label{eqn:target-game-dual-distribution}
\end{align}
For any matrix strategy $\bY$, with leading principal block $\bX$,
strong duality and dual feasibility give
\begin{align}
    \sum_{j=0}^{\mc}
    p_j^\star
    \inner{\bD_j(g)}{\bY}
    &=
    \frac{
        \mathbf b^\top\by^\star-g
        +
        \inner{
            \bC-\sum_{j=1}^{\mc}y_j^\star\A{j}
        }{
            \bX
        }
    }{
        1+\norm{\by^\star}_1
    }
    \leq
    \frac{
        \operatorname{OPT}-g
    }{
        1+\norm{\by^\star}_1
    }.
    \label{eqn:target-game-dual-payoff}
\end{align}
The minimum payoff is no larger than this convex combination, so
\begin{align}
    \mathfrak{s}_g
    &\leq
    \frac{
        \operatorname{OPT}-g
    }{
        1+\norm{\by^\star}_1
    }.
    \label{eqn:target-game-value-upper-bound}
\end{align}
If
$g\geq\operatorname{OPT}+(1+\Ry)\tau$, then
$\norm{\by^\star}_1\leq\Ry$ implies
$\mathfrak{s}_g\leq-\tau$.
\end{proof}

The upper certificate $u_g$ is defined through an exact best response of
the matrix player, but the target-value search only needs to determine its
sign up to additive tolerance.  It therefore suffices to compute a
one-sided approximation $\widetilde u_g$ satisfying
$u_g-\zeta\leq \widetilde u_g\leq u_g$.  Such an estimate is obtained by
a standard randomized Lánczos approximation of the largest eigenvalue of
the returned matrix mixture~\cite{KuczynskiWozniakowski1992}.

\begin{lemma}[Computable Upper Certificate]
\label{lem:approximate-upper-certificate}
Let
\begin{align}
    \bM_g
    &:=
    \sum_{j=0}^{\mc}
    p_j^{\mathrm{out}}\bD_j(g),
\end{align}
so that $u_g=\Rx\lambda_{\max}(\bM_g)$ and
$\norm{\bM_g}\leq\om_D$.  For every $\zeta>0$ and
$\delta_{\mathrm{eig}}\in(0,1)$, randomized Lánczos returns a value
$\widetilde u_g$ such that, with probability at least
$1-\delta_{\mathrm{eig}}$,
\begin{align}
    u_g-\zeta
    &\leq
    \widetilde u_g
    \leq
    u_g.
    \label{eqn:approximate-upper-certificate}
\end{align}
The computation uses
\begin{align}
    \widetilde O\left(\min\left\{
            \dn+1,
            \sqrt{\Rx\om_D/\zeta}
        \right\}
    \right)
    \label{eqn:upper-certificate-lanczos-degree}
\end{align}
matrix--vector products with $\bM_g$.
\end{lemma}

\begin{proof}
Shift the matrix by defining
$\bA_g:=\bM_g+\om_D\bI$.  Then
$0\preceq\bA_g\preceq2\om_D\bI$ and
$\lambda_{\max}(\bA_g)=\lambda_{\max}(\bM_g)+\om_D$.
The standard random-start Lánczos guarantee gives, using the number of matrix--vector products in
\Cref{eqn:upper-certificate-lanczos-degree}, a Ritz value
$\widehat\lambda_g$ satisfying
\begin{align}
    \lambda_{\max}(\bA_g)-\frac{\zeta}{\Rx}
    &\leq
    \widehat\lambda_g
    \leq
    \lambda_{\max}(\bA_g)
\end{align}
with probability at
least $1-\delta_{\mathrm{eig}}$.  The upper bound also follows directly from Rayleigh--Ritz, since the Ritz value maximizes the Rayleigh quotient over the Krylov subspace and therefore cannot exceed the global maximum $\lambda_{\max}(\bA_g)$.
Defining
\begin{align}
    \widetilde u_g
    &:=
    \Rx\bigl(\widehat\lambda_g-\om_D\bigr)
\end{align}
and using
$u_g=\Rx\bigl(\lambda_{\max}(\bA_g)-\om_D\bigr)$,
we obtain
\begin{align}
    u_g-\zeta
    \leq
    \widetilde u_g
    \leq
    u_g,
\end{align}
which proves \Cref{eqn:approximate-upper-certificate}.
\end{proof}

We therefore use the computable acceptance test
\begin{align}
    \widetilde u_g+\zeta
    &\geq
    0.
    \label{eqn:computable-target-test}
\end{align}
The one-sided guarantee in
\Cref{eqn:approximate-upper-certificate} is exactly what the target-value
search requires.  Indeed, if $g$ is attainable, then $u_g\geq0$, and hence
$\widetilde u_g+\zeta\geq u_g\geq0$, so the test never rejects an
attainable threshold.  Conversely, whenever the test retains $g$, we have
$\widetilde u_g\geq-\zeta$ and therefore $u_g\geq-\zeta$.  By
\Cref{prop:target-value-reduction}, the returned matrix strategy then
yields a primal solution whose objective and constraint errors increase
by at most an additional $\zeta$. We can therefore use the computable test
$\widetilde u_g+\zeta\geq0$ in a gapped binary search over the objective
value, yielding the following reduction.

\begin{corollary}[Recovering an Approximate SDP Solution]
\label{cor:target-value-binary-search}
Fix $\eta,\xi,\zeta>0$.  Suppose that every queried target-value game is
solved to saddle-point gap at most $\xi$, and that every computed upper
certificate satisfies
\Cref{eqn:approximate-upper-certificate}. Before beginning
the binary search, solve the target-value game once at $g=-\Rx$ and retain
this threshold together with its returned matrix.  Then perform binary
search on $[-\Rx,\Rx]$, moving the lower endpoint to the midpoint whenever
$\widetilde u_g+\zeta\geq0$, and moving the upper endpoint to the midpoint
otherwise. The procedure uses
$O(\log(2+\Rx/\eta))$ game solves in total and returns a
retained threshold $\widehat g$, together with its
associated matrix
$\widehat{\bX}\succeq\bzero$ with
$\Tr(\widehat{\bX})\leq\Rx$, such that
\begin{align}
    \operatorname{OPT}-\eta
    &\leq
    \widehat g
    \leq
    \operatorname{OPT}+(1+\Ry)(\xi+\zeta),
    \label{eqn:target-game-binary-search-bracket}
\end{align}
and
\begin{align}
    \inner{\bC}{\widehat{\bX}}
    &\geq
    \operatorname{OPT}-\eta-\xi-\zeta,
    \notag\\
    \inner{\A{j}}{\widehat{\bX}}
    &\leq
    b_j+\xi+\zeta,
    \qquad
    \forall~j\in[\mc].
    \label{eqn:target-game-final-primal-guarantee}
\end{align}
\end{corollary}

\begin{proof}
We have $\operatorname{OPT}\in[-\Rx,\Rx]$, since
$\norm{\bC}\leq1$ and a promised optimal primal matrix has trace at most
$\Rx$. Every $g\leq\operatorname{OPT}$ is attainable, so
$u_g\geq0$ by \Cref{prop:target-value-reduction}.  Hence
$\widetilde u_g+\zeta\geq0$, and such a threshold is always retained. On the other hand, suppose
\begin{align}
    g
    &>
    \operatorname{OPT}
    +(1+\Ry)(\xi+\zeta).
\end{align}
Then, setting
$\tau=(g-\operatorname{OPT})/(1+\Ry)$ gives
$\tau>\xi+\zeta$.  By
\Cref{prop:target-value-reduction},
$\mathfrak{s}_g\leq-\tau$, and therefore
\begin{align}
    u_g
    &\leq
    \ell_g+\xi
    \leq
    \mathfrak{s}_g+\xi
    <
    -\zeta.
\end{align}
Since $\widetilde u_g\leq u_g$, the test
\Cref{eqn:computable-target-test} rejects such a threshold. Thus, every threshold at or below $\operatorname{OPT}$ is retained, while
every threshold above
$\operatorname{OPT}+(1+\Ry)(\xi+\zeta)$ is rejected.
First solve the game at $g=-\Rx$, compute its upper certificate,
and save this accepted threshold together with its matrix.  Set the
search endpoints to $L=-\Rx$ and $U=\Rx$.  On acceptance of a midpoint,
replace $L$ and the saved pair by that midpoint and its returned matrix.
On rejection, replace $U$.  Since a rejected threshold is strictly above
$\operatorname{OPT}$, the invariant $U\geq\operatorname{OPT}$ is
preserved.  Stop when $U-L\leq\eta$ and return the saved pair, with
$\widehat g=L$.  Then, $L\geq U-\eta\geq\operatorname{OPT}-\eta$,
and every saved threshold obeys the upper bound just proved.  Including
initialization, the number of game solves is at most
$1+\max\{0,\lceil\log_2(2\Rx/\eta)\rceil\}$, which is
$O(\log(2+\Rx/\eta))$.  This proves
\Cref{eqn:target-game-binary-search-bracket}.
 Finally, retention implies
$\widetilde u_{\widehat g}\geq-\zeta$, and therefore
$u_{\widehat g}\geq-\zeta$.  Applying
\Cref{prop:target-value-reduction} gives
\begin{align}
    \inner{\bC}{\widehat{\bX}}
    &\geq
    \widehat g-\xi-\zeta
    \geq
    \operatorname{OPT}-\eta-\xi-\zeta
\end{align}
and
$\inner{\A{j}}{\widehat{\bX}}\leq b_j+\xi+\zeta$ for every
$j\in[\mc]$.
\end{proof}

Taking $\eta=\eps/2$ and
$\xi=\zeta=\eps/(8(1+\Ry))$ gives
$\eta+\xi+\zeta\leq\eps$ and
$\xi+\zeta\leq\eps/\Ry$.
Consequently, the recovered primal matrix has objective error at most
$\eps$ and constraint violations at most $\eps/\Ry$.
 The reduction uses
$O(\log(2+\Rx/\eps))$ target-value game solves and upper-certificate
computations.  Allocating the failure probability across these calls and
applying a union bound preserves the desired overall success probability.

Thus, the original SDP reduces to logarithmically many
spectrahedron--simplex games with matrix radius $\Rx$, simplex radius one,
and payoff matrices of operator norm at most two.  The original dual
radius $\Ry$ enters through the accuracy
$\Theta(\eps/(1+\Ry))$ required of the game solver and target-value
certificate.  The sparse implementation cost of the latter is included in
the runtime analysis of \Cref{subsec:main-theorem-proof}.

\section{A Coordinate-Sampling Estimator for Sparse Observables}
\label{sec:energy-estimator}

Both of our proposed algorithms require estimating expectation values of the
form $\bpsi^\dagger\A{j}\bpsi$, for $j\in[\mc]$, given a unit vector
$\bpsi\in\C^{\dn}$ and $\sr$-row-sparse Hermitian constraint matrices
$\A{1},\ldots,\A{\mc}$.  Exactly computing all of these quantities can
require accessing every row of every constraint matrix and therefore costs
$O(\dn\mc\sr)$ time, in the worst case.  We instead estimate all $\mc$
expectations from a single sampled coordinate of $\bpsi$, reducing the
matrix-access cost of one simultaneous estimate to $O(\mc\sr)$.  Moreover,
the estimator's second moment is controlled by the operator norms of the
constraint matrices rather than their Frobenius norms, so the resulting
guarantee remains effective even for generic full-rank observables.

To sample the coordinate, we use the squared-magnitude sampling procedure
from \Cref{sec:squared-magnitude-sampling}.  Since $\bpsi$ is stored
explicitly, we can sample a random coordinate $I\in[\dn]$ in
$O(\dn)$ time according to
$\Prb[I=i\mid\bpsi]=\abs{[\bpsi]_i}^2$.  For a Hermitian matrix
$\bA\in\C^{\dn\times\dn}$ and a fixed coordinate $i\in[\dn]$ in the
support of $\bpsi$, define
\begin{align}
\Zest_{\bA}(\bpsi,i)
&:=
\operatorname{Re}\left\{
\frac{[\bA\bpsi]_i}{[\bpsi]_i}
\right\}.
\label{eqn:energy-estimator}
\end{align}
Since $I$ lies in the support of $\bpsi$ almost surely,
$\Zest_{\bA}(\bpsi,I)$ is well-defined almost surely.  The ratio in
\Cref{eqn:energy-estimator} is the finite-dimensional analogue of the
\emph{local energy} estimator used in variational Monte
Carlo~\cite[Section~III.C]{FoulkesEtAl2001}. Crucially, the same realization of $I$ can be reused across all constraint
matrices.  If $\A{j}$ is $\sr$-row-sparse, then computing
$[\A{j}\bpsi]_I$ requires accessing only the $I^{th}$ row of $\A{j}$ and
therefore costs $O(\sr)$ time.  Consequently, all $\mc$ estimates
$\bigl(\Zest_{\A{j}}(\bpsi,I)\bigr)_{j\in[\mc]}$ can be generated in
$O(\dn+\mc\sr)$ time, including the cost of sampling $I$.

The following lemma establishes the two properties of the estimator that will be
used throughout the paper: 1) unbiasedness and 2) an operator-norm bounded second-moment.

\begin{lemma}[Moments of the Coordinate-Sampling Estimator]
\label{lem:energy-estimator-moments}
Let $\bA\in\C^{\dn\times\dn}$ be Hermitian with
$\norm{\bA}\leq\om$, let $\bpsi\in\C^{\dn}$ be a unit vector, and sample
$I\in[\dn]$ according to
$\Prb[I=i\mid\bpsi]=\abs{[\bpsi]_i}^2$.  Then the estimator
$\Zest_{\bA}(\bpsi,I)$ satisfies the following properties:
\begin{enumerate}[leftmargin=2.4em,label=\textup{(\roman*)}]
    \item \textbf{Unbiasedness:} $\E_I[\Zest_{\bA}(\bpsi,I)]=\bpsi^\dagger\bA\bpsi.$
    \item \textbf{Operator-norm bounded second moment:} $\E_I[\Zest_{\bA}(\bpsi,I)^2]\leq\norm{\bA\bpsi}_2^2 \leq
        \om^2$
\end{enumerate}
\end{lemma}

\begin{proof}
Let $S:=\{i\in[\dn]:[\bpsi]_i\neq 0\}$ denote the support of $\bpsi$. 

First, since $I\in S$ with probability one and
$\abs{[\bpsi]_i}^2/[\bpsi]_i=[\bpsi]_i^*$ for every $i\in S$,
we have that
\begin{align}
\E_I[\Zest_{\bA}(\bpsi,I)]
&=
\sum_{i\in S}
\abs{[\bpsi]_i}^2
\operatorname{Re}\left\{
\frac{[\bA\bpsi]_i}{[\bpsi]_i}
\right\}=
\operatorname{Re}\left\{
\sum_{i\in S}
[\bpsi]_i^*[\bA\bpsi]_i
\right\}=
\operatorname{Re}\left\{
\bpsi^\dagger\bA\bpsi
\right\}
=
\bpsi^\dagger\bA\bpsi,
\end{align}
where the final equality follows because $\bA$ is Hermitian.  This, thus, proves
the unbiasedness claim.

Second, using the fact that
$\operatorname{Re}(z)^2\leq\abs{z}^2$ for every $z\in\C$, we obtain
\begin{align}
\E_I[\Zest_{\bA}(\bpsi,I)^2]\leq
\sum_{i\in S}
\abs{[\bpsi]_i}^2
\left|
\frac{[\bA\bpsi]_i}{[\bpsi]_i}
\right|^2=
\sum_{i\in S}
\abs{[\bA\bpsi]_i}^2\leq
\norm{\bA\bpsi}_2^2
\leq
\norm{\bA}^2\norm{\bpsi}_2^2
\leq
\om^2.
\end{align}
Thus, although a coordinate $i$ with small $\abs{[\bpsi]_i}$ can produce
a large estimator value through the factor $1/[\bpsi]_i$, this is offset by
the sampling probability $\abs{[\bpsi]_i}^2$.  In the second moment, the
factor $\abs{[\bpsi]_i}^2$ therefore exactly cancels the squared denominator,
yielding the operator-norm bound above.
\end{proof}

The operator-norm second-moment bound controls the estimator on average, but
individual realizations remain unbounded.  In particular,
$\Zest_{\bA}(\bpsi,I)$ can be arbitrarily large when
$\abs{[\bpsi]_I}$ is small, even though such outcomes occur with
relatively small probability.  This suffices for applications that can control the unbounded estimator through robust averaging.  In our stochastic solver, however, the online learner
requires uniformly bounded feedback, so we additionally clip the estimator.

For a clipping threshold $B>0$, define
$\clip_B(z):=\max\{-B,\min\{z,B\}\}$ for $z\in\R$.  We then define the
clipped estimator
\begin{align}
    \widehat Z_{\bA,B}
    &:=
    \clip_B\!\left(
        \Rx\Zest_{\bA}(\bpsi,I)
    \right).
    \label{eqn:clipped-energy-estimator}
\end{align}
Clipping enforces the desired deterministic range and can only decrease
the second moment, at the cost of introducing bias.  The following lemma establishes the properties of the clipped estimator needed in our later analysis.

\begin{lemma}[Properties of the Clipped Estimator]
\label{lem:clipping-properties}
Under the assumptions of \Cref{lem:energy-estimator-moments}, let
$\Rx>0$, $B>0$, and $\bX:=\Rx\bpsi\bpsi^\dagger$.  Then the clipped
estimator $\widehat Z_{\bA,B}$ satisfies the following properties:
\begin{enumerate}
    \item \textbf{Uniform boundedness:}
    $|\widehat Z_{\bA,B}|\leq B$, almost surely.

    \item \textbf{Second-moment bound:}
    $\E_I[\widehat Z_{\bA,B}^2]\leq \Rx^2\om^2$.

    \item \textbf{Clipping bias:} $|\E_I[\widehat Z_{\bA,B}]-\inner{\bA}{\bX}|\leq \Rx^2\om^2/B.$
\end{enumerate}
\end{lemma}

\begin{proof}
For the first claim, the bound
$|\widehat Z_{\bA,B}|\leq B$ follows immediately from the definition
of $\clip_B$.

For the second claim, clipping can only decrease absolute value, so
$\widehat Z_{\bA,B}^2
\leq \Rx^2\Zest_{\bA}(\bpsi,I)^2$.  Taking expectations and applying
\Cref{lem:energy-estimator-moments} gives
$\E_I[\widehat Z_{\bA,B}^2]\leq \Rx^2\om^2$.

For the third claim, observe that for every $z\in\R$,
\begin{align}
    \abs{z-\clip_B(z)}
    &=
    \bigl(\abs{z}-B\bigr)_+
    \leq
    \frac{z^2}{B}.
    \label{eqn:scalar-clipping-bias}
\end{align}
Indeed, the left-hand side vanishes when $\abs{z}\leq B$, while for
$\abs{z}>B$ we have
$\abs{z}-B\leq\abs{z}\leq z^2/B$.  Moreover, by
\Cref{lem:energy-estimator-moments},
$\E_I[\Rx\Zest_{\bA}(\bpsi,I)]
=\Rx\bpsi^\dagger\bA\bpsi
=\inner{\bA}{\bX}$.  Therefore,
\begin{align}
    \abs{
        \E_I[\widehat Z_{\bA,B}]
        -
        \inner{\bA}{\bX}
    }\leq
    \E_I\left[
        \abs{
            \widehat Z_{\bA,B}
            -
            \Rx\Zest_{\bA}(\bpsi,I)
        }
    \right]\leq
    \frac{
        \Rx^2\E_I[\Zest_{\bA}(\bpsi,I)^2]
    }{B}
    \leq
    \frac{\Rx^2\om^2}{B},
\end{align}
where the final inequality again follows from
\Cref{lem:energy-estimator-moments}.
\end{proof}

The two lemmas isolate the guarantees needed by our later applications.
The raw estimator in \Cref{lem:energy-estimator-moments} is unbiased and
has an operator-norm second-moment bound, preventing the sample complexity
from acquiring a hidden dependence on $\norm{\A{j}}_F^2$ or $\dn$.  In
\Cref{sec:gibbs-or}, its unbounded realizations are controlled through
robust averaging across independent random starts.  By contrast, the
stochastic solver of \Cref{sec:algorithm} uses only one coordinate sample
per round and therefore clips the resulting feedback before passing it to
the online learner.  The clipping threshold is chosen there to balance the
bias from \Cref{lem:clipping-properties} against the concentration error
arising from the bounded range.

\section{Lánczos Approximation of Matrix-Exponential--Vector Products}
\label{sec:lanczos}

A common computational primitive required for both of our algorithms is the
ability to compute a vector of the form $e^{\bK}\bu$, given a Hamiltonian
$\bK$ and a random vector $\bu$.  In particular, the simultaneous Gibbs
expectation estimator computes $e^{-\bH/2}\bu$ for a Haar-random vector
$\bu$, while the stochastic matrix learner computes
$e^{\lrX\Ht{t}/2}\ut{t}$ for a real random vector $\ut{t}$.

Although the Hamiltonians $\bH$ and $\Ht{t}$ themselves are sparse, or at
least admit efficient matrix--vector multiplication, their exponentials
are generally dense.  Explicitly forming the exponential would therefore
destroy the structure that makes the algorithms efficient.  To avoid this,
we use Lánczos approximation, which computes $e^{\bK}\bu$ by projecting the
problem onto a low-dimensional Krylov space.  Starting from $\bu$, the
Lánczos procedure constructs an orthonormal basis for
$\operatorname{span}\{\bu,\bK\bu,\ldots,\bK^{k-1}\bu\}$, compresses
$\bK$ to a small tridiagonal matrix, evaluates the exponential of this
compressed matrix, and lifts the resulting vector back to the original
space.  Consequently, the high-dimensional part of the computation
requires only repeated $\bK$ multiplication.

Although the Lánczos recurrence is deterministic once the starting vector
is fixed, both of our algorithms choose that vector randomly.  This
randomness is useful because the deterministic Lánczos analysis naturally
gives an absolute approximation guarantee, whereas our algorithms
subsequently normalize the resulting vector and therefore require relative
error.  A random starting direction has, with high probability,
sufficiently large overlap with a top $\bK$ eigendirection to convert
the absolute approximation guarantee into a relative one.

We first record the standard deterministic Krylov approximation guarantee
and then derive the random-start guarantees used by our two algorithms.

\subsection{Deterministic Lánczos Approximation}

We first make precise the exact-arithmetic Hermitian Lánczos procedure used
throughout the paper.  For a fixed starting vector, the procedure is
deterministic. Randomness enters only through the choice of the starting
vector in our two applications.  For $j\geq1$, define the Krylov space
\begin{align}
\mathcal K_j(\bK,\bu)
&:=
\operatorname{span}
\{\bu,\bK\bu,\ldots,\bK^{j-1}\bu\}.
\end{align}
Lánczos constructs an orthonormal basis for this space one vector at a
time, using one multiplication by $\bK$ per iteration.

The procedure begins with the normalized starting vector
$\mathbf q_1=\bu/\norm{\bu}_2$.  At iteration $j$, it multiplies the
current basis vector $\mathbf q_j$ by $\bK$ and removes its components
along $\mathbf q_{j-1}$ and $\mathbf q_j$.  The corresponding coefficients
are denoted by $\beta_{j-1}$ and $\alpha_j$, respectively.  Since $\bK$ is
Hermitian, in exact arithmetic the resulting residual is automatically
orthogonal to all previously generated basis vectors.  If the residual is
nonzero, its norm $\beta_j$ is used to normalize it and produce the next
basis vector $\mathbf q_{j+1}$.  If the residual is zero, applying $\bK$ to
the current Krylov space produces no new direction, and the procedure
terminates.

\begin{algorithm}[t!]
\small
\caption{Lánczos Approximation of
$e^{\bK}\bu$}
\label{alg:lanczos-exp}
\begin{algorithmic}[1]
\Require Hermitian $\bK\in\C^{\dn\times\dn}$, nonzero
$\bu\in\C^{\dn}$, and iteration budget $1\leq k\leq\dn$
\State $\mathbf q_0\gets\bzero$, $\beta_0\gets0$, and
$\mathbf q_1\gets\bu/\norm{\bu}_2$
\State $\ell\gets k$
\For{$j=1,\ldots,k$}
    \State $\mathbf w\gets
    \bK\mathbf q_j-\beta_{j-1}\mathbf q_{j-1}$
    \State $\alpha_j\gets\mathbf q_j^\dagger\mathbf w$
    \State $\mathbf w\gets\mathbf w-\alpha_j\mathbf q_j$
    \State $\beta_j\gets\norm{\mathbf w}_2$
    \If{$\beta_j=0$}
        \State $\ell\gets j$
        \State \textbf{break}
    \EndIf
    \If{$j<k$}
        \State $\mathbf q_{j+1}\gets\mathbf w/\beta_j$
    \EndIf
\EndFor
\State $\mathbf Q_\ell\gets
[\mathbf q_1\ \cdots\ \mathbf q_\ell]$
\State Form $\mathbf T_\ell\in\C^{\ell\times\ell}$ with
\Statex \hspace{\algorithmicindent}
$(\mathbf T_\ell)_{j,j}=\alpha_j$ for $j=1,\ldots,\ell$
\Statex \hspace{\algorithmicindent}
and
$(\mathbf T_\ell)_{j,j+1}
=(\mathbf T_\ell)_{j+1,j}=\beta_j$
for $j=1,\ldots,\ell-1$
\State \Return
$\norm{\bu}_2\mathbf Q_\ell e^{\mathbf T_\ell}\mathbf e_1$
\end{algorithmic}
\end{algorithm}

Here $\ell\leq k$ is the number of basis vectors generated by the
procedure.  The columns of
$\mathbf Q_\ell=[\mathbf q_1\ \cdots\ \mathbf q_\ell]$ form an
orthonormal basis for $\mathcal K_\ell(\bK,\bu)$, while the coefficients
$\alpha_j$ and $\beta_j$ form the Hermitian tridiagonal compression
\begin{align}
\mathbf T_\ell
&:=
\mathbf Q_\ell^\dagger\bK\mathbf Q_\ell.
\end{align}
Thus, $\mathbf T_\ell$ represents the action of the large matrix $\bK$
after projection onto the low-dimensional Krylov space. Since
$\mathbf Q_\ell^\dagger\bu=\norm{\bu}_2\mathbf e_1$, where
$\mathbf e_1$ is the first standard basis vector in $\C^\ell$, the final
step first computes the exponential on the compressed problem and then
lifts the resulting vector back to the original space.  We denote the
returned vector by
\begin{align}
\LanczosExp(\bK,\bu,k)
&:=
\norm{\bu}_2\mathbf Q_\ell
e^{\mathbf T_\ell}\mathbf e_1
=
\mathbf Q_\ell e^{\mathbf T_\ell}
\mathbf Q_\ell^\dagger\bu.
\label{eqn:lanczos-exp-definition}
\end{align}
The high-dimensional part of the computation therefore consists only of
the $\ell$ matrix--vector products $\bK\mathbf q_j$ and the matrix
exponential is evaluated only for the $\ell\times\ell$ tridiagonal matrix
$\mathbf T_\ell$. If the loop encounters $\beta_j=0$, then no additional Krylov direction
can be generated.  In this case $\ell=j$, the space
$\mathcal K_\ell(\bK,\bu)$ is invariant under $\bK$, and the returned
vector equals $e^{\bK}\bu$ exactly.  Otherwise, the procedure exhausts its
iteration budget and $\ell=k$.

Although \Cref{alg:lanczos-exp} is stated for the exponential, the Krylov
basis construction does not depend on the choice of matrix function.
Replacing $e^{\mathbf T_\ell}$ by $f(\mathbf T_\ell)$ gives the standard
Lánczos approximation $\mathbf Q_\ell f(\mathbf T_\ell)
\mathbf Q_\ell^\dagger\bu$ to $f(\bK)\bu$.  The standard exact-arithmetic bound of
Saad~\cite{Saad1992} and Orecchia, Sachdeva, and
Vishnoi~\cite[Theorem~6.7]{OSV12} reduces the error of this approximation
to the error of uniformly approximating the scalar function $f$ on the
spectral interval of $\bK$.  We include the short proof to make explicit
that polynomial approximation is the only numerical ingredient needed in
our exact-arithmetic analysis.

\begin{lemma}[Exact-Arithmetic Lánczos Approximation]
\label{lem:exact-lanczos}
Let $\bK\in\C^{\dn\times\dn}$ be Hermitian, let
$\bu\in\C^{\dn}$ be nonzero, and let $1\leq k\leq\dn$.  Let
$\ell$, $\mathbf Q_\ell$, and $\mathbf T_\ell$ be the quantities generated
by \Cref{alg:lanczos-exp} with iteration budget $k$.  Then, for every
continuous real-valued function $f$ on the spectral interval of $\bK$,
\begin{align}
\norm{
f(\bK)\bu
-
\mathbf Q_\ell f(\mathbf T_\ell)
\mathbf Q_\ell^\dagger\bu
}_2
&\leq
2\norm{\bu}_2
\min_{\deg(p)<k}
\max_{x\in[\lambda_{\min}(\bK),\lambda_{\max}(\bK)]}
\abs{f(x)-p(x)}.
\label{eqn:exact-lanczos-polynomial-error}
\end{align}
\end{lemma}

\begin{proof}
Suppose first that the algorithm encounters $\beta_j=0$, so that
$\ell=j$.  The vanishing residual means that applying $\bK$ to the final
Lánczos basis produces no component outside its span.  Hence
$\mathcal K_\ell(\bK,\bu)$ is invariant under $\bK$, and the restriction of
$\bK$ to this subspace is represented in the basis $\mathbf Q_\ell$ by
$\mathbf T_\ell$.  Since $\bu$ belongs to this invariant subspace, $f(\bK)\bu=\mathbf Q_\ell f(\mathbf T_\ell) \mathbf Q_\ell^\dagger\bu$.
Thus, \Cref{eqn:exact-lanczos-polynomial-error} follows trivially.

Suppose now that the algorithm does not encounter a zero residual.  It
then performs all $k$ iterations, so $\ell=k$.  Lánczos reproduces every
polynomial of degree less than $\ell$ when applied to the starting vector. That is,
for every polynomial $p$ with $\deg(p)<\ell=k$, $p(\bK)\bu=\mathbf Q_\ell p(\mathbf T_\ell)\mathbf Q_\ell^\dagger\bu$.
Consequently,
\begin{align}
f(\bK)\bu
-
\mathbf Q_\ell f(\mathbf T_\ell)
\mathbf Q_\ell^\dagger\bu
&=
\bigl(f(\bK)-p(\bK)\bigr)\bu+
\mathbf Q_\ell
\bigl(p(\mathbf T_\ell)-f(\mathbf T_\ell)\bigr)
\mathbf Q_\ell^\dagger\bu.
\label{eqn:lanczos-polynomial-comparison}
\end{align}
Because $\mathbf T_\ell$ is a Hermitian compression of $\bK$, its spectrum
is contained in
$[\lambda_{\min}(\bK),\lambda_{\max}(\bK)]$.  The spectral theorem
therefore gives
\begin{align}
\norm{\bigl(f(\bK)-p(\bK)\bigr)\bu}_2
&\leq
\norm{\bu}_2
\max_{x\in[\lambda_{\min}(\bK),\lambda_{\max}(\bK)]}
\abs{f(x)-p(x)}
\end{align}
and
\begin{align}
\norm{
\mathbf Q_\ell
\bigl(p(\mathbf T_\ell)-f(\mathbf T_\ell)\bigr)
\mathbf Q_\ell^\dagger\bu
}_2
&\leq
\norm{\bu}_2
\max_{x\in[\lambda_{\min}(\bK),\lambda_{\max}(\bK)]}
\abs{f(x)-p(x)}.
\end{align}
Combining these two bounds and minimizing over all polynomials of degree
less than $k$ proves
\Cref{eqn:exact-lanczos-polynomial-error}.
\end{proof}

\subsection{Random-Start Lánczos Approximation}
\label{subsec:random-start-lanczos}

The deterministic guarantee in \Cref{lem:exact-lanczos} controls the
absolute error in the vector returned by \Cref{alg:lanczos-exp}.  Our
algorithms, however, subsequently normalize this vector, so we require an
error bound relative to the norm of the exact matrix-exponential--vector
product.  Such a guarantee need not hold for an arbitrary fixed starting
vector $\bu$. Namely, the deterministic approximation error naturally scales with
$e^{\lambda_{\max}(\bK)}$, while $\norm{e^{\bK}\bu}_2$ can be much smaller
if $\bu$ has little overlap with the top eigenspaces of $\bK$.  Choosing
the starting vector at random rules out this unfavorable alignment with
high probability.

Our analysis of the stochastic matrix learner is carried out for real symmetric Hamiltonians and real spherical starting vectors in order to invoke the average-projection regret theorem of Carmon et al. In this setting, the required guarantee is already
provided by Carmon et al.~\cite[Corollary~20 and
Proposition~23]{CDST19}.  We record the result in the notation of
\Cref{alg:lanczos-exp} for later reference.  Since we work in exact
arithmetic, we omit the finite-precision guarantees from their statement.

\begin{fact}[Uniform Random Real Lánczos]
\label{fact:real-random-start-lanczos}
Let $\bK\in\R^{\dn\times\dn}$ be symmetric, let
$\bu\sim\operatorname{Unif}(\mathbb S^{\dn-1})$, and fix
$\zeta,\eta\in(0,1/2)$.  For
$q=\widetilde O(\sqrt{\max\{1,\norm{\bK}\}})$,
capped at $\dn$, the vector
$\widetilde\bv^{(\bu)}:=\LanczosExp(\bK,\bu,q)$ satisfies
\begin{align}
    \Prb_{\bu}\left[
        \norm{
            \widetilde\bv^{(\bu)}
            -
            e^{\bK}\bu
        }_2
        \leq
        \zeta\norm{e^{\bK}\bu}_2
    \right]
    &\geq
    1-\eta.
    \label{eqn:real-lanczos-relative-guarantee}
\end{align}
If one multiplication by $\bK$ costs $T_{\bK}$, the arithmetic cost is $O(q(T_{\bK}+\dn)+q^2)$.
\end{fact}

We will use \Cref{fact:real-random-start-lanczos} directly for the stochastic matrix learner by conditioning on the history
before round $t$ and taking $\bK=\lrX\Ht{t}/2$.  The simultaneous Gibbs expectation estimator instead uses a complex Hermitian
matrix and a Haar-random vector in $\C^{\dn}$, so the real-valued result
does not apply directly.  We therefore prove the following complex
analogue.

\begin{proposition}[Haar Random Lánczos]
\label{prop:gibbs-random-start-lanczos}
Let $\bK\in\C^{\dn\times\dn}$ be Hermitian, let
$\bu\sim\operatorname{Haar}(\C^{\dn})$, and fix
$\zeta,\eta\in(0,1/2)$.  For
$q=\widetilde O(\sqrt{\max\{1,\norm{\bK}\}})$,
capped at $\dn$, the vector
$\widetilde\bv^{(\bu)}:=\LanczosExp(\bK,\bu,q)$ satisfies
\begin{align}
    \Prb_{\bu}\left[
        \norm{
            \widetilde\bv^{(\bu)}
            -
            e^{\bK}\bu
        }_2
        \leq
        \zeta\norm{e^{\bK}\bu}_2
    \right]
    &\geq
    1-\eta.
    \label{eqn:gibbs-lanczos-relative-guarantee}
\end{align}
If one multiplication by $\bK$ costs $T_{\bK}$, the arithmetic cost is
$O(q(T_{\bK}+\dn)+q^2)$.
\end{proposition}

\begin{proof}
If the cap $q=\dn$ is active, then \Cref{alg:lanczos-exp} returns
$e^{\bK}\bu$ exactly by or before the $\dn$-th iteration.  We therefore
assume that $q<\dn$.

By \cite[Corollary~20]{CDST19}, for any interval $[a,b]$ and
$\delta\in(0,1]$, there exists a polynomial $p$ of degree
$O(\sqrt{\max\{b-a,\log(1/\delta)\}\log(1/\delta)})$
such that
\begin{align}
\max_{x\in[a,b]}\abs{e^x-p(x)}
&\leq
\delta e^b.
\end{align}
Apply this with
$a=\lambda_{\min}(\bK)$,
$b=\lambda_{\max}(\bK)$, and
$\delta=(\zeta/2)\sqrt{\eta/\dn}$.  Since
$\lambda_{\max}(\bK)-\lambda_{\min}(\bK)\leq2\norm{\bK}$,
the required polynomial has degree
$\widetilde O(\sqrt{\max\{1,\norm{\bK}\}})$.  Thus, for $q$ as in the
proposition, we may choose $p$ with $\deg(p)<q$ and
\begin{align}
\max_{
x\in[
\lambda_{\min}(\bK),
\lambda_{\max}(\bK)
]
}
\abs{e^x-p(x)}
&\leq
\frac{\zeta}{2}
e^{\lambda_{\max}(\bK)}
\sqrt{\frac{\eta}{\dn}}.
\label{eqn:complex-exponential-polynomial-error}
\end{align}
Since $\norm{\bu}_2=1$, \Cref{lem:exact-lanczos} then gives
\begin{align}
\norm{
\widetilde\bv^{(\bu)}
-
e^{\bK}\bu
}_2
&\leq
\zeta
e^{\lambda_{\max}(\bK)}
\sqrt{\frac{\eta}{\dn}}.
\label{eqn:complex-lanczos-absolute-error}
\end{align}

It remains to lower bound $\norm{e^{\bK}\bu}_2$.  Let $\bw_{\max}$ be a
unit eigenvector of $\bK$ corresponding to $\lambda_{\max}(\bK)$.  Since
\begin{align}
\norm{e^{\bK}\bu}_2
&\geq
e^{\lambda_{\max}(\bK)}
\abs{\bw_{\max}^\dagger\bu},
\end{align}
it suffices to lower bound the random overlap
$\abs{\bw_{\max}^\dagger\bu}$.  For $\dn\geq2$, the squared overlap
$\abs{\bw_{\max}^\dagger\bu}^2$ has the
$\operatorname{Beta}(1,\dn-1)$ distribution.  Hence,
\begin{align}
\Prb_{\bu}\left[
\abs{\bw_{\max}^\dagger\bu}
<
\sqrt{\frac{\eta}{\dn}}
\right]
&=
1-
\left(
1-\frac{\eta}{\dn}
\right)^{\dn-1}\leq
\frac{(\dn-1)\eta}{\dn}
\leq
\eta.
\label{eqn:complex-haar-anticoncentration}
\end{align}
Thus, with probability at least $1-\eta$,
\begin{align}
\norm{e^{\bK}\bu}_2
&\geq
e^{\lambda_{\max}(\bK)}
\sqrt{\frac{\eta}{\dn}}.
\end{align}
Combining this bound with
\Cref{eqn:complex-lanczos-absolute-error} gives
\begin{align}
\norm{
\widetilde\bv^{(\bu)}
-
e^{\bK}\bu
}_2
&\leq
\zeta\norm{e^{\bK}\bu}_2,
\end{align}
which proves \Cref{eqn:gibbs-lanczos-relative-guarantee}.

Finally, \Cref{alg:lanczos-exp} performs at most $q$ matrix--vector
multiplications by $\bK$ and $O(q\dn)$ additional vector arithmetic.
Computing the exponential of the tridiagonal compression and lifting the
result through the Krylov basis costs $O(q^2+q\dn)$.  Thus the total
arithmetic cost is $O(q(T_{\bK}+\dn)+q^2)$.
\end{proof}

\section{Simultaneous Gibbs Expectation Estimation}
\label{sec:gibbs-or}

This section develops the main estimation primitive underlying our
dequantization.  Given a Hamiltonian $\bH$, with corresponding Gibbs state
$\Gibbs{\bH}:=e^{-\bH}/\Tr(e^{-\bH})$, our goal is to estimate
simultaneously the Gibbs expectation of many sparse observables. The key idea is to separate the expensive Hamiltonian-dependent step from
the inexpensive observable-dependent step. In particular, for each random start $\bu$,
we compute $e^{-\bH/2}\bu$ once and then reuse the resulting filtered
vector to estimate every observable.

Throughout this section, unless stated otherwise, the random starting
vectors are Haar-random, with
$\bu\sim\operatorname{Haar}(\C^{\dn})$.  We use parenthesized
superscripts to indicate dependence on a particular random start.  Thus
$\bv^{(\bu)}$, $W^{(\bu)}$, $\bpsi^{(\bu)}$, and $I^{(\bu)}$ are all
constructed from the same draw $\bu$. When we use several independent random starts, we write
$\bu^{(1)},\ldots,\bu^{(R)}$ for the draws and use the same superscript to
index all quantities derived from them.  For example,
$\bv^{(r)}$, $W^{(r)}$, $\bpsi^{(r)}$, and $I^{(r)}$ are all constructed
from the random start $\bu^{(r)}$.  Subscripts are reserved for coordinates
and observables.

We now formalize this simultaneous estimation primitive.  The following theorem shows that a single collection of Hamiltonian-filtered random
vectors suffices to estimate all $\mc$ Gibbs expectations, with the
Hamiltonian-dependent cost incurred only once and an additional sparse
observable-dependent cost that scales linearly with $\mc$.

\begin{theorem}[Efficient Simultaneous Gibbs Expectation Estimation]
\label{thm:classical-gibbs-moments}
Let
$\bH=\sum_{\ell=1}^{h} y_\ell \bB_\ell$
with $\by\in\R^h$ and $\norm{\by}_1\leq B_{\bH}$, where each
$\bB_\ell\in\C^{\dn\times\dn}$ is Hermitian,
$s_{\bB}$-row-sparse, and satisfies $\norm{\bB_\ell}\leq1$.
Let
$\A{1},\ldots,\A{\mc}\in\C^{\dn\times\dn}$ be Hermitian,
$s_{\bA}$-row-sparse, and satisfy $\norm{\A{j}}\leq\om$ for every
$j\in[\mc]$. For any accuracy $\alpha\in(0,\om]$ and failure probability
$\del\in(0,1)$, there is a randomized exact-arithmetic algorithm that
outputs estimates
$\widehat\mu_1,\ldots,\widehat\mu_{\mc}$ such that
\begin{align}
    \Prb\!\left[
        \max_{j\in[\mc]}
        \left|
            \widehat\mu_j
            -
            \Tr\!\left(\A{j}\Gibbs{\bH}\right)
        \right|
        \leq
        \alpha
    \right]
    &\geq
    1-\del.
    \label{eqn:simultaneous-gibbs-moment-guarantee}
\end{align}
The sparse-oracle running time is
\begin{align}
    \widetilde O\!\left(
        \frac{\om^2}{\alpha^2}
        \left[
            \dn(hs_{\bB}+1)
            \min\!\left\{
                \dn,
                \sqrt{1+B_{\bH}}
            \right\}
            +
            \mc s_{\bA}
        \right]
    \right).
    \label{eqn:gibbs-moment-runtime}
\end{align}
\end{theorem}

\noindent The proof of \Cref{thm:classical-gibbs-moments} proceeds in four
steps:
\begin{enumerate}[leftmargin=2.3em,label=\textup{\arabic*.}]
    \item We first show how Gibbs expectations can be expressed as ratios
    of expectations over random vectors filtered by $e^{-\bH/2}$
    (\Cref{sec:two-exponential-averages}).

    \item We derive the Haar moment identities for these filtered random
    vectors and combine them with the shared-coordinate estimator from
    \Cref{sec:energy-estimator} to obtain unbiased samples, with controlled
    second moments
    (\Cref{sec:unnormalized-gibbs-moments}).

    \item We convert these second-moment bounds into simultaneous
    high-probability estimates for all observables using a
    median-of-means construction
    (\Cref{sec:robust-simultaneous-ratios}).

    \item Finally, we show that the estimator remains accurate when the exact
    filtered vector $e^{-\bH/2}\bu$ is replaced by its finite-Krylov Lánczos
    approximation.  We then combine the preceding ingredients to prove
    \Cref{thm:classical-gibbs-moments} and derive its sparse-oracle runtime
    (\Cref{sec:gibbs-moment-proof}).
\end{enumerate}

The estimator itself is simple to describe at a high level.  For each
Haar-random direction $\bu$, we form the filtered vector
$\bv^{(\bu)}=e^{-\bH/2}\bu$.  We then normalize this vector, sample a
single coordinate, and reuse that coordinate to estimate all observables
simultaneously.  The weight
$W^{(\bu)}=\norm{\bv^{(\bu)}}_2^2$ compensates for the normalization of
the filtered vector, so that ratios of the resulting empirical averages
recover the desired Gibbs expectations.  The following subsections make
this construction precise and establish its accuracy and runtime.

\subsection{Recovering Gibbs Expectations from Filtered Random Vectors}
\label{sec:two-exponential-averages}

We begin by showing how Gibbs expectations can be recovered from random
vectors filtered by the matrix exponential.  Let
$\bu\sim\operatorname{Haar}(\C^{\dn})$.  We define the filtered random
vector
\begin{align}
    \bv^{(\bu)}
    &:=
    e^{-\bH/2}\bu,
    \label{eqn:gibbs-filtered-vector}
\end{align}
and its squared norm
\begin{align}
    W^{(\bu)}
    &:=
    \|\bv^{(\bu)}\|_2^2.
    \label{eqn:gibbs-filtered-weight}
\end{align}

Our first observation is that averaging the outer products of the filtered
vectors recovers the unnormalized Gibbs operator.  By the Haar first-moment
identity of \Cref{fact:haar-first-moment},
\begin{align}
    \E_{\bu}\!\left[
        \bv^{(\bu)}(\bv^{(\bu)})^\dagger
    \right]
    &=
    e^{-\bH/2}
    \E_{\bu}\!\left[\bu\bu^\dagger\right]
    e^{-\bH/2}
    =
    \frac{e^{-\bH}}{\dn}.
    \label{eqn:filtered-vector-matrix-expectation}
\end{align}
Thus, up to the factor $1/\dn$, the average filtered outer product is
exactly the unnormalized Gibbs operator $e^{-\bH}$. The corresponding normalization factor is obtained by averaging the
squared norm of the filtered vector.  Taking the trace of
\Cref{eqn:filtered-vector-matrix-expectation} gives
\begin{align}
    \E_{\bu}\!\left[
        W^{(\bu)}
    \right]
    &=
    \frac{\Tr(e^{-\bH})}{\dn}.
    \label{eqn:filtered-vector-weight-expectation}
\end{align}
The factors of $1/\dn$ cancel between these two expectations, so their
ratio is exactly the Gibbs state:
\begin{align}
    \Gibbs{\bH}
    :=
    \frac{e^{-\bH}}{\Tr(e^{-\bH})}
    &=
    \frac{
        \E_{\bu}\!\left[
            \bv^{(\bu)}(\bv^{(\bu)})^\dagger
        \right]
    }{
        \E_{\bu}\!\left[
            W^{(\bu)}
        \right]
    }.
    \label{eqn:ratio-expectations-is-gibbs}
\end{align}

This representation immediately yields the form needed for estimating
Gibbs expectations.  For any Hermitian observable $\bA$, taking the
inner product with $\bA$ gives
\begin{align}
    \Tr(\bA\Gibbs{\bH})
    &=
    \frac{
        \E_{\bu}\!\left[
            (\bv^{(\bu)})^\dagger
            \bA
            \bv^{(\bu)}
        \right]
    }{
        \E_{\bu}\!\left[
            W^{(\bu)}
        \right]
    }.
    \label{eqn:gibbs-expectation-ratio}
\end{align}
Thus, to estimate $\Tr(\bA\Gibbs{\bH})$, it suffices to
estimate the observable-dependent quantity
$\E_{\bu}[(\bv^{(\bu)})^\dagger\bA\bv^{(\bu)}]$ together with the common
normalization factor $\E_{\bu}[W^{(\bu)}]$.  The following subsections
develop an efficient simultaneous estimator for these quantities.

\subsection{Unnormalized Gibbs Moments}
\label{sec:unnormalized-gibbs-moments}

By \Cref{eqn:gibbs-expectation-ratio}, each Gibbs expectation is determined
by the expectations over the Haar-random vector $\bu$ of
$(\bv^{(\bu)})^\dagger\bA\bv^{(\bu)}$ and $W^{(\bu)}$.  We now bound the
first and second moments of these random quantities in order to control
their empirical averages and obtain a finite-sample estimator.  Throughout
this subsection, we work with the exact filtered vector
$\bv^{(\bu)}=e^{-\bH/2}\bu$. The error introduced by its finite-Krylov
Lánczos approximation is analyzed separately in
\Cref{sec:gibbs-perturbation}.

Recall from \Cref{eqn:gibbs-filtered-weight} that
$W^{(\bu)}=\|\bv^{(\bu)}\|_2^2$.  Since its expectation appears
repeatedly throughout the analysis, we will denote it as
\begin{align} \label{eqn:def_tau}
    \tau
    &:=
    \E_{\bu}[W^{(\bu)}]
    =
    \frac{\Tr(e^{-\bH})}{\dn}.
\end{align}
We begin by bounding the second moment of $W^{(\bu)}$, which controls the
fluctuations of the common denominator and will also enter the
second-moment bound for the observable-dependent numerator.

\begin{lemma}[Second moment of the Gibbs weight]
\label{lem:gibbs-haar-moments}
Let $\bH\in\C^{\dn\times\dn}$ be Hermitian and
$\bu\sim\operatorname{Haar}(\C^{\dn})$.  Then,
\begin{align}
    \E_{\bu}\left[(W^{(\bu)})^2\right]
    &\leq
    2\tau^2.
    \label{eqn:gibbs-weight-moments}
\end{align}
\end{lemma}

\begin{proof}
Since
$W^{(\bu)}=\bu^\dagger e^{-\bH}\bu$, the Haar quadratic-form
second-moment identity from \Cref{fact:haar-second-moment}, applied with
$\bM=\bN=e^{-\bH}$, gives
\begin{align}
    \E_{\bu}\!\left[(W^{(\bu)})^2\right]
    &=
    \frac{
        \Tr(e^{-\bH})^2+\Tr(e^{-2\bH})
    }{
        \dn(\dn+1)
    }.
\end{align}
Because $e^{-\bH}\succeq\bzero$, we have
$\Tr(e^{-2\bH})\leq\Tr(e^{-\bH})^2$.  Hence, by the definition of $\tau$ in \Cref{eqn:def_tau},
\begin{align}
    \E_{\bu}\!\left[(W^{(\bu)})^2\right]
    &\leq
    \frac{2\Tr(e^{-\bH})^2}{\dn(\dn+1)}
    \leq
    \frac{2\Tr(e^{-\bH})^2}{\dn^2}
    =
    2\tau^2.
\end{align}
\end{proof}
We now combine the exponential-filtered vectors with the shared-coordinate estimator
from \Cref{sec:energy-estimator}.  For each realization of $\bu$, define
the normalized filtered vector
\begin{align}
    \bpsi^{(\bu)}
    &:=
    \frac{\bv^{(\bu)}}{\|\bv^{(\bu)}\|_2}.
    \label{eqn:gibbs-normalized-filtered-vector}
\end{align}
We then sample a coordinate $I^{(\bu)}\in[\dn]$ according to
\begin{align}
    \Prb\!\left[
        I^{(\bu)}=i
        \,\middle|\,
        \bv^{(\bu)}
    \right]
    &:=
    \abs{[\bpsi^{(\bu)}]_i}^2.
    \label{eqn:gibbs-coordinate-sampling}
\end{align}
For a Hermitian observable $\bA$, we weight the corresponding
shared-coordinate estimate by $W^{(\bu)}$ and define
\begin{align}
    X_{\bA}^{(\bu)}
    &:=
    W^{(\bu)}
    \Zest_{\bA}\!\left(
        \bpsi^{(\bu)},
        I^{(\bu)}
    \right).
    \label{eqn:weighted-energy-sample}
\end{align}
Crucially, the same sampled coordinate $I^{(\bu)}$ can be reused for all
observables.

The next lemma shows that these weighted samples have mean
$\tau\mu_j$, matching the numerator in
\Cref{eqn:gibbs-expectation-ratio}, while their second moments remain
controlled by the operator norm of the observable.

\begin{lemma}[Weighted shared-coordinate Gibbs sample]
\label{lem:weighted-gibbs-sample}
Let $\A{1},\ldots,\A{\mc}$ be Hermitian, with
$\norm{\A{j}}\leq\om$, and define
$X_j^{(\bu)}:=X_{\A{j}}^{(\bu)}$ as in
\Cref{eqn:weighted-energy-sample}.  The random variables
$X_1^{(\bu)},\ldots,X_{\mc}^{(\bu)}$ may be correlated, since they are
constructed from the same random start and sampled coordinate, i.e. no
independence across observables is required. Denote the Gibbs expectation of $\A{j}$ by $\mu_j:=\Tr(\A{j}\Gibbs{\bH})$.
Then, for every $j\in[\mc]$,
\begin{align}
    \E_{\bu}\!\left[
        \E_{I^{(\bu)}\mid\bu}\!\left[
            X_j^{(\bu)}
        \right]
    \right]
    &=
    \tau\mu_j,
    \label{eqn:weighted-sample-mean}
\end{align}
and
\begin{align}
    \E_{\bu}\!\left[
        \E_{I^{(\bu)}\mid\bu}\!\left[
            (X_j^{(\bu)})^2
        \right]
    \right]
    &\leq
    2\om^2\tau^2.
    \label{eqn:weighted-sample-second-moment}
\end{align}
Moreover, $|\mu_j|\leq\om$.
\end{lemma}

\begin{proof}
Fix $j\in[\mc]$.  We first prove the claimed expression for the mean of
$X_j^{(\bu)}$.  Conditional on $\bu$, the only remaining randomness is the
sampled coordinate $I^{(\bu)}$.  By
\Cref{lem:energy-estimator-moments}, averaging the sampled-coordinate
estimator over this coordinate sampling recovers the expectation of
$\A{j}$ with respect to the normalized filtered vector
$\bpsi^{(\bu)}$.  Hence,
\begin{align}
    \E_{I^{(\bu)}\mid\bu}\!\left[
        X_j^{(\bu)}
    \right]=
    W^{(\bu)}
    (\bpsi^{(\bu)})^\dagger
    \A{j}
    \bpsi^{(\bu)}=
    W^{(\bu)}
    \left(
        \frac{\bv^{(\bu)}}{\sqrt{W^{(\bu)}}}
    \right)^\dagger
    \A{j}
    \left(
        \frac{\bv^{(\bu)}}{\sqrt{W^{(\bu)}}}
    \right)
    =
    (\bv^{(\bu)})^\dagger
    \A{j}
    \bv^{(\bu)}.
\end{align}
Taking expectation over the Haar-random vector $\bu$ and using the law of
total expectation gives
\begin{align}
    \E_{\bu}\!\left[
        \E_{I^{(\bu)}\mid\bu}\!\left[
            X_j^{(\bu)}
        \right]
    \right]
    &=
    \E_{\bu}\!\left[
        (\bv^{(\bu)})^\dagger
        \A{j}
        \bv^{(\bu)}
    \right].
\end{align}
By \Cref{eqn:gibbs-expectation-ratio} as well as the definitions of $\tau$ and $\mu_j$,
\begin{align}
    \E_{\bu}\!\left[
        (\bv^{(\bu)})^\dagger
        \A{j}
        \bv^{(\bu)}
    \right]
    &=
    \E_{\bu}[W^{(\bu)}]\,
    \Tr(\A{j}\Gibbs{\bH})=
    \tau\mu_j.
\end{align}

We next prove the claimed second-moment bound.  Conditional on $\bu$, we
have
\begin{align}
    \E_{I^{(\bu)}\mid\bu}\!\left[
        (X_j^{(\bu)})^2
    \right]
    &=
    (W^{(\bu)})^2
    \E_{I^{(\bu)}\mid\bu}\!\left[
        \Zest_{\A{j}}\!\left(
            \bpsi^{(\bu)},
            I^{(\bu)}
        \right)^2
    \right].
\end{align}
By \Cref{lem:energy-estimator-moments} and
$\norm{\A{j}}\leq\om$,
\begin{align}
    \E_{I^{(\bu)}\mid\bu}\!\left[
        \Zest_{\A{j}}\!\left(
            \bpsi^{(\bu)},
            I^{(\bu)}
        \right)^2
    \right]
    &\leq
    \om^2.
\end{align}
Hence,
\begin{align}
    \E_{I^{(\bu)}\mid\bu}\!\left[
        (X_j^{(\bu)})^2
    \right]
    &\leq
    \om^2(W^{(\bu)})^2.
\end{align}
Taking an expectation over the Haar-random vector $\bu$ gives
\begin{align}
    \E_{\bu}\!\left[
        \E_{I^{(\bu)}\mid\bu}\!\left[
            (X_j^{(\bu)})^2
        \right]
    \right]
    &\leq
    \om^2
    \E_{\bu}\!\left[
        (W^{(\bu)})^2
    \right].
\end{align}
Applying \Cref{lem:gibbs-haar-moments} yields
\begin{align}
    \E_{\bu}\!\left[
        \E_{I^{(\bu)}\mid\bu}\!\left[
            (X_j^{(\bu)})^2
        \right]
    \right]
    &\leq
    2\om^2\tau^2.
\end{align}
Finally, since $\Gibbs{\bH}$ is a density matrix,
$|\mu_j|=|\Tr(\A{j}\Gibbs{\bH})|\leq\norm{\A{j}}\leq\om$.
\end{proof}

\subsection{Simultaneous Ratio Estimation via Median-of-Means}
\label{sec:robust-simultaneous-ratios}

The previous subsection showed that the random numerator
$X_j^{(\bu)}$ and common denominator $W^{(\bu)}$ satisfy the moment bounds
needed to control their empirical averages.  We now isolate the resulting
ratio-estimation argument from the Gibbs setting and prove it in a general
form, assuming only corresponding first- and second-moment bounds. Because the underlying random variables need not be bounded, we use a
standard median-of-means argument~\cite{NemirovskiYudin1983}.  We divide
independent samples into groups, form the numerator and denominator means
within each group, take their ratio, and then return the median of these
groupwise ratios.  The same denominator estimate is shared across all
$\mc$ observables, allowing all ratios to be estimated simultaneously.

\begin{lemma}[Simultaneous Ratio Estimation]
\label{lem:median-of-ratios}
Let $(W,X_1,\ldots,X_{\mc})$ be a jointly distributed real-valued random
vector, with $W>0$ almost surely.  We allow arbitrary dependence between
$W,X_1,\ldots,X_{\mc}$ within a single draw.  Suppose that for some
$\tau>0$, $\kappa\geq1$, $\om>0$, and
$\mu_j\in[-\om,\om]$, we have
$\E[W]=\tau$, $\E[W^2]\leq\kappa\tau^2$, and, for every $j\in[\mc]$,
$\E[X_j]=\tau\mu_j$ and
$\E[X_j^2]\leq\kappa\om^2\tau^2$.

Fix an accuracy parameter $\alpha\in(0,\om]$ and failure probability
$\del\in(0,1)$.  Draw $gb$ independent copies
$(W^{(r)},X_1^{(r)},\ldots,X_{\mc}^{(r)})$ for $r\in[gb]$ of the entire
random vector, where $g=\Theta(\log(2\mc/\del))$ and
$b=\Theta(\kappa\om^2/\alpha^2)$.  Partition the $gb$ draws into $g$ disjoint groups
$\cG_1,\ldots,\cG_g$, each of size $b$.  For every group $a\in[g]$,
define the denominator mean
\begin{align}
    \overline W_a
    &:=
    \frac{1}{b}
    \sum_{r\in\cG_a} W^{(r)},
\end{align}
and, for each $j\in[\mc]$, the corresponding numerator mean
\begin{align}
    \overline X_{a,j}
    &:=
    \frac{1}{b}
    \sum_{r\in\cG_a} X_j^{(r)}.
\end{align}
Using the same denominator mean $\overline W_a$ for all $j$, define
\begin{align}
    \widehat\mu_j
    &:=
    \operatorname{median}_{a\in[g]}
    \frac{\overline X_{a,j}}{\overline W_a}.
\end{align}
Then, with probability at least $1-\del$ over the $gb$ independent joint
draws,
\begin{align}
    \max_{j\in[\mc]}
    \abs{\widehat\mu_j-\mu_j}
    &\leq
    \alpha.
    \label{eqn:median-ratio-guarantee}
\end{align}
\end{lemma}

\begin{proof}
Fix $j\in[\mc]$.  For each group $a\in[g]$, define its ratio estimate
\begin{align}
    R_{a,j}
    &:=
    \frac{\overline X_{a,j}}{\overline W_a}.
\end{align}
By definition, the final estimator is the median of these $g$ independent
group estimates:
\begin{align}
    \widehat\mu_j
    &:=
    \operatorname{median}_{a\in[g]} R_{a,j}.
\end{align}
We first show that a single group estimate $R_{a,j}$ is
$\alpha$-accurate with constant probability, and then use the median to
amplify this constant success probability.

Fix a group $a\in[g]$.  Since the $b$ joint draws within the group are
independent,
\begin{align}
    \Var(\overline W_a)
    &=
    \frac{\Var(W)}{b}
    \leq
    \frac{\kappa\tau^2}{b},
\end{align}
where we used
$\Var(W)\leq\E[W^2]\leq\kappa\tau^2$.  Similarly,
\begin{align}
    \Var(\overline X_{a,j})
    &=
    \frac{\Var(X_j)}{b}
    \leq
    \frac{\kappa\om^2\tau^2}{b}.
\end{align}

Define the event
\begin{align}
    \cE_{a,j}
    &:=
    \left\{
        \abs{\overline W_a-\tau}
        \leq
        \frac{\alpha\tau}{4\om}
    \right\}
    \cap
    \left\{
        \abs{\overline X_{a,j}-\tau\mu_j}
        \leq
        \frac{\alpha\tau}{4}
    \right\}.
\end{align}
By Chebyshev's inequality,
\begin{align}
    \Prb\!\left[
        \abs{\overline W_a-\tau}
        >
        \frac{\alpha\tau}{4\om}
    \right]
    &\leq
    \frac{16\kappa\om^2}{b\alpha^2},
\end{align}
and
\begin{align}
    \Prb\!\left[
        \abs{\overline X_{a,j}-\tau\mu_j}
        >
        \frac{\alpha\tau}{4}
    \right]
    &\leq
    \frac{16\kappa\om^2}{b\alpha^2}.
\end{align}
Choosing the constant implicit in
$b=\Theta(\kappa\om^2/\alpha^2)$ sufficiently large makes each failure
probability at most $1/16$.  Hence, by a union bound,
\begin{align}
    \Prb[\cE_{a,j}]
    &\geq
    \frac{7}{8}.
    \label{eqn:good-group-probability}
\end{align}

We next show that $\cE_{a,j}$ guarantees that the ratio $R_{a,j}$ is
$\alpha$-accurate.  On $\cE_{a,j}$, since $\alpha\leq\om$,
\begin{align}
    \overline W_a
    &\geq
    \tau-\frac{\alpha\tau}{4\om}
    \geq
    \frac{3\tau}{4}.
\end{align}
Using also $|\mu_j|\leq\om$, we obtain
\begin{align}
    \abs{R_{a,j}-\mu_j}
    &=
    \frac{
        \abs{
            \overline X_{a,j}
            -
            \mu_j\overline W_a
        }
    }{
        \overline W_a
    }
    \\
    &=
    \frac{
        \abs{
            \overline X_{a,j}
            -
            \tau\mu_j
            +
            \mu_j(\tau-\overline W_a)
        }
    }{
        \overline W_a
    }
    \\
    &\leq
    \frac{
        \abs{\overline X_{a,j}-\tau\mu_j}
        +
        |\mu_j|\abs{\tau-\overline W_a}
    }{
        \overline W_a
    }
    \\
    &\leq
    \frac{
        \alpha\tau/4
        +
        \om\bigl(\alpha\tau/(4\om)\bigr)
    }{
        3\tau/4
    }
    =
    \frac{2\alpha}{3}
    \leq
    \alpha.
\end{align}
We have therefore shown that whenever $\cE_{a,j}$ occurs, the ratio estimate
$R_{a,j}$ is $\alpha$-accurate.  Equivalently,
\begin{align}
    \cE_{a,j}
    &\subseteq
    \left\{
        \abs{R_{a,j}-\mu_j}
        \leq
        \alpha
    \right\}.
\end{align}
By \Cref{eqn:good-group-probability}, this implies that
\begin{align}
    \Prb\!\left[
        \abs{R_{a,j}-\mu_j}
        \leq
        \alpha
    \right]
    &\geq
    \Prb[\cE_{a,j}]
    \geq
    \frac{7}{8}.
    \label{eqn:single-group-ratio-success}
\end{align}

We now repeat this construction across independent groups and take the
median of the resulting ratio estimates to obtain a high-probability
guarantee. For
$a\in[g]$, define the success indicator as
\begin{align}
    G_{a,j}
    &:=
    \mathbf{1}\!\left\{
        \abs{R_{a,j}-\mu_j}
        \leq
        \alpha
    \right\},
\end{align}
and the number of $\alpha$-accurate group estimates as 
\begin{align}
    S_j
    &:=
    \sum_{a=1}^{g}G_{a,j}.
\end{align}
By
\Cref{eqn:single-group-ratio-success},
$\Prb[G_{a,j}=1]\geq7/8$.  Moreover, for fixed $j$, the indicators
$G_{1,j},\ldots,G_{g,j}$ are independent because the groups are formed
from disjoint sets of independent joint draws.  Hence
\begin{align}
    \E[S_j]
    &=
    \sum_{a=1}^{g}\Prb[G_{a,j}=1]
    \geq
    \frac{7g}{8}.
\end{align}
If $S_j>g/2$, then a strict majority of the values
$R_{1,j},\ldots,R_{g,j}$ lie in
$[\mu_j-\alpha,\mu_j+\alpha]$.  Since $\widehat\mu_j$ is their median,
this immediately implies
$\abs{\widehat\mu_j-\mu_j}\leq\alpha$.

It therefore remains only to bound the probability that
$S_j\leq g/2$.  Since
\begin{align}
    \frac{g}{2}
    &=
    \left(1-\frac{3}{7}\right)\frac{7g}{8}
    \leq
    \left(1-\frac{3}{7}\right)\E[S_j],
\end{align}
a multiplicative Chernoff bound gives
\begin{align}
    \Prb\!\left[
        S_j\leq\frac{g}{2}
    \right]
    &\leq
    \exp\!\left(
        -\frac{9g}{112}
    \right).
    \label{eqn:median-ratio-chernoff}
\end{align}
Consequently,
\begin{align}
    \Prb\!\left[
        \abs{\widehat\mu_j-\mu_j}
        >
        \alpha
    \right]
    &\leq
    \exp\!\left(
        -\frac{9g}{112}
    \right).
\end{align}
Choosing the constant implicit in
$g=\Theta(\log(2\mc/\del))$ sufficiently large makes the right-hand side
at most $\del/\mc$.  A union bound over $j\in[\mc]$ therefore yields
\begin{align}
    \Prb\!\left[
        \max_{j\in[\mc]}
        \abs{\widehat\mu_j-\mu_j}
        >
        \alpha
    \right]
    &\leq
    \sum_{j=1}^{\mc}
    \Prb\!\left[
        \abs{\widehat\mu_j-\mu_j}
        >
        \alpha
    \right]
    \leq
    \del,
\end{align}
which proves \Cref{eqn:median-ratio-guarantee}.
\end{proof}

\subsection{Stability Under Lánczos Approximation}
\label{sec:gibbs-perturbation}

The preceding analysis assumed access to the exact filtered vector
$\bv^{(\bu)}=e^{-\bH/2}\bu$.  In the algorithm, however, this vector is
computed only approximately using the Lánczos method from
\Cref{sec:lanczos}.  We therefore need to show that the weighted ratio
estimator remains accurate even when each filtered vector is replaced by a
relative-error approximation. The following deterministic lemma quantifies this stability, allowing us to transfer the guarantees proved
for the exact filtered vectors to their Lánczos approximations.

\begin{lemma}[Stability under Relative Vector Error]
\label{lem:gibbs-relative-perturbation}
Let $\bu\sim\operatorname{Haar}(\C^{\dn})$ and
$\zeta\in(0,1/10]$.  Suppose that, for every $\bu$,
$\bv^{(\bu)}\neq\bzero$ and its approximation
$\widetilde{\bv}^{(\bu)}$ satisfies
\begin{align}
    \|
        \widetilde{\bv}^{(\bu)}-\bv^{(\bu)}
    \|_2
    &\leq
    \zeta\|\bv^{(\bu)}\|_2.
    \label{eqn:gibbs-relative-vector-error}
\end{align}
Define
$W^{(\bu)}:=\|\bv^{(\bu)}\|_2^2$,
$\widetilde W^{(\bu)}:=\|\widetilde{\bv}^{(\bu)}\|_2^2$,
$\tau:=\E_{\bu}[W^{(\bu)}]$, and
$\widetilde\tau:=\E_{\bu}[\widetilde W^{(\bu)}]$. Assume $\tau<\infty$, so $0<\tau<\infty$ and $0<\widetilde\tau<\infty$.
Finally, let $\beta:=2\zeta+\zeta^2$.
Then,
\begin{enumerate}
    \item \textbf{Pointwise stability:}
    For every $\bu$,
    \begin{align}
        \abs{
            \widetilde W^{(\bu)}-W^{(\bu)}
        }
        &\leq
        \beta W^{(\bu)}.
        \label{eqn:gibbs-weight-pointwise-perturbation}
    \end{align}
    Moreover, for every Hermitian $\bA$ with $\norm{\bA}\leq\om$,
    \begin{align}
        \abs{
            (\widetilde{\bv}^{(\bu)})^\dagger
            \bA
            \widetilde{\bv}^{(\bu)}
            -
            (\bv^{(\bu)})^\dagger
            \bA
            \bv^{(\bu)}
        }
        &\leq
        \beta\om W^{(\bu)}.
        \label{eqn:gibbs-observable-pointwise-perturbation}
    \end{align}

    \item \textbf{Second-moment stability:}
    If $\E_{\bu}[(W^{(\bu)})^2]\leq2\tau^2$, then
    \begin{align}
        \E_{\bu}\!\left[
            (\widetilde W^{(\bu)})^2
        \right]
        &\leq
        5\widetilde\tau^2.
        \label{eqn:gibbs-perturbed-weight-second-moment}
    \end{align}

    \item \textbf{Ratio stability:}
    For every Hermitian $\bA$ with $\norm{\bA}\leq\om$,
    \begin{align}
        \abs{
            \frac{
                \E_{\bu}\!\left[
                    (\widetilde{\bv}^{(\bu)})^\dagger
                    \bA
                    \widetilde{\bv}^{(\bu)}
                \right]
            }{
                \E_{\bu}[\widetilde W^{(\bu)}]
            }
            -
            \frac{
                \E_{\bu}\!\left[
                    (\bv^{(\bu)})^\dagger
                    \bA
                    \bv^{(\bu)}
                \right]
            }{
                \E_{\bu}[W^{(\bu)}]
            }
        }
        &\leq
        \frac{2\beta\om}{1-\beta}.
        \label{eqn:gibbs-population-ratio-bias}
    \end{align}
\end{enumerate}
\end{lemma}

\begin{proof}
Let
\begin{align}
    \be^{(\bu)}
    &:=
    \widetilde{\bv}^{(\bu)}-\bv^{(\bu)}.
\end{align}
By assumption,
$\norm{\be^{(\bu)}}_2\leq
\zeta\norm{\bv^{(\bu)}}_2$ for every $\bu$.

We first prove the pointwise stability bounds.  Let
$\be^{(\bu)}:=\widetilde{\bv}^{(\bu)}-\bv^{(\bu)}$, so that
$\widetilde{\bv}^{(\bu)}=\bv^{(\bu)}+\be^{(\bu)}$.  Expanding the squared
norm gives
\begin{align}
    \widetilde W^{(\bu)}-W^{(\bu)}
    &=
    \norm{\bv^{(\bu)}+\be^{(\bu)}}_2^2
    -
    \norm{\bv^{(\bu)}}_2^2 \\
    &=
    \left(
        \bv^{(\bu)}+\be^{(\bu)}
    \right)^\dagger
    \left(
        \bv^{(\bu)}+\be^{(\bu)}
    \right)
    -
    (\bv^{(\bu)})^\dagger\bv^{(\bu)}
    \\
    &=
    (\bv^{(\bu)})^\dagger\be^{(\bu)}
    +
    (\be^{(\bu)})^\dagger\bv^{(\bu)}
    +
    \norm{\be^{(\bu)}}_2^2\\
    &=
    2\Re\!\left[
        (\bv^{(\bu)})^\dagger\be^{(\bu)}
    \right]
    +
    \norm{\be^{(\bu)}}_2^2.
\end{align}
Taking absolute values and applying Cauchy--Schwarz,
\begin{align}
    \abs{
        \widetilde W^{(\bu)}-W^{(\bu)}
    }
    &\leq
    2
    \abs{
        (\bv^{(\bu)})^\dagger\be^{(\bu)}
    }
    +
    \norm{\be^{(\bu)}}_2^2\leq
    2\norm{\bv^{(\bu)}}_2
    \norm{\be^{(\bu)}}_2
    +
    \norm{\be^{(\bu)}}_2^2.
\end{align}
Finally, using the relative-error assumption
$\norm{\be^{(\bu)}}_2\leq
\zeta\norm{\bv^{(\bu)}}_2$, we obtain
\begin{align}
    \abs{
        \widetilde W^{(\bu)}-W^{(\bu)}
    }
    &\leq
    2\zeta\norm{\bv^{(\bu)}}_2^2
    +
    \zeta^2\norm{\bv^{(\bu)}}_2^2=
    (2\zeta+\zeta^2)W^{(\bu)}
    =
    \beta W^{(\bu)},
\end{align}
which proves
\Cref{eqn:gibbs-weight-pointwise-perturbation}.
For the observable term, expanding around
$\widetilde{\bv}^{(\bu)}
=
\bv^{(\bu)}+\be^{(\bu)}$ gives
\begin{align}
    &
    (\widetilde{\bv}^{(\bu)})^\dagger
    \bA
    \widetilde{\bv}^{(\bu)}
    -
    (\bv^{(\bu)})^\dagger
    \bA
    \bv^{(\bu)}=
    (\bv^{(\bu)})^\dagger
    \bA
    \be^{(\bu)}
    +
    (\be^{(\bu)})^\dagger
    \bA
    \bv^{(\bu)}
    +
    (\be^{(\bu)})^\dagger
    \bA
    \be^{(\bu)}.
\end{align}
Using $\norm{\bA}\leq\om$ and the relative-error bound,
\begin{align}
    &
    \abs{
        (\widetilde{\bv}^{(\bu)})^\dagger
        \bA
        \widetilde{\bv}^{(\bu)}
        -
        (\bv^{(\bu)})^\dagger
        \bA
        \bv^{(\bu)}
    }\leq
    2\om
    \norm{\bv^{(\bu)}}_2
    \norm{\be^{(\bu)}}_2
    +
    \om\norm{\be^{(\bu)}}_2^2\leq
    (2\zeta+\zeta^2)\om
    \norm{\bv^{(\bu)}}_2^2
    =
    \beta\om W^{(\bu)}.
\end{align}
This proves
\Cref{eqn:gibbs-observable-pointwise-perturbation}.

We next prove stability of the second-moment bound.  By the triangle and
reverse triangle inequalities,
\begin{align}
    (1-\zeta)\norm{\bv^{(\bu)}}_2
    &\leq
    \norm{\widetilde{\bv}^{(\bu)}}_2
    \leq
    (1+\zeta)\norm{\bv^{(\bu)}}_2.
\end{align}
Squaring gives
\begin{align}
    (1-\zeta)^2W^{(\bu)}
    &\leq
    \widetilde W^{(\bu)}
    \leq
    (1+\zeta)^2W^{(\bu)}.
    \label{eqn:gibbs-weight-relative-bounds}
\end{align}
Taking expectations in the lower bound yields
\begin{align}
    \widetilde\tau
    &=
    \E_{\bu}[\widetilde W^{(\bu)}]
    \geq
    (1-\zeta)^2\tau.
    \label{eqn:gibbs-perturbed-tau-lower-bound}
\end{align}
Meanwhile, leveraging the assumptions that $\zeta\leq1/10$ and $\E_{\bu}[(W^{(\bu)})^2]\leq2\tau^2$,
\begin{align}
    \E_{\bu}\!\left[
        (\widetilde W^{(\bu)})^2
    \right]
    &\leq
    (1+\zeta)^4
    \E_{\bu}\!\left[
        (W^{(\bu)})^2
    \right]
    \leq
    2(1+\zeta)^4\tau^2
    \leq
    2
    \left(
        \frac{1+\zeta}{1-\zeta}
    \right)^4
    \widetilde\tau^2
    \leq
    5\widetilde\tau^2.
\end{align}
This proves
\Cref{eqn:gibbs-perturbed-weight-second-moment}.

Finally, we prove stability of the ratio of expectations.  Let
\begin{align}
    a
    &:=
    \E_{\bu}\!\left[
        (\bv^{(\bu)})^\dagger
        \bA
        \bv^{(\bu)}
    \right]
\end{align}
denote the exact numerator, and let
\begin{align}
    \widetilde a
    &:=
    \E_{\bu}\!\left[
        (\widetilde{\bv}^{(\bu)})^\dagger
        \bA
        \widetilde{\bv}^{(\bu)}
    \right]
\end{align}
denote its approximate counterpart.  Their corresponding denominators are
$\tau=\E_{\bu}[W^{(\bu)}]$ and
$\widetilde\tau=\E_{\bu}[\widetilde W^{(\bu)}]$, respectively.  We first
bound the perturbations of these numerator and denominator terms.  Taking
expectation over $\bu$ in
\Cref{eqn:gibbs-weight-pointwise-perturbation} gives
\begin{align}
    \abs{\widetilde\tau-\tau}
    &=
    \abs{
        \E_{\bu}\!\left[
            \widetilde W^{(\bu)}-W^{(\bu)}
        \right]
    }
    \leq
    \E_{\bu}\!\left[
        \abs{\widetilde W^{(\bu)}-W^{(\bu)}}
    \right]
    \leq
    \beta\E_{\bu}[W^{(\bu)}]
    =
    \beta\tau.
    \label{eqn:gibbs-tau-perturbation}
\end{align}
In particular, $\widetilde\tau\geq(1-\beta)\tau$.  Similarly, taking
expectation in \Cref{eqn:gibbs-observable-pointwise-perturbation} yields
\begin{align}
    \abs{\widetilde a-a}
    &\leq
    \beta\om\E_{\bu}[W^{(\bu)}]
    =
    \beta\om\tau.
    \label{eqn:gibbs-numerator-perturbation}
\end{align}
We also have
$\abs{(\bv^{(\bu)})^\dagger\bA\bv^{(\bu)}}
\leq\norm{\bA}\norm{\bv^{(\bu)}}_2^2
\leq\om W^{(\bu)}$, and hence $|a|\leq\om\tau$ after taking expectation
over $\bu$.  Combining these bounds, we can compare the exact and
approximate ratios directly:
\begin{align}
    \abs{
        \frac{\widetilde a}{\widetilde\tau}
        -
        \frac{a}{\tau}
    }
    &=
    \abs{
        \frac{\widetilde a-a}{\widetilde\tau}
        +
        a
        \left(
            \frac{1}{\widetilde\tau}
            -
            \frac{1}{\tau}
        \right)
    }
    \leq
    \frac{\abs{\widetilde a-a}}{\widetilde\tau}
    +
    \frac{
        |a|\abs{\widetilde\tau-\tau}
    }{
        \tau\widetilde\tau
    }
    \leq
    \frac{\beta\om\tau}{(1-\beta)\tau}
    +
    \frac{
        \om\tau\cdot\beta\tau
    }{
        \tau(1-\beta)\tau
    }
    =
    \frac{2\beta\om}{1-\beta}.
\end{align}
This proves \Cref{eqn:gibbs-population-ratio-bias}.
\end{proof}

\subsection{Proof of the Simultaneous Gibbs Estimator}
\label{sec:gibbs-moment-proof}

We now assemble the preceding ingredients into the simultaneous estimator
of \Cref{alg:simultaneous-gibbs-estimator}, which gives the detailed
version of the procedure summarized in
\Cref{alg:intro-gibbs-estimator}.  Each Haar-random start is filtered
approximately using Lánczos, after which one sampled coordinate is reused
by the sampled-coordinate estimator for all observables.  The resulting
numerator and denominator samples are then combined using the
median-of-ratios construction from
\Cref{lem:median-of-ratios}.  The moment bounds from
\Cref{sec:unnormalized-gibbs-moments} and the Lánczos stability guarantee
from \Cref{sec:gibbs-perturbation} will allow us to prove the claimed
accuracy and runtime.
{\color{red}
\begin{algorithm}[t]
\caption{Simultaneous Gibbs Expectation Estimator}
\label{alg:simultaneous-gibbs-estimator}
\begin{algorithmic}[1]
\Require Hamiltonian
$\bH=\sum_{\ell=1}^{h}y_\ell\bB_\ell$,
observables $(\A{j})_{j=1}^{\mc}$,
accuracy $\alpha$, and failure probability $\del$

\State Set
$g\gets\Theta(\log(2\mc/\del))$,
$b\gets\Theta(\om^2/\alpha^2)$, and
$N\gets gb$
\State Set the Lánczos relative accuracy
$\zeta\gets c_0\alpha/\om$
and per-run failure probability
$\eta\gets\del/(2N)$,
where $c_0>0$ is a sufficiently small universal constant
\State Set
$L_{\bH}\gets
\log\!\bigl((2/\zeta)\sqrt{\dn/\eta}\bigr)$
and
\[
    q_{\bH}
    \gets
    \min\!\left\{
        \dn,\,
        \left\lceil
            C_{\mathrm L}
            \sqrt{
                \max\{1,B_{\bH},L_{\bH}\}
                L_{\bH}
            }
        \right\rceil
    \right\},
\]
where $C_{\mathrm L}$ is the constant from
\Cref{prop:gibbs-random-start-lanczos}

\For{$a=1,\ldots,g$}
    \For{$k=1,\ldots,b$}
        \State Draw
        $\bu^{(a,k)}\sim\operatorname{Haar}(\C^{\dn})$

        \State Using \Cref{prop:gibbs-random-start-lanczos}, compute
        $\widetilde{\bv}^{(a,k)}
        \approx e^{-\bH/2}\bu^{(a,k)}$
        with relative accuracy $\zeta$, failure probability $\eta$,
        and degree $q_{\bH}$

        \State Set
        $\widetilde W^{(a,k)}
        \gets\norm{\widetilde{\bv}^{(a,k)}}_2^2$
        and
        $\widetilde{\bpsi}^{(a,k)}
        \gets
        \widetilde{\bv}^{(a,k)}/
        \norm{\widetilde{\bv}^{(a,k)}}_2$

        \State Sample
        $I^{(a,k)}\in[\dn]$ according to
        $\Prb[I^{(a,k)}=i\mid\widetilde{\bpsi}^{(a,k)}]
        =\abs{[\widetilde{\bpsi}^{(a,k)}]_i}^2$

        \For{$j=1,\ldots,\mc$}
            \State Set
            $\widetilde X_j^{(a,k)}
            \gets
            \widetilde W^{(a,k)}
            \Zest_{\A{j}}(
                \widetilde{\bpsi}^{(a,k)},
                I^{(a,k)}
            )$
        \EndFor
    \EndFor

    \For{$j=1,\ldots,\mc$}
        \State Form the group ratio
        \[
            \widehat\mu_j^{(a)}
            \gets
            \frac{
                \sum_{k=1}^{b}\widetilde X_j^{(a,k)}
            }{
                \sum_{k=1}^{b}\widetilde W^{(a,k)}
            }
        \]
    \EndFor
\EndFor

\For{$j=1,\ldots,\mc$}
    \State Set
    $\widehat\mu_j
    \gets
    \operatorname{median}_{a\in[g]}
    \widehat\mu_j^{(a)}$
\EndFor

\State \Return
$\widehat\mu_1,\ldots,\widehat\mu_{\mc}$
\end{algorithmic}
\end{algorithm}
}

\begin{proof}[Proof of \Cref{thm:classical-gibbs-moments}]
We claim that \Cref{alg:simultaneous-gibbs-estimator} outputs estimates
$\widehat\mu_1,\ldots,\widehat\mu_{\mc}$ satisfying
\Cref{eqn:simultaneous-gibbs-moment-guarantee} within the running time
claimed in \Cref{eqn:gibbs-moment-runtime}.  We briefly recall the
procedure, then prove its correctness, and finally bound its running time.

\paragraph{Procedure.}
The estimator in \Cref{alg:simultaneous-gibbs-estimator} directly combines
the ingredients developed in the preceding subsections.  Let $g$ denote the number of
independent groups.  To obtain the desired simultaneous failure
probability over all $\mc$ observables, we choose
\begin{align}
    g
    &:=
    \Theta\!\left(
        \log\frac{2\mc}{\del}
    \right).
\end{align}
Let $b$ denote the number of independent samples within each group.  To
achieve statistical accuracy $\alpha/2$, we take
\begin{align}
    b
    &:=
    \Theta\!\left(
        \frac{\om^2}{\alpha^2}
    \right).
\end{align}
Thus the total number of Haar-random starts is $N:=gb$.
For each random start, Lánczos approximates the filtered vector
$e^{-\bH/2}\bu$ to relative accuracy
\begin{align}
    \zeta:=\frac{c_0\alpha}{\om},
\end{align}
with failure probability
\begin{align}
    \eta
    &:=
    \frac{\del}{2N}.
    \label{eqn:gibbs-lanczos-failure-probability}
\end{align}
The resulting approximate filtered vector is then used to form its weight
and the sampled coordinate estimates for all $\mc$ observables.  Finally,
the $N$ joint samples are combined using the median-of-ratios estimator of
\Cref{lem:median-of-ratios}.  The constant $c_0>0$ will be chosen
sufficiently small below.

\paragraph{Correctness Guarantee.}
For the analysis, index the $N$ random starts by $r\in[N]$, suppressing
their grouping into the $g$ median-of-ratios blocks.  Let
\begin{align}
    \bv^{(r)}
    &:=
    e^{-\bH/2}\bu^{(r)}
\end{align}
denote the exact filtered vector corresponding to the $r$-th random
start, and let $\widetilde{\bv}^{(r)}$ denote the Lánczos approximation
computed by the algorithm.  Define the corresponding Lánczos-success
event by
\begin{align}
    \cE_r
    &:=
    \left\{
        \norm{
            \widetilde{\bv}^{(r)}
            -
            \bv^{(r)}
        }_2
        \leq
        \zeta\norm{\bv^{(r)}}_2
    \right\}.
    \label{eqn:gibbs-lanczos-success-event}
\end{align}
By \Cref{prop:gibbs-random-start-lanczos},
$\Prb[\cE_r^c]\leq\eta$.  Let
\begin{align}
    \cE_N
    &:=
    \bigcap_{r=1}^{N}\cE_r
\end{align}
denote the event that all $N$ Lánczos computations succeed.  By a union
bound and the choice of $\eta$ in
\Cref{eqn:gibbs-lanczos-failure-probability},
\begin{align}
    \Prb[\cE_N^c]
    &=
    \Prb\!\left[
        \bigcup_{r=1}^{N}\cE_r^c
    \right]
    \leq
    \sum_{r=1}^{N}\Prb[\cE_r^c]
    \leq
    N\eta
    =
    \frac{\del}{2}.
    \label{eqn:all-gibbs-lanczos-success}
\end{align}

We cannot analyze the estimator simply by conditioning on $\cE_N$, since
doing so would change the distribution of the Haar-random starting
vectors.  Instead, following the pointwise stability analysis of
\Cref{sec:gibbs-perturbation}, we introduce a proof-only approximation
that satisfies the required relative-error condition for every
realization.  For each $r\in[N]$, define
\begin{align}
    \widetilde{\bv}_*^{(r)}
    &:=
    \begin{cases}
        \widetilde{\bv}^{(r)},
        & \text{on }\cE_r,\\
        \bv^{(r)},
        & \text{on }\cE_r^c.
    \end{cases}
    \label{eqn:pointwise-valid-lanczos-vector}
\end{align}
By construction,
\begin{align}
    \norm{
        \widetilde{\bv}_*^{(r)}
        -
        \bv^{(r)}
    }_2
    &\leq
    \zeta\norm{\bv^{(r)}}_2
    \label{eqn:pointwise-valid-relative-error}
\end{align}
with probability one.  These vectors therefore allow us to apply
\Cref{lem:gibbs-relative-perturbation} without conditioning on a
Lánczos-success event.

For the exact filtered vectors, recall that their squared norms provide the
common denominator samples in the Gibbs ratio estimator.  We therefore
define
\begin{align}
    W^{(r)}
    &:=
    \norm{\bv^{(r)}}_2^2.
\end{align}
For the pointwise-valid approximate vectors
$\widetilde{\bv}_*^{(r)}$, define the analogous weights
\begin{align}
    \widetilde W_*^{(r)}
    &:=
    \norm{\widetilde{\bv}_*^{(r)}}_2^2.
\end{align}
Their expectations give the corresponding normalization factors for the
exact and approximate ratio estimators.  We denote these by
\begin{align}
    \tau
    &:=
    \E[W^{(r)}],
\end{align}
and
\begin{align}
    \widetilde\tau_*
    &:=
    \E[\widetilde W_*^{(r)}].
\end{align}
By \Cref{lem:gibbs-haar-moments}, the exact weights satisfy
\begin{align}
    \E\!\left[
        (W^{(r)})^2
    \right]
    &\leq
    2\tau^2.
\end{align}
Since \Cref{eqn:pointwise-valid-relative-error} holds for every
realization, the second-moment stability guarantee of
\Cref{lem:gibbs-relative-perturbation} then gives
\begin{align}
    \E\!\left[
        (\widetilde W_*^{(r)})^2
    \right]
    &\leq
    5\widetilde\tau_*^2.
    \label{eqn:pointwise-valid-weight-moment}
\end{align}

We next identify the ratio estimated by these pointwise-valid filtered
vectors.  For each $j\in[\mc]$, let
\begin{align}
    \mu_j
    &:=
    \Tr(\A{j}\Gibbs{\bH})
\end{align}
denote the desired Gibbs expectation.  By
\Cref{eqn:gibbs-expectation-ratio},
\begin{align}
    \mu_j
    &=
    \frac{
        \E\!\left[
            (\bv^{(r)})^\dagger
            \A{j}
            \bv^{(r)}
        \right]
    }{
        \tau
    }.
\end{align}
The corresponding ratio for the pointwise-valid approximate vectors is
\begin{align}
    \widetilde\mu_j^*
    &:=
    \frac{
        \E\!\left[
            (\widetilde{\bv}_*^{(r)})^\dagger
            \A{j}
            \widetilde{\bv}_*^{(r)}
        \right]
    }{
        \widetilde\tau_*
    }.
    \label{eqn:pointwise-valid-population-ratio}
\end{align}
Applying the ratio-stability guarantee of
\Cref{lem:gibbs-relative-perturbation}, with
$\beta:=2\zeta+\zeta^2$, gives
\begin{align}
    \abs{
        \widetilde\mu_j^*-\mu_j
    }
    &\leq
    \frac{2\beta\om}{1-\beta}.
\end{align}
Since $\zeta=c_0\alpha/\om$ and $\alpha\leq\om$, choosing the universal
constant $c_0$ sufficiently small ensures that, for every $j\in[\mc]$,
\begin{align}
    \abs{
        \widetilde\mu_j^*-\mu_j
    }
    &\leq
    \frac{\alpha}{2}.
    \label{eqn:pointwise-valid-population-bias}
\end{align}

It remains to control the statistical estimation of
$\widetilde\mu_j^*$.  Normalize each pointwise-valid filtered vector as
\begin{align}
    \widetilde{\bpsi}_*^{(r)}
    &:=
    \frac{
        \widetilde{\bv}_*^{(r)}
    }{
        \norm{\widetilde{\bv}_*^{(r)}}_2
    },
\end{align}
sample a coordinate $I_*^{(r)}$ according to its squared amplitudes, and
define
\begin{align}
    \widetilde X_{*,j}^{(r)}
    &:=
    \widetilde W_*^{(r)}
    \Zest_{\A{j}}\!\left(
        \widetilde{\bpsi}_*^{(r)},
        I_*^{(r)}
    \right).
\end{align}
Conditional on $\widetilde{\bv}_*^{(r)}$, the moment identities of the
sampled coordinate estimator (\Cref{lem:energy-estimator-moments}) give
\begin{align}
    \E_{I_*^{(r)}\mid\widetilde{\bv}_*^{(r)}}\!\left[
        \widetilde X_{*,j}^{(r)}
    \right]
    &=
    (\widetilde{\bv}_*^{(r)})^\dagger
    \A{j}
    \widetilde{\bv}_*^{(r)}.
\end{align}
Consequently,
\begin{align}
    \E\!\left[
        \widetilde X_{*,j}^{(r)}
    \right]
    &=
    \widetilde\tau_*
    \widetilde\mu_j^*.
\end{align}
Similarly, using
\Cref{eqn:pointwise-valid-weight-moment},
\begin{align}
    \E\!\left[
        (\widetilde X_{*,j}^{(r)})^2
    \right]
    &\leq
    \om^2
    \E\!\left[
        (\widetilde W_*^{(r)})^2
    \right]
    \leq
    5\om^2\widetilde\tau_*^2.
\end{align}
Moreover,
$|\widetilde\mu_j^*|\leq\om$.

Thus the joint samples
$(\widetilde W_*^{(r)},
\widetilde X_{*,1}^{(r)},\ldots,
\widetilde X_{*,\mc}^{(r)})$
satisfy the hypotheses of
\Cref{lem:median-of-ratios} with $\kappa=5$.  Applying that lemma with
statistical accuracy $\alpha/2$ and failure probability $\del/2$, the
choices of $g$ and $b$ in
\Cref{alg:simultaneous-gibbs-estimator} imply that the corresponding
median-of-ratios estimates $\widehat\mu_j^*$ satisfy
\begin{align}
    \Prb\!\left[
        \max_{j\in[\mc]}
        \abs{
            \widehat\mu_j^*
            -
            \widetilde\mu_j^*
        }
        \leq
        \frac{\alpha}{2}
    \right]
    &\geq
    1-\frac{\del}{2}.
    \label{eqn:pointwise-valid-statistical-success}
\end{align}

Finally, couple $I_*^{(r)}$ with the coordinate $I^{(r)}$ used by
\Cref{alg:simultaneous-gibbs-estimator} by using the same independent
uniform random seed to sample from their respective squared-amplitude
distributions.  On the event $\cE_N$, we have
$\widetilde{\bv}_*^{(r)}=\widetilde{\bv}^{(r)}$ for every $r\in[N]$.
Hence their weights, sampled coordinates, and sampled coordinate estimates
coincide such that, for every $j\in[\mc]$, $\widehat\mu_j=\widehat\mu_j^*$.
On the intersection of $\cE_N$ and the statistical-success event in
\Cref{eqn:pointwise-valid-statistical-success}, we therefore have,
simultaneously for all $j\in[\mc]$,
\begin{align}
    \abs{
        \widehat\mu_j-\mu_j
    }
    &=
    \abs{
        \widehat\mu_j^*-\mu_j
    }\leq
    \abs{
        \widehat\mu_j^*
        -
        \widetilde\mu_j^*
    }
    +
    \abs{
        \widetilde\mu_j^*
        -
        \mu_j
    }\leq
    \frac{\alpha}{2}
    +
    \frac{\alpha}{2}
    =
    \alpha.
\end{align}
Combining
\Cref{eqn:all-gibbs-lanczos-success,eqn:pointwise-valid-statistical-success}
with a union bound gives
\begin{align}
    \Prb\!\left[
        \max_{j\in[\mc]}
        \abs{
            \widehat\mu_j-\mu_j
        }
        \leq
        \alpha
    \right]
    &\geq
    1-\del,
\end{align}
which proves
\Cref{eqn:simultaneous-gibbs-moment-guarantee}.

\paragraph{Runtime.}
The algorithm uses
\begin{align}
    N
    &=
    gb
    =
    O\!\left(
        \frac{\om^2}{\alpha^2}
        \log\frac{2\mc}{\del}
    \right)
\end{align}
independent random starts.  Sampling and explicitly storing one
Haar-random vector costs $O(\dn)$.  Moreover, since
$\bH=\sum_{\ell=1}^{h}y_\ell\bB_\ell$ and each $\bB_\ell$ is
$s_{\bB}$-row-sparse, multiplying a vector by $\bH$ costs
$O(\dn h s_{\bB})$.

By \Cref{prop:gibbs-random-start-lanczos}, using relative accuracy
$\zeta=c_0\alpha/\om$ and failure probability
$\eta=\del/(2N)$ requires Lánczos degree
\begin{align}
    q_{\bH}
    &=
    \widetilde O\!\left(
        \min\!\left\{
            \dn,
            \sqrt{1+B_{\bH}}
        \right\}
    \right).
\end{align}
Thus, computing one approximate filtered vector has runtime
\begin{align}
    O\!\left(
        q_{\bH}\dn(hs_{\bB}+1)
        +
        q_{\bH}^2
    \right)
    &=
    O\!\left(
        q_{\bH}\dn(hs_{\bB}+1)
    \right),
\end{align}
where we used $q_{\bH}\leq\dn$.

For each filtered vector, constructing its squared-amplitude sampling
distribution costs $O(\dn)$, and evaluating the sampled-coordinate
estimator for all $\mc$ observables at the single sampled coordinate costs
$O(\mc s_{\bA})$.  Hence the cost per random start is
\begin{align}
    \widetilde O\!\left(
        \dn(hs_{\bB}+1)
        \min\!\left\{
            \dn,
            \sqrt{1+B_{\bH}}
        \right\}
        +
        \mc s_{\bA}
    \right).
\end{align}
Multiplying by the $N$ random starts gives the total running time
\begin{align}
    \widetilde O\!\left(
        \frac{\om^2}{\alpha^2}
        \left[
            \dn(hs_{\bB}+1)
            \min\!\left\{
                \dn,
                \sqrt{1+B_{\bH}}
            \right\}
            +
            \mc s_{\bA}
        \right]
    \right),
\end{align}
which proves \Cref{eqn:gibbs-moment-runtime}.
\end{proof}

\section{Oracle-Based SDP Solving via Simultaneous Gibbs Estimation}
\label{sec:direct-solver}

The previous section showed how to estimate the Gibbs expectations of many
sparse observables simultaneously from a shared collection of filtered
random vectors.  We now use this primitive to classically implement the
Gibbs-state oracle interfaces appearing in the sparse-access SDP solver of
van Apeldoorn and Gily\'en~\cite{AG19}.  This gives our first end-to-end
classical sparse-oracle SDP solver with sublinear dependence on the product
of the matrix dimension $n$ and number of constraints $m$.

The outer convergence argument is unchanged from the matrix
multiplicative-weights framework of van Apeldoorn and
Gily\'en~\cite{AG19}.  Our task is therefore to verify that simultaneous
Gibbs expectation estimation supplies the oracle outputs required by that
framework and then account for the cost of implementing these oracles
classically.

\subsection{Classical Gibbs-State Constraint Search}
\label{sec:classical-gibbs-or}

We first reduce violated-constraint search directly to the simultaneous
Gibbs expectation estimator of \Cref{thm:classical-gibbs-moments}.  Given
constraint thresholds $b_1,\ldots,b_{\mc}$, it suffices to estimate each
quantity $\Tr(\A{j}\Gibbs{\bH})$ to sufficiently small additive error and
then compare the resulting estimates with the corresponding thresholds.
Thus, once the simultaneous estimates are available, the remaining search
is purely deterministic.

This is the classical analogue of the role played by the Quantum OR in
the quantum SDP solver.  There, the Quantum OR is used to determine
whether one of many tests accepts with high probability while reusing the
same input state across all tests~\cite{HLM17,BKL19}.  In the formulation
of Brand\~ao et al.~\cite[Lemma~2]{BKL19}, given a state $\brho$ and
projectors $\mathbf{\Pi}_1,\ldots,\mathbf{\Pi}_{\mc}$, one distinguishes
between the case in which some test has large acceptance probability,
\begin{align}
    \exists j\in[\mc]
    \qquad \text{s.t.} \qquad
    \Tr(\brho\mathbf{\Pi}_j)
    &\geq
    1-\eta,
\end{align}
and the case in which the average acceptance probability is small,
\begin{align}
    \frac{1}{\mc}
    \sum_{j=1}^{\mc}
    \Tr(\brho\mathbf{\Pi}_j)
    &\leq
    \nu.
\end{align}
Here a sufficient constant-gap promise is
$0<\eta\leq1/2$ and $0\leq\nu\leq(1-\eta)^2/(24\mc)$,
with a sufficiently small constant implementation error in the
the Quantum OR procedure~\cite[Lemma~5]{AG19}.
Our classical routine does not reproduce this general black-box
primitive.  Rather, in the Gibbs-state setting relevant to SDP solving,
\Cref{thm:classical-gibbs-moments} gives additive estimates of all
constraint expectations simultaneously, from which the required
constraint search follows immediately.

\begin{corollary}[Classical Gibbs-State Constraint Search]
\label{cor:gibbs-or}
Under the assumptions of \Cref{thm:classical-gibbs-moments}, let
$b_1,\ldots,b_{\mc}\in\R$, let $\alpha\in(0,\om]$, and let
$\del\in(0,1)$. There is a classical algorithm that, with probability at least
$1-\del$, either returns an index $j\in[\mc]$ satisfying
\begin{align}
    \Tr\!\left(
        \A{j}\Gibbs{\bH}
    \right)
    &>
    b_j+\alpha,
\end{align}
or certifies that, for every $j\in[\mc]$,
\begin{align}
    \Tr\!\left(
        \A{j}\Gibbs{\bH}
    \right)
    &\leq
    b_j+2\alpha.
\end{align}
Its sparse-oracle running time is the same, up to constant factors, as
that of \Cref{thm:classical-gibbs-moments}.
\end{corollary}

\begin{proof}
Apply \Cref{thm:classical-gibbs-moments} with accuracy $\alpha/2$ to
obtain estimates $\widehat\mu_1,\ldots,\widehat\mu_{\mc}$.  With
probability at least $1-\del$, these estimates satisfy, simultaneously for every $j\in[\mc]$,
\begin{align}
    \abs{
        \widehat\mu_j
        -
        \Tr\!\left(
            \A{j}\Gibbs{\bH}
        \right)
    }
    &\leq
    \frac{\alpha}{2}.
    \label{eqn:gibbs-or-simultaneous-accuracy}
\end{align}
Given these estimates, return any $j\in[\mc]$ satisfying
\begin{align}
    \widehat\mu_j
    &>
    b_j+\frac{3\alpha}{2}.
    \label{eqn:gibbs-or-threshold-rule}
\end{align}
If such an index is returned, then on the event in
\Cref{eqn:gibbs-or-simultaneous-accuracy},
\begin{align}
    \Tr\!\left(
        \A{j}\Gibbs{\bH}
    \right)
    &\geq
    \widehat\mu_j-\frac{\alpha}{2}
    >
    b_j+\alpha,
\end{align}
so the returned constraint is genuinely violated by the claimed margin. If no index is returned, then
$\widehat\mu_j\leq b_j+3\alpha/2$ for every $j\in[\mc]$.  Again using
\Cref{eqn:gibbs-or-simultaneous-accuracy},
\begin{align}
    \Tr\!\left(
        \A{j}\Gibbs{\bH}
    \right)
    &\leq
    \widehat\mu_j+\frac{\alpha}{2}
    \leq
    b_j+2\alpha
\end{align}
for every $j\in[\mc]$.  Thus the approximate-feasibility certificate is
valid.  Both conclusions hold with probability at least $1-\del$.
\end{proof}

The reduction is immediate once the simultaneous Gibbs
expectations have been estimated. The same collection of filtered random
vectors supplies all $\mc$ expectation estimates, after which
constraint search requires only a threshold scan.  This preserves the
sample-reuse mechanism underlying the quantum Fast-OR application, since
the expensive Hamiltonian-dependent filtering is performed once per
random start and reused across all constraints. The classical procedure does not, however, reproduce the quantum
square-root search over the constraints.  Each filtered random vector
still requires explicitly reading one sparse row of every observable,
giving an $O(\mc s_{\bA})$ readout cost.  Thus, what is dequantized here is
the reuse of the Gibbs-dependent computation across all constraints,
rather than the Grover-type search advantage itself.

\subsection{Direct Gibbs--MMWU Solver}
\label{sec:direct-gibbs-mmwu}

We now use the simultaneous Gibbs estimator to obtain a complete
classical implementation of the sparse-access SDP solver of
van Apeldoorn and Gilyén~\cite{AG19}.  The outer MMWU updates and
convergence analysis are unchanged.  The only modification is how the
algorithm extracts the expectation values required from the Gibbs state
associated with the current Hamiltonian.

As in \cite{AG19}, fix a candidate objective value $g$ and encode the
objective requirement as the additional constraint
$\A{0}:=-\bC$ with threshold $b_0:=-g$.  The MMWU procedure uses the
accuracy parameter
\begin{align}
\theta
&:=
\frac{\eps}{6\Rx\Ry}.
\end{align}
For a density matrix $\brho$, define
\begin{align}
P_\theta(\brho)
&:=
\left\{
\widetilde\by\in\R_{\geq0}^{\mc+1}:
\mathbf b^\top\widetilde\by\leq0,\;
\sum_{j=0}^{\mc}
\widetilde y_j
\Tr(\A{j}\brho)
\geq-\theta,\;
\widetilde y_0=\frac{1}{2\Ry},\;
\norm{\widetilde\by}_1\leq1
\right\}.
\label{eqn:direct-ak-polytope}
\end{align}
We write $P_0(\brho)$ for the corresponding exact set, obtained by
replacing the lower bound $-\theta$ by $0$.

On each round, the Arora--Kale update requires a subroutine that either
returns a constant-sparse vector
$\widetilde\by\in P_\theta(\brho)$ or certifies that
$P_0(\brho)$ is empty.  Following \cite{AG19}, this subroutine is called
the $\theta$-oracle.  The geometric construction of
van Apeldoorn et al.~\cite[Lemma~16]{AGGW17} shows that it can be
implemented from sufficiently accurate additive estimates of the
expectations $\Tr(\A{j}\brho)$.

Our simultaneous Gibbs estimator supplies these estimates classically.
On round $t$, the current MMWU coefficients determine the Hamiltonian
\begin{align}
\bH_t
&=
\sum_{j=0}^{\mc}
y_j^{(t)}\A{j}.
\end{align}
Rather than explicitly constructing the Gibbs state
$\Gibbs{\bH_t}=e^{-\bH_t}/\Tr(e^{-\bH_t})$, we apply
\Cref{alg:simultaneous-gibbs-estimator} directly to the sparse
representation of $\bH_t$ to estimate all quantities
$\Tr(\A{j}\Gibbs{\bH_t})$.  The geometric procedure above then converts
these estimates into the constant-sparse vector required by the
$\theta$-oracle and hence by the next MMWU update.  The resulting
classical implementation of the first MMWU phase is summarized in
\Cref{alg:direct-gibbs-mmwu}.

\begin{algorithm}[t]
\small
\caption{Direct classical Gibbs--MMWU solver}
\label{alg:direct-gibbs-mmwu}
\begin{algorithmic}[1]
\Require Normalized sparse SDP instance, target value $g$, accuracy
$\eps$, and failure probability $\del$
\Ensure Either the conclusion $\operatorname{OPT}>g$, or an
$\eps$-feasible dual certificate $\by_{\mathrm{out}}$ with
$\mathbf b^\top\by_{\mathrm{out}}\leq g+\eps$

\State $\A{0}\gets-\bC$, $b_0\gets-g$
\State $\theta\gets\eps/(6\Rx\Ry)$,
$\TT\gets\lceil\log(2\dn)/\theta^2\rceil$, and
$\by^{(1)}\gets\bzero$

\For{$t=1,\ldots,\TT$}
\State
$\bH_t\gets\sum_{j=0}^{\mc}y_j^{(t)}\A{j}$
\State By \Cref{alg:simultaneous-gibbs-estimator}, for the constant $c>0$ required by
\cite[Lemma~16]{AGGW17}, compute
$(\widehat\mu_{t,j})_{j=0}^{\mc}$ such that
\begin{align*}
\Prb\left[
    \max_{0\leq j\leq\mc}
    \abs{
        \widehat\mu_{t,j}
        -
        \Tr(\A{j}\Gibbs{\bH_t})
    }
    \leq
    c\theta
    \,\middle|\,
    \bH_t
\right]
&\geq
1-\frac{\del}{2\TT}.
\end{align*}
\State Using $(\widehat\mu_{t,j})_{j=0}^{\mc}$ and
\cite[Lemma~16]{AGGW17}, compute either $\begin{cases}
    \widetilde\by^{(t)}\in
P_\theta(\Gibbs{\bH_t}), \text{ s.t. } \|\widetilde\by^{(t)}\|_0\leq3 \\
\text{a certificate that } P_0(\Gibbs{\bH_t})=\varnothing
\end{cases}$
\If{$P_0(\Gibbs{\bH_t})=\varnothing$}
    \State \Return $\operatorname{OPT}>g$
\EndIf

\State
$\by^{(t+1)}
\gets
\by^{(t)}
+
\theta\widetilde\by^{(t)}$
\EndFor

\State
$\displaystyle
\by_{\mathrm{out}}
\gets
\frac{2\Ry}{\TT\theta}
\by^{(\TT+1)}
+
\frac{\eps}{\Rx}\be_1
-\be_0$

\State Discard coordinate $0$ of $\by_{\mathrm{out}}$
\State \Return $\by_{\mathrm{out}}$
\end{algorithmic}
\end{algorithm}

The final normalization is chosen using the actual integer number of rounds $\TT$. Since
$\by^{(\TT+1)}=\theta\sum_{t=1}^{\TT}\widetilde\by^{(t)}$, we have
$(2\Ry/(\TT\theta))\by^{(\TT+1)}
=(2\Ry/\TT)\sum_{t=1}^{\TT}\widetilde\by^{(t)}$. Moreover, every oracle output satisfies
$\widetilde y_0^{(t)}=1/(2\Ry)$, so coordinate $0$ of this normalized average is exactly one, and subtracting $\be_0$ therefore removes the artificial objective coordinate exactly. To verify dual feasibility, define
$\bM_t:=\sum_{j=0}^{\mc}\widetilde y_j^{(t)}\A{j}$. The standard MMWU guarantee, together with the oracle condition
$\Tr(\bM_t\Gibbs{\bH_t})\geq-\theta$, gives
\begin{align}
\lambda_{\min}\!\left(\frac1\TT\sum_{t=1}^{\TT}\bM_t\right)
&\geq-2\theta-\frac{\log\dn}{\TT\theta}
\geq-3\theta.
\end{align}
After multiplying by $2\Ry$, the resulting matrix therefore has minimum eigenvalue at least $-6\Ry\theta$. The additional term $(\eps/\Rx)\be_1$ contributes $(\eps/\Rx)\bI$, since $\A{1}=\bI$, and hence, using $\theta=\eps/(6\Rx\Ry)$, the final dual slack is at least
$(-6\Ry\theta+\eps/\Rx)\bI=\bzero$. Finally, every oracle output satisfies
$\mathbf b^\top\widetilde\by^{(t)}\leq0$, so the normalized average has nonpositive $\mathbf b$-cost before the final correction. Since $b_0=-g$, subtracting $\be_0$ contributes $g$, while $(\eps/\Rx)\be_1$ contributes $\eps$ because $b_1=\Rx$. Therefore
$\mathbf b^\top\by_{\mathrm{out}}\leq g+\eps$.

The complete SDP solver of \cite{AG19} contains a second MMWU phase for
constructing a primal certificate when the target value $g$ is
sufficiently small.  In this phase, the required information from the
current Gibbs state is whether any SDP constraint is violated by more
than the prescribed tolerance.  We implement this search using
\Cref{cor:gibbs-or}.  
In this phase we use the lifted matrices
$\widehat{\bA}_j:=\operatorname{diag}(\A{j},0)$ and thresholds $b_j/\Rx$
for $0\leq j\leq\mc$, including $\A{0}=-\bC$ and $b_0=-g$.
Set $\alpha_{\mathrm P}:=\min\{1/4,\eps/(4\Rx\Ry)\}$ and
$\eta_{\mathrm P}:=\alpha_{\mathrm P}/2$.
Apply \Cref{cor:gibbs-or} with accuracy $\alpha_{\mathrm P}$, so a
returned index has violation strictly greater than $\alpha_{\mathrm P}$,
while no returned index certifies violation at most $2\alpha_{\mathrm P}$.
Use at most
$T_{\mathrm P}:=\lceil4\log(2(\dn+1))/\alpha_{\mathrm P}^2\rceil$
rounds, starting from zero coefficients.
 If a violated constraint $j_t$ is found, its
coefficient is updated according to
\begin{align}
\by^{(t+1)}
&=
\by^{(t)}
+
\eta_{\mathrm P}\be_{j_t}.
\end{align}
If no sufficiently violated constraint exists, the current Gibbs state
satisfies all constraints to the required tolerance and determines the
primal certificate.  The scalar normalization needed to recover the
original primal matrix is obtained by including the projector onto the
original primal block among the simultaneously estimated observables.
This adds only one norm-one, one-row-sparse observable and therefore
does not affect the asymptotic running time.
More precisely, let $q:=\Tr(\mathbf P\brho)$ for
$\mathbf P=\operatorname{diag}(\bI_{\dn},0)$ and let $\boldsymbol{\sigma}$ be the
normalized leading block of the terminal Gibbs state $\brho$.
The exact recovered matrix is $\bX_0=\Rx q\boldsymbol{\sigma}$.
Obtain $\widehat q$ with
$|\widehat q-q|\leq\beta:=\min\{1/2,\eps/(2\Rx\Ry)\}$ and return
$z:=\Rx\max\{0,\min\{1,\widehat q\}\}$ and $\bX=z\boldsymbol{\sigma}$.
Then $\|\bX-\bX_0\|_1\leq\Rx\beta\leq\eps/(2\Ry)$.
Since the feasibility phase incurs error at most
$2\Rx\alpha_{\mathrm P}\leq\eps/(2\Ry)$, the returned matrix satisfies
$\inner{\A{j}}{\bX}\leq b_j+\eps/\Ry$
for every $j\in[\mc]$ and
$\inner{\bC}{\bX}\geq g-\eps/\Ry\geq g-\eps$.

The classical and quantum implementations differ only in how the
necessary constraint information is extracted from each Gibbs state.
The quantum solver obtains a square-root dependence on the number of
constraints by using quantum minimum finding in the Arora--Kale phase and 
Fast Quantum OR in the primal phase. Our
classical implementation instead estimates all constraint expectations
from the same collection of filtered random vectors and then scans the
resulting values explicitly.  Thus, the Hamiltonian-dependent filtering
is still reused across all constraints, while the Grover-type search
advantage itself is not reproduced.

The sparsity of the updates is important for the runtime.  In the first
phase, every $\theta$-oracle output has constant support, while in the
second phase only the coefficient of the single returned constraint is
updated.  Hence each round adds only constantly many input matrices to
the accumulated Hamiltonian.  After $t$ rounds, it can therefore be
written as
\begin{align}
\bH_t
&=
\sum_{\ell=1}^{h_t}
y_{t,\ell}\bB_{t,\ell},
\qquad \text{where} \quad 
h_t=
O(t),
\label{eqn:direct-round-hamiltonian}
\end{align}
with total coefficient mass
$\norm{\by_t}_1=\widetilde O(\gam)$.  Every Gibbs-estimation call
therefore falls within the structured Hamiltonian regime covered by
\Cref{thm:classical-gibbs-moments}.

\begin{theorem}[Direct Classical Gibbs--MMWU Solver]
\label{thm:direct-sdp}
Consider an SDP of the form \Cref{eqn:intro-standard-sdp-primal,eqn:intro-standard-sdp-dual} under the
sparse-access model of \Cref{def:sparse-access}.  Assume strong duality,
supplied primal and dual radii ${\Rx,\Ry\geq1}$, and the normalization
$\A{1}=\bI$ and $b_1=\Rx$.  Suppose further that the input matrices have
operator norm at most one and row sparsity $\sr$.  Let
$\gam:=\Rx\Ry/\eps\geq1$. For every target $g\in[-\Rx,\Rx]$ and failure probability
$\del\in(0,1)$, there is an exact-arithmetic sparse-oracle algorithm with
the following guarantees.

\begin{enumerate}
\item \textbf{Primal certificate.}
If $\operatorname{OPT}\geq g+\eps$, the algorithm returns coefficients
$\by'=(y'_0,\ldots,y'_{\mc})\in\R_{\geq0}^{\mc+1}$ and
$z\in[0,\Rx]$ defining the implicit primal matrix
\begin{align}
\bX
&:=
z\cdot 
\frac{
\exp\left(
y'_0\bC
-
\sum_{j=1}^{\mc}y'_j\A{j}
\right)
}{
\Tr\left(
\exp\left(
y'_0\bC
-
\sum_{j=1}^{\mc}y'_j\A{j}
\right)
\right)
},
\label{eqn:direct-implicit-primal}
\end{align}
which satisfies $\inner{\bC}{\bX}\geq g-\eps$ and $\inner{\A{j}}{\bX} \leq b_j+\eps/\Ry$ for every $j\in[\mc]$.
\item \textbf{Dual certificate.}
If $\operatorname{OPT}\leq g-\eps$, the algorithm returns a dual
vector $\by\in\R_{\geq0}^{\mc}$ satisfying
\begin{align}
\sum_{j=1}^{\mc}
y_j\A{j}
-
\bC
&\succeq
-\eps\bI,
 \quad \text{and} \quad
\mathbf b^\top\by
\leq
g+\eps.
\label{eqn:direct-dual-guarantee}
\end{align}
\end{enumerate}
Either output is permitted when
$\operatorname{OPT}\in[g-\eps,g+\eps]$.  With probability at least
$1-\del$, the returned certificate is valid.  The coefficient
representation has support $\widetilde O(\gam^2)$, and the
running time is $\widetilde O\left(
\dn\sr\gam^{13/2}
+
\mc\sr\gam^4
\right)$.
\end{theorem}

\begin{proof}
The convergence and certificate guarantees are those of the MMWU solver
of van Apeldoorn and Gilyén
\cite[Sections~2.1--2.2 and Theorem~8]{AG19}.  Thus it suffices to show
that our simultaneous Gibbs estimator implements the two Gibbs-state
subroutines required by that solver, and then to bound the cost of doing
so classically.

For the certificate-selection rule, run
\Cref{alg:direct-gibbs-mmwu} at target $g$ with internal accuracy $\eps/2$.
If it returns a dual certificate, return that certificate.  Otherwise it
certifies $\operatorname{OPT}>g$, and we run the primal phase just
specified.  The completed first phase gives a feasible dual of value at
most $g+\eps/2$, so it cannot return a dual when
$\operatorname{OPT}\geq g+\eps$.  Conversely, when
$\operatorname{OPT}\leq g-\eps$, it cannot correctly report
$\operatorname{OPT}>g$.

For the Arora--Kale phase, the required $\theta$-oracle must either
return a constant-sparse vector in
$P_\theta(\Gibbs{\bH_t})$ or certify that
$P_0(\Gibbs{\bH_t})$ is empty.  On round $t$, apply
\Cref{thm:classical-gibbs-moments} with additive accuracy
$c\theta$, where $c>0$ is the sufficiently small universal constant
required by \cite[Lemma~16]{AGGW17}.  This produces simultaneous
approximations to all expectations
$\Tr(\A{j}\Gibbs{\bH_t})$.  The geometric construction of
\cite[Lemma~16]{AGGW17} then converts these estimates into the required
constant-sparse oracle output, or certifies that the exact set is empty.
Hence the classical implementation satisfies the oracle assumptions used
in the Arora--Kale analysis.
To make the estimation tolerance explicit, take $c\leq1/2$ and
let $\widehat P$ be the set defining $P_\theta$ with its payoff condition
replaced by
$\sum_{j=0}^{\mc}\widetilde y_j\widehat\mu_{t,j}\geq-\theta/2$.
On the simultaneous-estimation event,
$P_0(\Gibbs{\bH_t})\subseteq\widehat P
\subseteq P_\theta(\Gibbs{\bH_t})$, because
$\|\widetilde\by\|_1\leq1$.
After fixing $\widetilde y_0=1/(2\Ry)$, the geometric routine applies
to the two remaining payoff inequalities and mass bound
$\sum_{j=1}^{\mc}\widetilde y_j\leq1-1/(2\Ry)$.
It uses $O(\mc)$ classical time and returns at most two nonzero
coordinates among $1,\ldots,\mc$, or certifies $\widehat P$ empty.
Thus its output has support at most three, and an empty-set certificate
implies $P_0(\Gibbs{\bH_t})=\varnothing$.

In the primal phase, the required subroutine instead asks whether the
current Gibbs state significantly violates any SDP constraint.
\Cref{cor:gibbs-or} provides exactly this guarantee. Specifically, it either returns a
genuinely violated constraint or certifies that every constraint is
satisfied to the prescribed additive tolerance.  Including the projector
onto the original primal block among the simultaneously estimated
observables also provides the scalar normalization needed to recover the
implicit primal matrix.  Thus the oracle assumptions of the primal phase
are satisfied as well.
Indeed, if the target is feasible, let $\brho^\star$ be a feasible
lifted density matrix.  If all $T_{\mathrm P}$ rounds returned indices,
then every selected constraint would give
$\langle\widehat{\bA}_{j_t},\brho_t-\brho^\star\rangle
>\alpha_{\mathrm P}$.  In contrast, the standard MMWU regret bound gives
\begin{align}
\sum_{t=1}^{T_{\mathrm P}}
\langle\widehat{\bA}_{j_t},\brho_t-\brho^\star\rangle
&\leq\frac{\log(\dn+1)}{\eta_{\mathrm P}}
+\eta_{\mathrm P}T_{\mathrm P}
<\alpha_{\mathrm P}T_{\mathrm P},
\end{align}
a contradiction.  Therefore the primal phase terminates with the
certificate whose scalar-approximation error was bounded above.

The two phases together make $\widetilde O(\gam^2)$ adaptive
Gibbs-estimation calls.  Assign failure probabilities to the individual
calls whose sum is at most $\del$.  Conditional on the history preceding
any call, \Cref{thm:classical-gibbs-moments} fails with probability at
most the assigned budget.  A union bound therefore implies that all
oracle calls are simultaneously correct with probability at least
$1-\del$.  On this event, the convergence analysis of
\cite{AG19} applies unchanged and yields the claimed primal and dual guarantees.  Since each
MMWU update changes only constantly many coefficients and there are
$\widetilde O(\gam^2)$ rounds, the final coefficient representation has
support $\widetilde O(\gam^2)$.

It remains to bound the running time.  Both MMWU phases use
$\widetilde O(\gam^2)$ rounds.  The stronger feasibility tolerance does not change the runtime bound,
since $\alpha_{\mathrm P},\beta=\Theta(1/\gam)$ and the primal phase
still uses $\widetilde O(\gam^2)$ rounds. On round $t$, the accumulated
Hamiltonian contains $O(t)$ sparse input matrices and has total
coefficient mass $\widetilde O(\gam)$.  It suffices to use additive expectation accuracy
$c_1\gam^{-1}$ in every oracle call, for a sufficiently small universal
constant $c_1>0$. This is at least as accurate as required by either phase. Substituting these
parameters into \Cref{thm:classical-gibbs-moments}, with
$s_{\bA}=s_{\bB}=\sr$, gives a round-$t$ running time of
\begin{align}
\widetilde O\left(
\dn\sr\gam^{5/2}t
+
\mc\sr\gam^2
\right).
\label{eqn:direct-cost-one-round}
\end{align}
The first term is the cost of applying the Hamiltonian-dependent
filtering for an $O(t)$-term Hamiltonian, while the second is the cost of
extracting all $\mc$ constraint expectations from the resulting shared
samples. Summing over
$\TT=\widetilde O(\gam^2)$ rounds achieves the overall claimed runtime:
\begin{align}
\sum_{t=1}^{\TT}
\widetilde O\left(
\dn\sr\gam^{5/2}t
+
\mc\sr\gam^2
\right)
&=
\widetilde O\left(
\dn\sr\gam^{5/2}\TT^2
+
\mc\sr\gam^2\TT
\right)=
\widetilde O\left(
\dn\sr\gam^{13/2}
+
\mc\sr\gam^4
\right).
\end{align}
\end{proof}

For constant $\gam$ and $\sr$, the runtime is
$\widetilde O(\dn+\mc)$.  The simultaneous Gibbs estimator is what makes
the dependence on the matrix dimension and the number of constraints
additive. Each Hamiltonian-dependent filtered vector is generated once
and then reused to estimate all $\mc$ constraint expectations. The main limitation is the nested accuracy requirement.  The outer MMWU
procedure already performs $\widetilde O(\gam^2)$ rounds, and on every
round the Gibbs expectations must be estimated to additive accuracy
$\Theta(\gam^{-1})$, requiring
$\widetilde O(\gam^2)$ filtered random vectors.  The next section avoids
this repeated high-accuracy estimation by replacing per-round
constraint search with a single stochastic feedback sample.

\section{The Two-Sided Stochastic SDP solver}
\label{sec:algorithm}

We now turn to our second main algorithm: a fully stochastic primal--dual
solver for sparse SDPs.  As described in
\Cref{sec:preliminaries,sec:standard-sdp-to-game}, SDP solving reduces to
finding an approximate saddle point of a zero-sum game between a matrix
player and a constraint player.  Our algorithm solves this game by allowing
both players to update from a single noisy sample, rather than accurately
computing the quantities required by a deterministic update.

On each round, the matrix player updates against a single sampled
constraint, while the constraint player updates from a single sampled
coordinate of the matrix player's rank-one response.  Thus neither player
needs to construct the dense object arising in the corresponding
deterministic implementation: the matrix player avoids the dense mixture
$\mathcal A^\ast\mathbf z$, while the constraint player avoids evaluating
the full payoff vector
$\bigl(\inner{\A{j}}{\Xt{t}}\bigr)_{j\in[\mc]}$.

Our analysis will prove that both stochastic updates can be controlled
simultaneously. The matrix and constraint players each incur
$O(\sqrt{\TT})$ regret, and the additional sampling and Lánczos errors
remain of the same or lower order. Averaging the $\TT$ played actions
therefore yields an $O(1/\sqrt{\TT})$ saddle-point gap while each round
accesses only sparse matrix data. This gives the following end-to-end
guarantee.

\begin{theorem}[Main theorem for sampled sparse-oracle SDP solving]
\label{thm:main}
Assume that each $\A{j}\in\R^{\dn\times\dn}$ is real symmetric,
$\sr$-row-sparse, and satisfies $\norm{\A{j}}\leq\om$.  Define
$\gam:=\Rx\Ry\om/\eps$ and
$\Lam:=\log(16\mc\dn/\del)$.  Run
\Cref{alg:sampled-rank-one} for $\TT$ rounds with the parameters specified
there.  With probability at least $1-\del$,
\begin{align}
\Gap(\Xout,\zout)
&\leq
O\left(\Rx\Ry\om\sqrt{\frac{\Lam}{\TT}}\right).
\label{eqn:main-gap}
\end{align}
Thus,
$\TT=O(\gam^2\Lam)$ rounds suffice for an $\eps$-accurate saddle point, with running time
\begin{align}
\widetilde O\left(
\dn\sr\min\left\{\gam^2,\mc,\frac{\dn}{\sr}\right\}\gam^{5/2}
+(\dn+\mc\sr)\gam^2
\right).
\label{eqn:main-theorem-runtime}
\end{align}
\end{theorem}

Our algorithm draws inspiration from two complementary online-learning
approaches, corresponding to the two players in the SDP game.  On the
constraint-player side, Garber and Hazan~\cite{GarberHazan2016} show how
to use sampled constraints together with clipped stochastic feedback, but
their complexity depends on Frobenius-norm and total-nonzero parameters.
On the matrix-player side, Carmon, Duchi, Sidford, and
Tian~\cite[Theorem~9 and Appendix~E.1]{CDST19} provide the rank-one matrix
learner and its average-projection regret bound, but their primal--dual
implementation maintains a dense dual mixture and evaluates the full
payoff vector $[\mathbf c_t]_j:=\inner{\A{j}}{\Xt{t}}$ for every
$j\in[\mc]$.  We adapt these ideas by sampling the dual action, so that
the matrix learner evolves under a sparse accumulated Hamiltonian, and by
replacing the exact payoff vector with the operator-norm-controlled
shared-coordinate estimator developed above.  This yields the runtime in
\Cref{eqn:main-theorem-runtime}. In particular, for fixed $\gam$ and
$\sr$, it is $\widetilde O(\dn+\mc)$ without any rank or Frobenius-norm
assumption on the constraint matrices.

We organize the analysis as follows.  We first state the two-sided
stochastic algorithm (\Cref{subsec:practical-two-sided-algorithm}) and
decompose the saddle-point gap of its averaged output into separate
matrix- and constraint-player regret terms (\Cref{sec:analysis}).  We
then formalize the stochastic history and concentration tools used to
control the adaptive sampling errors
(\Cref{subsec:stochastic-history}).  For the matrix player, we couple the
computed finite-Lánczos actions to exact rank-one reference responses
along the same realized trajectory
(\Cref{subsec:reference-actions}) and use this coupling to prove the
matrix-side guarantee
(\Cref{subsec:matrix-player-guarantee}).  For the constraint player, we
first establish a stochastic exponential-weights guarantee for clipped,
noisy feedback (\Cref{subsec:stochastic-hedge}) and then combine it with
the pure-constraint sampling bound to obtain the constraint-side
guarantee (\Cref{subsec:constraint-player-guarantee}).  Finally, we
combine the two regret bounds to prove convergence, choose the required
number of rounds, and account directly for the sparse-oracle and Lánczos
implementation costs in the proof of \Cref{thm:main}
(\Cref{subsec:main-theorem-proof}).

\subsection{The Stochastic Algorithm}
\label{subsec:practical-two-sided-algorithm}

We now describe the two-sided stochastic solver underlying
\Cref{thm:main}.  Each round couples two stochastic online-learning
updates.  The constraint player maintains a probability distribution over
constraints, placing greater weight on constraints with smaller cumulative
estimated payoff, and samples a single constraint from this distribution.
The matrix player therefore updates against only one sparse constraint
matrix per round.  In the opposite direction, the matrix player produces a
rank-one response, and a single sampled coordinate of this response is
reused to estimate the payoffs of all constraints through the
shared-coordinate estimator of \Cref{sec:energy-estimator}.  Thus neither
player needs to construct the dense object appearing in the corresponding
deterministic update.

The matrix update is based on the randomized rank-one learner of Carmon
et al.  For a real symmetric score matrix $\bH$ and a real unit vector
$\bu$, define
\begin{align}
    \cP_{\bu}(\bH)
    &:={}
    \frac{e^{\bH/2}\bu\bu^\top e^{\bH/2}}
         {\bu^\top e^{\bH}\bu},
    &
    \overline{\cP}(\bH)
    &:={}
    \E_{\bu\sim\operatorname{Unif}(\mathbb S^{\dn-1})}
    [\cP_{\bu}(\bH)].
    \label{eqn:rank-one-projection}
\end{align}
Here $\cP_{\bu}(\bH)$ is the rank-one projector obtained by normalizing
$e^{\bH/2}\bu$, while $\overline{\cP}(\bH)$ is its expectation over the
random direction.  In general, this expectation does not coincide with
the Gibbs state used by standard MMWU, i.e.,
$\overline{\cP}(\bH)\neq e^{\bH}/\Tr(e^{\bH})$.  Carmon et al.\
nevertheless prove a regret guarantee directly for
$\overline{\cP}(\bH)$ and show that replacing this mean action by one
fresh rank-one sample contributes only a zero-mean sampling fluctuation.
We retain this matrix learner, but modify its primal--dual implementation
so that it neither forms the dense mixture
$\Ry\sum_j p_{t,j}\A{j}$ nor evaluates every payoff
$\inner{\A{j}}{\bX_t}$ exactly.

The resulting proposed algorithm is summarized in
\Cref{alg:sampled-rank-one}.  On round $t$, the constraint player first
forms the exponential-weights distribution
$\phat{t}=\softmax(-\lrY\lhat{t})$ from its cumulative stochastic feedback
$\lhat{t}=\sum_{\tau<t}\chat{\tau}$ and samples a constraint
$\jt{t}\sim\phat{t}$ (Line~\ref{line:two-constraint-dist}).  It then plays
the corresponding pure action $\zt{t}=\Ry\be_{\jt{t}}$, so the matrix
player receives the gain
$\Gt{t}=\cA^\ast\zt{t}=\Ry\A{\jt{t}}$
(Line~\ref{line:two-constraint-gain}).  Because this gain consists of a
single sampled constraint matrix, it remains $\sr$-row-sparse and satisfies
$\norm{\Gt{t}}\leq\Ry\om$.

The matrix player independently draws a fresh random unit vector
$\ut{t}$ and computes a rank-one response to the accumulated Hamiltonian
$\Ht{t}=\sum_{\tau<t}\Gt{\tau}$
(Lines~\ref{line:two-direction}--\ref{line:two-primal-action}).  Notice
that $\Ht{t}$ contains only gains from previous rounds. In particular, the newly sampled
gain $\Gt{t}$ is added after the current matrix action is computed.  This
one-round lag ensures that, once the previous rounds are fixed, the
matrix response is independent of the constraint sampled on the current
round.  Rather than forming the matrix exponential explicitly, we
approximate its action on $\ut{t}$ using Lánczos and normalize the
resulting vector to obtain the rank-one action
$\Xtilde{t}=\Rx\psitilde{t}\psitilde{t}^{\top}$.

After the matrix action is fixed, the algorithm draws a single coordinate
$\isamp{t}$ according to
$\Prb[\isamp{t}=i]=|[\psitilde{t}]_i|^2$
(Line~\ref{line:two-coordinate}).  This same coordinate is reused in the
shared-coordinate estimator of \Cref{sec:energy-estimator} to produce a
stochastic payoff estimate for every constraint.  The estimates are
clipped and accumulated to form the feedback used by the constraint
player on the next round
(Lines~\ref{line:two-feedback}--\ref{line:two-update}).  Consequently, a
round never forms the dense mixture $\sum_j[\phat{t}]_j\A{j}$ or computes
the full collection of quadratic forms
$\{\langle\A{j},\Xtilde{t}\rangle\}_{j=1}^{\mc}$ exactly.

\begin{algorithm}[t]
\small
\caption{Two-Sided Stochastic SDP Solver}
\label{alg:sampled-rank-one}
\begin{algorithmic}[1]
\Require $\A{1},\ldots,\A{\mc}$, $\om$, $\Rx$, $\Ry$, $\TT$, $\del$

\State $\Lam\gets{\log(16\mc\dn/\del)}$ and
$\del_{\mathrm L}\gets\del/3$
\label{line:two-params-first}
\State $\LamL\gets{\log\left(\frac{4\TT^2\sqrt{\dn}}{\del_{\mathrm L}}\right)}$
\State $\lrX\gets{\frac1{\Ry\om}\sqrt{\frac{2\log(4\dn)}{3\TT}}}$
\State $\Bloc\gets
\Theta\!\left(
\Rx\om
\sqrt{
\frac{\TT}{\Lam}
}
\right)$
\State $\lrY\gets
\Theta\!\left(
\min\left\{
\frac{1}{\Bloc},
\sqrt{
\frac{{\log(2\mc)}}
{\Rx^2\om^2\TT}
}
\right\}
\right)$
\label{line:two-params-last}
\Statex Choose the constant in $\lrY$ so that $\lrY\Bloc\leq1$.

\State $\Ht{1}\gets\bzero$ and
$\lhat{1}\gets\bzero\in\R^{\mc}$

\For{$t=1,\ldots,\TT$}
    \State $\phat{t}\gets\softmax(-\lrY\lhat{t})$ and draw
    $\jt{t}\sim\phat{t}$
    \label{line:two-constraint-dist}

    \State $\zt{t}\gets\Ry\be_{\jt{t}}$ and
    $\Gt{t}\gets\Ry\A{\jt{t}}$
    \label{line:two-constraint-gain}

    \State Independently draw
    $\ut{t}\sim\operatorname{Unif}(\mathbb S^{\dn-1})$
    \label{line:two-direction}

    \State $k_t\gets
    {\min\left\{
        \dn,
        \left\lceil C_{\mathrm L}
        \sqrt{\LamL\max\{\lrX\Ry\om\,t,\LamL\}}
        \right\rceil
    \right\}}$, for sufficiently large universal constant $C_{\mathrm L}$.
    \label{line:two-lanczos-degree}

    \State $\vtilde{t}\gets
    \LanczosExp(\lrX\Ht{t}/2,\ut{t},k_t)$

    \State $\psitilde{t}\gets
    \vtilde{t}/\norm{\vtilde{t}}_2$ and
    $\Xtilde{t}\gets
    \Rx\psitilde{t}\psitilde{t}^\top$
    \label{line:two-primal-action}

    \State Draw $\isamp{t}$ with
    $\Prb[\isamp{t}=i\mid\psitilde{t}]
    =\abs{[\psitilde{t}]_i}^2$
    \label{line:two-coordinate}

    \For{$j=1,\ldots,\mc$}
        \State $Z_{t,j}\gets
        \Rx\Zest_{\A{j}}(\psitilde{t},\isamp{t})$ and
        $[\chat{t}]_j\gets\clip_{\Bloc}(Z_{t,j})$
        \label{line:two-feedback}
    \EndFor

    \State $\Ht{t+1}\gets\Ht{t}+\Gt{t}$ and
    $\lhat{t+1}\gets\lhat{t}+\chat{t}$
    \label{line:two-update}
\EndFor

\State $\Xout\gets
\TT^{-1}\sum_{t=1}^{\TT}\Xtilde{t}$ and
$\zout\gets
\TT^{-1}\sum_{t=1}^{\TT}\zt{t}$
\label{line:two-output}

\State \Return $\Xout,\zout$
\end{algorithmic}
\end{algorithm}

The parameter scales in
Lines~\ref{line:two-params-first}--\ref{line:two-params-last} balance the
matrix regret, stochastic-feedback variance, and clipping bias.  The
Lánczos degree is chosen to approximate each exponential action to the
precision required by the regret analysis while exploiting the gradually
growing norm of the accumulated Hamiltonian.  Consequently, each matrix response is computed by applying Lánczos to the
Hamiltonian formed from the sampled constraints, while each constraint
update uses a single sampled coordinate of that response to estimate all
payoffs, avoiding both dense constraint mixtures and exact quadratic-form
evaluation.

Finally, the algorithm returns the empirical averages of the primal and
dual actions played across the $\TT$ rounds.  On the primal matrix side, this
gives
\begin{align}
\Xout
&:={}
\frac1\TT\sum_{t=1}^{\TT}\Xtilde{t}.
\label{eqn:primal-output}
\end{align}
Since each $\Xtilde{t}$ is a rank-one element of $\Rx\cS_{\dn}$, their
average remains in $\Rx\cS_{\dn}$ and satisfies
$\rank(\Xout)\leq\min\{\TT,\dn\}$.  In particular, $\Xout$ need not be
materialized as a dense matrix and can instead be represented by its
$\TT$ rank-one factors. On the dual constraint side,
\begin{align}
\zout
&:={}
\frac1\TT\sum_{t=1}^{\TT}\zt{t}
=
\frac{\Ry}{\TT}\sum_{t=1}^{\TT}\be_{\jt{t}}.
\label{eqn:dual-output}
\end{align}
Thus $\zout\in\Ry\Delta_{\mc}$ is simply the empirical distribution of
the sampled constraints, scaled by $\Ry$.  It can therefore be represented
by the constraint-sampling counts and has
$\abs{\supp(\zout)}\leq\min\{\TT,\mc\}$.

\subsection{Regret Decomposition}
\label{sec:analysis}

We begin the analysis by reducing the saddle-point gap of the averaged
output to the realized regrets of the two players.  In particular, we specialize the standard regret-to-equilibrium relation from
\Cref{eqn:intro-regret-to-gap-ideal} to the sampled actions produced by
\Cref{alg:sampled-rank-one}.  The stochastic analysis enters only
afterwards, when we show that the resulting regret terms are small with
high probability.

To express the saddle-point gap in terms of the realized trajectory, let
$\ctilde{t}$ denote the exact constraint-payoff vector of the computed
matrix action $\Xtilde{t}$, i.e.
\begin{align}
    [\ctilde{t}]_j
    &:=
    \inner{\A{j}}{\Xtilde{t}}
    =
    \Rx\psitilde{t}^{\top}\A{j}\psitilde{t}.
    \label{eqn:exact-practical-payoffs}
\end{align}
Thus $\norm{\ctilde{t}}_\infty\leq\Rx\om$.  Although the algorithm does
not evaluate this vector exactly, it records the true payoffs of the
matrix action actually played.
Averaging the played actions gives, by linearity,
\begin{align}
    \cA^\ast\zout
    &=
    \frac{1}{\TT}\sum_{t=1}^{\TT}\Gt{t},
    \label{eqn:dual-output-linearity}
\end{align}
while, for every $j\in[\mc]$,
\begin{align}
    \inner{\A{j}}{\Xout}
    &=
    \frac{1}{\TT}\sum_{t=1}^{\TT}[\ctilde{t}]_j.
    \label{eqn:primal-output-linearity}
\end{align}
Substituting these expressions into the gap formula from \Cref{eqn:gap}
and multiplying by $\TT$ yields
\begin{align}
    \TT\Gap(\Xout,\zout)
    &=
    \Rx\lambda_{\max}\!\left(
        \sum_{t=1}^{\TT}\Gt{t}
    \right)
    -
    \Ry\min_{j\in[\mc]}
    \sum_{t=1}^{\TT}[\ctilde{t}]_j.
    \label{eqn:gap-before-decomposition}
\end{align}
These are exactly the cumulative payoffs of the best fixed matrix action
and the best fixed constraint action, respectively, chosen in hindsight.
To recover the two players' regrets, it remains to compare these
best-in-hindsight payoffs with the cumulative payoff of the actions
actually played.

On each round, that realized payoff has the same value from either
player's perspective.  Indeed, using
$\Gt{t}=\Ry\A{\jt{t}}$ and $\zt{t}=\Ry\be_{\jt{t}}$,
\begin{align}
    \inner{\Gt{t}}{\Xtilde{t}}
    &=
    \Ry\inner{\A{\jt{t}}}{\Xtilde{t}}
    =
    \Ry[\ctilde{t}]_{\jt{t}}
    =
    \inner{\ctilde{t}}{\zt{t}}.
    \label{eqn:realized-payoff-identity}
\end{align}
We can therefore insert the cumulative realized payoff between the two
best-in-hindsight terms in \Cref{eqn:gap-before-decomposition}.  Adding
and subtracting
$\sum_{t=1}^{\TT}\inner{\Gt{t}}{\Xtilde{t}}$ and applying
\Cref{eqn:realized-payoff-identity} gives
\begin{align}
    \TT\Gap(\Xout,\zout)
    &=
    \underbrace{
        \Rx\lambda_{\max}\!\left(
            \sum_{t=1}^{\TT}\Gt{t}
        \right)
        -
        \sum_{t=1}^{\TT}
        \inner{\Gt{t}}{\Xtilde{t}}
    }_{\MatSide}+
    \underbrace{
        \sum_{t=1}^{\TT}
        \inner{\ctilde{t}}{\zt{t}}
        -
        \Ry\min_{j\in[\mc]}
        \sum_{t=1}^{\TT}[\ctilde{t}]_j
    }_{\VecSide}.
    \label{eqn:realized-gap-decomposition}
\end{align}
The two terms are precisely the realized regrets of the two players.  By
Rayleigh--Ritz,
\begin{align}
    \max_{\bX\in\Rx\cS_{\dn}}
    \sum_{t=1}^{\TT}\inner{\Gt{t}}{\bX}
    &=
    \Rx\lambda_{\max}\!\left(
        \sum_{t=1}^{\TT}\Gt{t}
    \right),
\end{align}
so $\MatSide$ compares the matrix player's cumulative realized payoff with
that of the best fixed matrix action in hindsight.  Likewise,
\begin{align}
    \min_{\bz\in\Ry\Delta_{\mc}}
    \sum_{t=1}^{\TT}\inner{\ctilde{t}}{\bz}
    &=
    \Ry\min_{j\in[\mc]}
    \sum_{t=1}^{\TT}[\ctilde{t}]_j,
\end{align}
since a linear objective over $\Ry\Delta_{\mc}$ is minimized by placing
all mass on a minimum-payoff coordinate.  Hence $\VecSide$ is the
constraint player's cumulative realized regret.

Importantly, \Cref{eqn:realized-gap-decomposition} is expressed in terms of
the true payoff vectors $\ctilde{t}$ of the computed matrix actions, not
the stochastic feedback $\chat{t}$ used by the constraint learner.  The
decomposition is therefore exact for the trajectory produced by the
algorithm.  It remains to show that the two realized regrets are small. Doing so requires controlling several sampling errors that accumulate
adaptively across rounds.  We handle these fluctuations using martingale
concentration.  Before bounding either player's regret, we therefore
formalize the random history of the algorithm and identify the conditional
structure that will be used throughout the subsequent analysis.

\subsection{Stochastic History and Concentration Tools}
\label{subsec:stochastic-history}

In the regret analysis to follow, we repeatedly compare sampled quantities with their corresponding conditional means.  Conditional on all randomness revealed before the current sample is drawn, these sampling errors have mean zero, so standard martingale concentration bounds control how much they can accumulate over the course of the algorithm.  We first describe precisely which random choices have been revealed at each stage of a round, and then state the concentration inequalities used throughout the analysis.

Recall that $\jt{t}$ is the constraint sampled on round $t$, $\ut{t}$ is
the fresh random vector used to construct the matrix response, and
$\isamp{t}$ is the coordinate sampled from that response to generate the
stochastic constraint feedback.  Let
$\Fil{0}:=\{\varnothing,\Omega\}$ and, for $t\geq2$, define
\begin{align}
\Fil{t-1}
&:=
\sigma\!\left(
\jt{\tau},\ut{\tau},\isamp{\tau}
\;\middle|\;
\tau<t
\right).
\label{eqn:pre-round-filtration}
\end{align}
Thus $\Fil{t-1}$ records all randomness revealed before round $t$. In
particular, the accumulated Hamiltonian $\Ht{t}$, cumulative constraint
feedback, and mixed constraint action $\phat{t}$ are fixed conditional on
this history.  The remaining randomness on round $t$ is then revealed in
two stages.  First, the constraint $\jt{t}\sim\phat{t}$ and the random
vector $\ut{t}$ are drawn independently conditional on $\Fil{t-1}$. We express the  history after these two samples are revealed as
\begin{align}
\Fil{t-1/2}
&:=
\sigma(\Fil{t-1},\jt{t},\ut{t}).
\label{eqn:half-step-filtration}
\end{align}
Conditional on
$\Fil{t-1/2}$, the computed matrix action $\Xtilde{t}$ and its payoff
vector $\ctilde{t}$ are fixed.  The only remaining randomness is the
coordinate sample $\isamp{t}$ used to construct the stochastic feedback
$\chat{t}$.  After this final sample, the full history through round $t$
is
\begin{align}
\Fil{t}
&:=
\sigma(\Fil{t-1/2},\isamp{t}).
\label{eqn:round-filtration}
\end{align}

This ordering of the randomness will be used in different ways by the two
players.  On the matrix side, after conditioning on $\jt{t}$ the gain
$\Gt{t}=\Ry\A{\jt{t}}$ is fixed while $\ut{t}$ remains a fresh random
direction.  On the constraint side, after conditioning on $\ut{t}$ the
matrix action and its payoff vector are fixed while $\jt{t}$ remains a
fresh sample from $\phat{t}$.  The independence of these two samples
conditional on $\Fil{t-1}$ therefore makes both sampling errors
conditionally unbiased.  The half-step history $\Fil{t-1/2}$ is needed
later for the stochastic payoff estimator, whose only remaining
randomness is the coordinate sample $\isamp{t}$.

The filtration above identifies exactly what is fixed before each fresh random sample is drawn.  This lets us treat the corresponding sampling errors as conditionally mean-zero fluctuations and control their cumulative contribution using standard martingale concentration bounds.  We will use two such inequalities.  The first is Azuma--Hoeffding, which applies to martingale-difference sequences with uniformly bounded increments \cite{Azuma1967}.  Recall that a sequence $(X_t)_{t=1}^{\TT}$ adapted to a filtration $(\mathcal F_t)_{t=0}^{\TT}$ is a martingale-difference sequence if
\begin{align}
\E\left[
X_t
\,\middle|\,
\mathcal F_{t-1}
\right]
&=
0
\end{align}
for every $t$.  In other words, once all previously revealed randomness is fixed, the next increment has no systematic bias.

\begin{lemma}[Azuma--Hoeffding]
\label{lem:azuma-hoeffding}
Let $(X_t)_{t=1}^{\TT}$ be a martingale-difference sequence with respect
to a filtration $(\mathcal F_t)_{t=0}^{\TT}$, so that
$\E[X_t\mid\mathcal F_{t-1}]=0$ for every $t$.  If
$\abs{X_t}\leq c$ almost surely, then, for every $x>0$,
\begin{align}
\Prb\left[
\sum_{t=1}^{\TT}X_t
>
c\sqrt{2\TT x}
\right]
&\leq
e^{-x}.
\label{eqn:azuma-form}
\end{align}
\end{lemma}
\noindent Azuma--Hoeffding controls these fluctuations using only a uniform bound on the size of each increment.  This is sufficient for the action-sampling errors that arise in the two player analyses.  For the stochastic payoff feedback, however, the clipping threshold can be much larger than the typical size of the estimator, so a bound based only on the worst-case increment would lose useful information.  We therefore also use Freedman's inequality, which incorporates the conditional second moments of the increments and yields a sharper bound in this setting.

\begin{lemma}[Freedman]
\label{lem:freedman}
Let $(X_t)_{t=1}^{\TT}$ be a martingale-difference sequence with respect
to a filtration $(\mathcal F_t)_{t=0}^{\TT}$.  Suppose
$\abs{X_t}\leq c$ almost surely and
\begin{align}
\sum_{t=1}^{\TT}
\E\left[
X_t^2
\,\middle|\,
\mathcal F_{t-1}
\right]
&\leq
v.
\end{align}
Then, for every $x>0$,
\begin{align}
\Prb\left[
\sum_{t=1}^{\TT}X_t
>
\sqrt{2vx}
+
\frac{2cx}{3}
\right]
&\leq
e^{-x}.
\label{eqn:freedman-form}
\end{align}
\end{lemma}
\noindent We will apply these two bounds at the corresponding points in the player analyses below.  For Freedman's inequality, we use the bounded-increment form stated in \cite[Theorem~1.6]{Freedman1975}.

\subsection{Coupling the Computed Lánczos and Exact Matrix Responses}
\label{subsec:reference-actions}

Before turning to the matrix-player regret bound, we isolate the second approximation that enters its analysis.  The stochastic tools above control the error from sampling a single rank-one response, but the algorithm also computes that response only approximately.  In particular, for the sampled direction $\bu_t$, it replaces the exact matrix-exponential vector by the output of a finite-step Lánczos method.  Since the matrix-player regret guarantee of Carmon et al.~\cite[Theorem~2]{CDST19} is stated for the exact randomized response $\cP_{\bu}(\bH)$, we must show that this numerical approximation changes the realized payoffs by only a negligible amount.  This will allow the matrix-player analysis to treat the rank-one sampling error and the Lánczos approximation error separately.

The comparison requires some care because the algorithm is adaptive.  An approximation error on one round changes the stochastic feedback and can therefore influence the Hamiltonians encountered on later rounds.  Rather than introducing a counterfactual trajectory in which all matrix-exponential actions are computed exactly, we compare the exact and computed responses locally on each round, using the same realized Hamiltonian $\Ht{t}$ and the same random direction $\ut{t}$.  In this way, any effect of earlier approximation errors is already incorporated into the history being conditioned on.

Recall the exact randomized projection $\cP_{\bu}(\bH)$ and its mean $\overline{\cP}(\bH)$ from \Cref{eqn:rank-one-projection}. For the realized Hamiltonian $\Ht{t}$ and randomly sampled vector $\ut{t}$, define the exact exponential vector, its normalization, and the corresponding rank-one response by
\begin{align}
    \vt{t}
    &:=
    e^{\lrX\Ht{t}/2}\ut{t},
    \qquad
    \psit{t}
    :=
    \frac{\vt{t}}{\norm{\vt{t}}_2},\qquad
    \rt{t}
    :=
    \psit{t}\psit{t}^{\top}
    =
    \cP_{\ut{t}}(\lrX\Ht{t}), \qquad
    \Xt{t}
    :=
    \Rx\rt{t}.
    \label{eqn:exact-reference-action}
\end{align}
The algorithm instead computes the approximate vectors and response
\begin{align}
\vtilde{t}
:=
\LanczosExp(\lrX\Ht{t}/2,\ut{t},k_t), \qquad
\psitilde{t}
:=
\frac{\vtilde{t}}{\norm{\vtilde{t}}_2},
\qquad
\rtilde{t}:=
\psitilde{t}\psitilde{t}^{\top},
\qquad
\Xtilde{t}:=
\Rx\rtilde{t}.
\label{eqn:computed-reference-action}
\end{align}
Finally, the matrix-player regret theorem is expressed in terms of the mean randomized response
\begin{align}
\xbar{t}
&:=
\Rx\overline{\cP}(\lrX\Ht{t}),
\label{eqn:mean-reference-action}
\end{align}
which averages the exact rank-one response over the random direction.  These three responses isolate the two discrepancies that must be controlled in the matrix-player analysis.  The difference $\xbar{t}-\Xt{t}$ is the sampling error incurred by replacing the mean randomized response with a single exact rank-one sample.  This error is conditionally mean zero and will be bounded in \Cref{subsec:matrix-player-guarantee}.  The difference $\Xt{t}-\Xtilde{t}$ is the numerical error incurred by replacing the exact rank-one response with its Lánczos approximation.  We now establish a uniform bound on this second error along the realized trajectory.

Because the Hamiltonian $\Ht{t}$ depends on the gains sampled on earlier rounds, the Lánczos approximation must remain accurate for every Hamiltonian encountered along the realized trajectory.  Moreover, since
$\Ht{t}=\sum_{\tau<t}\Gt{\tau}$ and
$\norm{\Gt{\tau}}\leq\Ry\om$, the norm of the exponential argument can increase with $t$.  The required number of Lánczos iterations must therefore increase accordingly.  The following lemma shows that the round-dependent choice of $k_t$ in \Cref{alg:sampled-rank-one} guarantees the desired approximation accuracy simultaneously for all $t\in[\TT]$.

\begin{lemma}[Uniform Lánczos Accuracy Along Sampled Trajectory]
\label{lem:lanczos-projector}
Consider the realized trajectory of \Cref{alg:sampled-rank-one} under the parameter choices
${\lrX=(\Ry\om)^{-1}\sqrt{2\log(4\dn)/(3\TT)}}$,
$\del_{\mathrm L}=\del/3$,
${\LamL=\log(4\TT^2\sqrt{\dn}/\del_{\mathrm L})}$, and
${k_t=\min\{\dn,\lceil C_{\mathrm L}\sqrt{\LamL\max\{\lrX\Ry\om t,\LamL\}}\rceil\}}$.
Define the Lánczos success event, on which every computed projector is within trace-distance $1/\TT$ of its exact counterpart,
\begin{align}
\GoodL
&:=
\bigcap_{t=1}^{\TT}
\left\{
\norm{\rtilde{t}-\rt{t}}_1
\leq
\frac1\TT
\right\}.
\label{eqn:lanczos-success-event}
\end{align}
Then $\Prb[\GoodL]\geq1-\del_{\mathrm L}$.
\end{lemma}

\begin{proof}
Fix a round $t$ and condition on the pre-round history $\Fil{t-1}$.
Then the accumulated Hamiltonian $\Ht{t}$ is fixed, while $\ut{t}$
remains an independent uniform direction on $\mathbb S^{\dn-1}$.  Set
\begin{align}
\bK_t
&:=
\frac{\lrX\Ht{t}}{2}.
\end{align}
Since
$\Ht{t}=\sum_{\tau<t}\Gt{\tau}$ and
$\norm{\Gt{\tau}}\leq\Ry\om$,
\begin{align}
\norm{\bK_t}
&\leq
\frac{\lrX\Ry\om(t-1)}{2}
\leq
\frac{\lrX\Ry\om t}{2}.
\label{eqn:lanczos-matrix-norm}
\end{align}

Conditioned on $\Fil{t-1}$, we may apply the random-start Lánczos
guarantee from \Cref{fact:real-random-start-lanczos}.  Recall that for a
symmetric matrix $\bK$, a uniform random starting vector $\bu$, target
relative error $\zeta$, and failure probability $\eta$, a number of
iterations
\begin{align}
q
&=
O\left(
\min\left\{
\dn,
\sqrt{
\Gamma
\max\left\{
\norm{\bK},
\Gamma
\right\}
}
\right\}
\right), \quad \text{where}
\quad
\Gamma
=
\Theta\left(
\log\left(
\frac{\sqrt{\dn}}{\zeta\eta}
\right)
\right),
\end{align}
suffices to ensure
\begin{align}
\norm{
\LanczosExp(\bK,\bu,q)-e^{\bK}\bu
}_2
&\leq
\zeta\norm{e^{\bK}\bu}_2
\end{align}
with probability at least $1-\eta$.

We apply this result with $\bK=\bK_t$, $\zeta=1/(4\TT)$, and $\eta=\del_{\mathrm L}/\TT$.
For these choices,
\begin{align}
\log\left(
\frac{\sqrt{\dn}}{\zeta\eta}
\right)
&=
\Theta\left(
\log\left(
\frac{\TT^2\sqrt{\dn}}{\del_{\mathrm L}}
\right)
\right)
=
\Theta(\LamL).
\label{eqn:lanczos-log-factor-specialized}
\end{align}
Together with \Cref{eqn:lanczos-matrix-norm}, this shows that the choice
of $k_t$ in \Cref{alg:sampled-rank-one} satisfies the required iteration
bound.  Hence, conditional on $\Fil{t-1}$, with probability at least
$1-\del_{\mathrm L}/\TT$,
\begin{align}
\norm{\vtilde{t}-\vt{t}}_2
&\leq
\frac{1}{4\TT}
\norm{\vt{t}}_2.
\label{eqn:lanczos-relative-vector-error}
\end{align}

We next transfer this relative approximation guarantee to the normalized
vectors.  By the reverse triangle inequality,
\begin{align}
\abs{
\norm{\vtilde{t}}_2-\norm{\vt{t}}_2
}
&\leq
\norm{\vtilde{t}-\vt{t}}_2
\leq
\frac{1}{4\TT}
\norm{\vt{t}}_2.
\label{eqn:lanczos-norm-error}
\end{align}
Using
$\psitilde{t}=\vtilde{t}/\norm{\vtilde{t}}_2$ and
$\psit{t}=\vt{t}/\norm{\vt{t}}_2$, we obtain
\begin{align}
\norm{\psitilde{t}-\psit{t}}_2
&\leq
\norm{
\frac{\vtilde{t}}{\norm{\vtilde{t}}_2}
-
\frac{\vtilde{t}}{\norm{\vt{t}}_2}
}_2
+
\norm{
\frac{\vtilde{t}-\vt{t}}{\norm{\vt{t}}_2}
}_2=
\frac{
\abs{
\norm{\vtilde{t}}_2-\norm{\vt{t}}_2
}
}{
\norm{\vt{t}}_2
}
+
\frac{
\norm{\vtilde{t}-\vt{t}}_2
}{
\norm{\vt{t}}_2
}
\leq
\frac{1}{2\TT}.
\label{eqn:lanczos-normalized-vector-error}
\end{align}
Thus normalization increases the approximation error by at most a
constant factor.

For any real unit vectors $\bx$ and $\by$,
\begin{align}
\norm{\bx\bx^\top-\by\by^\top}_1
&=
2\sqrt{1-\abs{\bx^\top\by}^2}
\leq
2\norm{\bx-\by}_2.
\label{eqn:pure-projector-distance}
\end{align}
Applying this to $\psitilde{t}$ and $\psit{t}$ gives
\begin{align}
\norm{\rtilde{t}-\rt{t}}_1
&\leq
2\norm{\psitilde{t}-\psit{t}}_2
\leq
\frac{1}{\TT}.
\label{eqn:lanczos-one-round-projector}
\end{align}
Hence relative vector accuracy of order $1/\TT$ suffices to obtain the
trace-norm accuracy needed by the regret analysis.

Define the round-$t$ failure event
\begin{align}
\mathcal E_t
&:=
\left\{
\norm{\rtilde{t}-\rt{t}}_1
>
\frac1\TT
\right\}.
\label{eqn:lanczos-round-failure-event}
\end{align}
The preceding argument holds for every realization of the pre-round
history, and therefore
\begin{align}
\Prb\left[
\mathcal E_t
\,\middle|\,
\Fil{t-1}
\right]
&\leq
\frac{\del_{\mathrm L}}{\TT}.
\label{eqn:lanczos-one-round-failure}
\end{align}
Averaging over the random history gives
\begin{align}
\Prb[\mathcal E_t]
&=
\E\left[
\Prb\left[
\mathcal E_t
\,\middle|\,
\Fil{t-1}
\right]
\right]
\leq
\frac{\del_{\mathrm L}}{\TT}.
\label{eqn:lanczos-unconditional-round-failure}
\end{align}
Thus adaptivity of the trajectory does not increase the per-round failure
probability. Finally,
$\GoodL^c=\bigcup_{t=1}^{\TT}\mathcal E_t$, so a union bound gives
\begin{align}
\Prb[\GoodL^c]
&\leq
\sum_{t=1}^{\TT}
\Prb[\mathcal E_t]
\leq
\del_{\mathrm L}.
\end{align}
Equivalently,
$\Prb[\GoodL]\geq1-\del_{\mathrm L}$, proving the lemma.
\end{proof}

The preceding lemma controls Lánczos error at the level of the rank-one
projectors.  The matrix-player analysis only depends on the corresponding
payoffs, so the following corollary converts this trace-norm approximation
into the cumulative payoff error used below.

\begin{corollary}[Cumulative Lánczos payoff error]
\label{cor:lanczos-payoff-error}
Suppose that the Lánczos success event $\GoodL$ occurs.  Then,
\begin{align}
\sum_{t=1}^{\TT}
\abs{
\inner{\Gt{t}}{\Xt{t}-\Xtilde{t}}
}
&\leq
\Rx\Ry\om.
\label{eqn:cumulative-lanczos-payoff-error}
\end{align}
\end{corollary}

\begin{proof}
On $\GoodL$,
\begin{align}
\norm{\rt{t}-\rtilde{t}}_1
&\leq
\frac1\TT
\end{align}
for every $t\in[\TT]$.  Using
$\Xt{t}=\Rx\rt{t}$,
$\Xtilde{t}=\Rx\rtilde{t}$,
$\norm{\Gt{t}}\leq\Ry\om$, and operator--trace norm duality,
\begin{align}
\abs{
\inner{\Gt{t}}{\Xt{t}-\Xtilde{t}}
}
&=
\Rx
\abs{
\inner{\Gt{t}}{\rt{t}-\rtilde{t}}
}\leq
\Rx\norm{\Gt{t}}
\norm{\rt{t}-\rtilde{t}}_1
\leq
\frac{\Rx\Ry\om}{\TT}.
\end{align}
Summing over $t\in[\TT]$ proves the claim.
\end{proof}

\subsection{The Matrix-Player Guarantee}
\label{subsec:matrix-player-guarantee}

With the Lánczos approximation controlled, we can now analyze the
matrix-player contribution $\MatSide$ from
\Cref{eqn:realized-gap-decomposition}.  The three matrix actions
introduced in \Cref{subsec:reference-actions} make the structure of the
argument explicit.  The mean response $\xbar{t}$ is controlled by the
average-projection regret theorem of Carmon et al. \cite{CDST19}, sampling a single
random direction replaces this mean by the rank-one action $\Xt{t}$, and
Lánczos replaces $\Xt{t}$ by the action $\Xtilde{t}$ actually computed
by the algorithm. Thus, the matrix-side error has three sources: 1) regret of the mean response,
2) sampling a single rank-one response around that mean, and 3) numerical
approximation by Lánczos.  The first is controlled deterministically, the
second by Azuma--Hoeffding, and the third by
\Cref{cor:lanczos-payoff-error}.

We first isolate the sampling term.
\begin{lemma}[Rank-one response sampling]
\label{lem:matrix-rank-one-sampling}
For every $\delta_X\in(0,1)$, with probability at least
$1-\delta_X$,
\begin{align}
\sum_{t=1}^{\TT}
\inner{\Gt{t}}{\xbar{t}-\Xt{t}}
&\leq
2\Rx\Ry\om
\sqrt{
2\TT\log\frac{1}{\delta_X}
}.
\label{eqn:matrix-rank-one-sampling-bound}
\end{align}
\end{lemma}

\begin{proof}
Define the payoff fluctuation
\begin{align}
\xi_t^X
&:=
\inner{\Gt{t}}{\xbar{t}-\Xt{t}}.
\end{align}
Conditional on the pre-round history $\Fil{t-1}$ and the sampled
constraint $\jt{t}$, the gain
$\Gt{t}=\Ry\A{\jt{t}}$ and Hamiltonian $\Ht{t}$ are fixed, while
$\ut{t}$ remains a fresh randomly sampled vector.  By definition,
$\xbar{t}$ is the mean of the exact rank-one response over this remaining
randomness, i.e.
\begin{align}
\xbar{t}
&=
{\E\left[
\Xt{t}
\,\middle|\,
\Fil{t-1},\jt{t}
\right]}.
\end{align}
Conditional on $\Fil{t-1}$ and $\jt{t}$, the sampled constraint $\jt{t}$ is known and hence
$\Gt{t}=\Ry\A{\jt{t}}$ is deterministic.  Therefore, it can be taken outside the conditional expectation, which gives
\begin{align}
\E\left[
\xi_t^X
\mid
\Fil{t-1},\jt{t}
\right]
&=
\inner{\Gt{t}}{
\E\left[
\xbar{t}-\Xt{t}
\mid
\Fil{t-1},\jt{t}
\right]
}=
\inner{\Gt{t}}{
\xbar{t}-
\E\left[
\Xt{t}
\mid
\Fil{t-1},\jt{t}
\right]
}
=
0.
\end{align}
We now remove the conditioning on $\jt{t}$.  Because
$\Fil{t-1}\subseteq \sigma(\Fil{t-1},\jt{t})$, the tower property gives
\begin{align}
\E\left[
\xi_t^X
\,\middle|\,
\Fil{t-1}
\right]
&=
\E\left[
\E\left[
\xi_t^X
\,\middle|\,
\Fil{t-1},\jt{t}
\right]
\,\middle|\,
\Fil{t-1}
\right]=
0.
\label{eqn:matrix-sampling-centered}
\end{align}
Thus an individual rank-one response need not be close to its mean, but
its payoff error is unbiased on every round.

Both $\xbar{t}$ and $\Xt{t}$ belong to $\Rx\cS_{\dn}$, and hence
\begin{align}
\norm{\xbar{t}-\Xt{t}}_1
&\leq
2\Rx.
\end{align}
Using $\norm{\Gt{t}}\leq\Ry\om$ and operator--trace norm duality,
\begin{align}
\abs{\xi_t^X}
&\leq
\norm{\Gt{t}}
\norm{\xbar{t}-\Xt{t}}_1
\leq
2\Rx\Ry\om.
\label{eqn:matrix-sampling-range}
\end{align}

By \Cref{eqn:matrix-sampling-centered}, the sequence
$(\xi_t^X)_{t=1}^{\TT}$ is a martingale-difference sequence with respect
to $(\Fil{t})_{t=0}^{\TT}$, and
\Cref{eqn:matrix-sampling-range} gives a uniform increment bound.
We may therefore apply Azuma--Hoeffding,
\Cref{lem:azuma-hoeffding}, with
$c=2\Rx\Ry\om$ and
$x=\log(1/\delta_X)$.  Since
\begin{align}
\sum_{t=1}^{\TT}\xi_t^X
&=
\sum_{t=1}^{\TT}
\inner{\Gt{t}}{\xbar{t}-\Xt{t}},
\end{align}
we obtain
\begin{align}
\Prb\left[
\sum_{t=1}^{\TT}
\inner{\Gt{t}}{\xbar{t}-\Xt{t}}
>2\Rx\Ry\om
\sqrt{
2\TT\log\frac{1}{\delta_X}
}
\right]
&\leq
\delta_X.
\end{align}
Equivalently, with probability at least $1-\delta_X$,
\begin{align}
\sum_{t=1}^{\TT}
\inner{\Gt{t}}{\xbar{t}-\Xt{t}}
&\leq
2\Rx\Ry\om
\sqrt{
2\TT\log\frac{1}{\delta_X}
},
\end{align}
which proves the claim.
\end{proof}

We now combine the preceding ingredients to obtain the full matrix-side bound. The deterministic matrix-learning guarantee controls the idealized mean responses, while the previous lemma bounds the cumulative fluctuation introduced by sampling a single rank-one response on each round. It remains only to account for the Lánczos approximation used to compute these responses. Together, these bounds yield the following guarantee.

\begin{theorem}[Matrix-side guarantee]
\label{thm:matrix-side-guarantee}
For every $\delta_X\in(0,1)$, with probability at least
$1-\delta_X-\del_{\mathrm L}$,
\begin{align}
\MatSide
&\leq
\frac{\Rx\log(4\dn)}{\lrX}
+
\frac{3\Rx\lrX\Ry^2\om^2\TT}{2}
+
2\Rx\Ry\om
\sqrt{
2\TT\log\frac{1}{\delta_X}
}
+
\Rx\Ry\om.
\label{eqn:matrix-side-bound}
\end{align}
\end{theorem}

\begin{proof}
Recall from \Cref{eqn:realized-gap-decomposition} that
\begin{align}
\MatSide
&:=
\Rx\lambda_{\max}\left(
\sum_{t=1}^{\TT}\Gt{t}
\right)
-
\sum_{t=1}^{\TT}
\inner{\Gt{t}}{\Xtilde{t}}.
\end{align}
This is the difference between the payoff of the best fixed matrix action
in hindsight and the cumulative payoff of the matrix actions actually
generated by the algorithm. Inserting the mean response $\xbar{t}$ and exact sampled response
$\Xt{t}$ gives
\begin{align}
\MatSide
&=
\underbrace{
\Rx\lambda_{\max}\left(
\sum_{t=1}^{\TT}\Gt{t}
\right)
-
\sum_{t=1}^{\TT}
\inner{\Gt{t}}{\xbar{t}}
}_{\text{Mean-Response Regret}}+
\underbrace{
\sum_{t=1}^{\TT}
\inner{\Gt{t}}{\xbar{t}-\Xt{t}}
}_{\text{Rank-One Sampling Error}}
+
\underbrace{
\sum_{t=1}^{\TT}
\inner{\Gt{t}}{\Xt{t}-\Xtilde{t}}
}_{\text{Lánczos Approximation Error}}.
\label{eqn:matrix-side-ledger}
\end{align}

The first term is controlled by the average-projection regret theorem of
Carmon et al.~\cite[Theorem~2]{CDST19}.  This theorem plays the role of
the usual exponential-weights guarantee for the matrix player.  Since
$\xbar{t}=\Rx\overline{\cP}(\lrX\Ht{t})$ and
$\Ht{t}=\sum_{\tau<t}\Gt{\tau}$, it gives
\begin{align}
\Rx\lambda_{\max}\left(
\sum_{t=1}^{\TT}\Gt{t}
\right)
-
\sum_{t=1}^{\TT}
\inner{\Gt{t}}{\xbar{t}}\leq
\frac{\Rx\log(4\dn)}{\lrX}
+
\frac{3\Rx\lrX}{2}
\sum_{t=1}^{\TT}\norm{\Gt{t}}^2\leq
\frac{\Rx\log(4\dn)}{\lrX}
+
\frac{3\Rx\lrX\Ry^2\om^2\TT}{2},
\label{eqn:matrix-mean-regret-specialized}
\end{align}
where the final inequality uses
$\norm{\Gt{t}}\leq\Ry\om$.  This bound applies deterministically to
every realized gain sequence, so the adaptive manner in which the gains
were generated causes no additional error. The second term is exactly the rank-one sampling error controlled by
\Cref{lem:matrix-rank-one-sampling}.  Hence, except with probability
$\delta_X$,
\begin{align}
\sum_{t=1}^{\TT}
\inner{\Gt{t}}{\xbar{t}-\Xt{t}}
&\leq
2\Rx\Ry\om
\sqrt{
2\TT\log\frac{1}{\delta_X}
}.
\end{align}
Finally, on the Lánczos success event $\GoodL$,
\Cref{cor:lanczos-payoff-error} gives
\begin{align}
\sum_{t=1}^{\TT}
\abs{
\inner{\Gt{t}}{\Xt{t}-\Xtilde{t}}
}
&\leq
\Rx\Ry\om.
\label{eqn:matrix-lanczos-regret-bound}
\end{align}

Thus the three terms in \Cref{eqn:matrix-side-ledger} are controlled,
respectively, by the average-response regret bound, the rank-one sampling
lemma, and the Lánczos approximation guarantee.  Let $\mathcal{G}_X$ denote the event on which the sampling-fluctuation bound holds, so that
$\Prb[\mathcal{G}_X]\geq 1-\delta_X$, and recall (\Cref{lem:lanczos-projector}) that the Lánczos success event
$\GoodL$ satisfies $\Prb[\GoodL]\geq 1-\del_{\mathrm L}$.  By a union bound,
\begin{align}
\Prb[\mathcal{G}_X\cap\GoodL]
&\geq
1-\delta_X-\del_{\mathrm L}.
\end{align}
On $\mathcal{G}_X\cap\GoodL$, the deterministic matrix-learning bound, the
sampling-fluctuation bound, and the Lánczos approximation bound all hold
simultaneously.  Summing these three contributions proves
\Cref{eqn:matrix-side-bound}.
\end{proof}

\subsection{Exponential Weights with Stochastic Feedback}
\label{subsec:stochastic-hedge}

The constraint-player analysis requires one additional ingredient. The constraint-side update maintains a probability distribution over the
$\mc$ constraints.  Intuitively, this distribution records which
constraints currently appear most favorable according to the feedback
observed so far.  Starting from equal weights, the algorithm updates the
weight of constraint $j$ after round $t$ by multiplying it by
$\exp(-\lrY[\chat{t}]_j)$, where $\chat{t}$ is the feedback vector and
$\lrY>0$ is a learning-rate parameter.  Thus constraints receiving
smaller cumulative feedback retain relatively larger weight.  Normalizing
these weights produces the distribution $\phat{t}$ used on the next
round.

The basic guarantee of this exponential-weights update is a
best-coordinate comparison.  If $\bc_t\in\mathbb R^{\mc}$ denotes the
true payoff vector on round $t$, then the distribution $\phat{t}$ incurs
the weighted payoff
$\inner{\bc_t}{\phat{t}}=\sum_{j=1}^{\mc}[\phat{t}]_j[\bc_t]_j$.
Over many rounds, exponential weights ensures that the sum of these
weighted payoffs is not much larger than the cumulative payoff of the
best single constraint $j\in[\mc]$ chosen after seeing the entire
sequence.  The difference
\begin{align}
\sum_{t=1}^{\TT}
\inner{\bc_t}{\phat{t}}
-\min_{j\in[\mc]}
\sum_{t=1}^{\TT}
[\bc_t]_j
\end{align}
is the corresponding regret.

In our setting, however, the constraint update does not have access to
the exact vector $\bc_t$.  Instead, it receives a random approximation
$\chat{t}$ produced by the shared-coordinate estimator and clipped to
control rare large fluctuations.  The update is therefore performed
using $\chat{t}$, even though the quantity we ultimately wish to control
is the regret with respect to the true payoffs $\bc_t$.  The following
lemma shows that the usual exponential-weights guarantee is stable under
this replacement.  The additional error depends only on three basic
properties of the stochastic feedback: 1) its range, 2) its second moment, and 3) its bias relative to the true payoff vector.

\begin{lemma}[Exponential Weights with Clipped Stochastic Feedback]
\label{lem:stochastic-hedge}
Let $(\mathcal G_t)_{t=0}^{\TT}$ be a filtration.  Let
$\bc_t\in\mathbb R^{\mc}$ be the true payoff vector on round $t$, and let
$\phat{t}\in\Delta_{\mc}$ be the distribution produced by exponential
weights {initialized uniformly} from the preceding feedback vectors
$\chat{1},\ldots,\chat{t-1}$ with learning rate $\lrY>0$. {Assume that $\chat{t}$ is $\mathcal G_t$-measurable.}
Suppose that $\bc_t$ and $\phat{t}$ are
$\mathcal G_{t-1}$-measurable and that, for every $j\in[\mc]$, the feedback has bounded
\begin{enumerate}
    \item range: $\abs{[\chat{t}]_j}\leq B$
    \item second moment: $\E\left[ [\chat{t}]_j^2\,\middle|\,\mathcal G_{t-1}\right]\leq\sigma^2$
    \item bias: $\abs{\E\left[ [\chat{t}]_j\,\middle|\,\mathcal G_{t-1}\right]-[\bc_t]_j}\leq b$.
\end{enumerate}
Let $\Lambda_{\mathrm F}:=\log\left(4(\mc+1)/\delta_{\mathrm F}\right)$. If $\lrY B\leq1$, then, for every
$\delta_{\mathrm F}\in(0,1)$, with probability at least
$1-\delta_{\mathrm F}$,
\begin{align}
\sum_{t=1}^{\TT}
\inner{\bc_t}{\phat{t}}
-
\min_{j\in[\mc]}
\sum_{t=1}^{\TT}[\bc_t]_j
\leq
O\Bigg(
\frac{\log(2\mc)}{\lrY}
+
\lrY\sigma^2\TT
+
\sigma\sqrt{\TT\Lambda_{\mathrm F}}
+
B\Lambda_{\mathrm F}
+
b\TT
\Bigg).
\label{eqn:stochastic-hedge-bound}
\end{align}
\end{lemma}

\begin{proof}
The proof separates the three effects introduced by stochastic clipped
feedback.  First, exponential weights controls regret relative to the
feedback $\chat{t}$ that it actually observes.  Second, concentration
controls the centered random difference between this feedback and its
conditional mean.  Third, clipping shifts that conditional mean slightly
away from the true payoff, producing a deterministic bias.

Define the conditional mean feedback
\begin{align}
\bmu_t
&:=
\E\left[
\chat{t}
\,\middle|\,
\mathcal G_{t-1}
\right],
\end{align}
the centered stochastic error
\begin{align}
\bxi_t
&:=
\chat{t}-\bmu_t,
\end{align}
and the bias
\begin{align}
\bbias_t
&:=
\bmu_t-\bc_t.
\end{align}
By construction, $ \norm{\bbias_t}_\infty \leq b$ and $\E\left[\bxi_t\,\middle|\,\mathcal G_{t-1}\right]=\bzero$. Hence,
\begin{align}
\bc_t
&=
\chat{t}-\bxi_t-\bbias_t.
\label{eqn:generic-feedback-decomposition}
\end{align}
Substituting this identity into the learner payoff and the best
fixed-coordinate comparator gives
\begin{align}
\sum_{t=1}^{\TT}
\inner{\bc_t}{\phat{t}}
-
\min_{j\in[\mc]}
\sum_{t=1}^{\TT}[\bc_t]_j
\leq&
\underbrace{
\left[
\sum_{t=1}^{\TT}
\inner{\chat{t}}{\phat{t}}
-
\min_{j\in[\mc]}
\sum_{t=1}^{\TT}[\chat{t}]_j
\right]
}_{\text{Regret for the Observed Feedback}}
+
\underbrace{
\left[
-
\sum_{t=1}^{\TT}
\inner{\bxi_t}{\phat{t}}
+
\max_{j\in[\mc]}
\sum_{t=1}^{\TT}[\bxi_t]_j
\right]
}_{\text{Centered Feedback Noise}} \nonumber \\
&+
\underbrace{
\left[
-
\sum_{t=1}^{\TT}
\inner{\bbias_t}{\phat{t}}
+
\max_{j\in[\mc]}
\sum_{t=1}^{\TT}[\bbias_t]_j
\right]
}_{\text{Clipping Bias}}.
\label{eqn:generic-regret-bridge}
\end{align}
We bound these three contributions separately.

\paragraph{Regret for the Observed Feedback.}
Since $\lrY B\leq1$ and $\abs{[\chat{t}]_j}\leq B$, we have
$\abs{\lrY[\chat{t}]_j}\leq1$ for every $t,j$. To derive the second-order Hedge bound, define the unnormalized weights
\begin{align}
w_{t,j}
&:=
\exp\left(
-\lrY
\sum_{\tau=1}^{t-1}
[\chat{\tau}]_j
\right)
\end{align}
and normalization term $W_t:=\sum_{j=1}^{\mc} w_{t,j}$, such that
\begin{align}
[\phat{t}]_j
&=
\frac{w_{t,j}}{W_t}.
\end{align}
After observing the feedback vector $\chat{t}$, each coordinate weight is
updated multiplicatively according to
\begin{align}
w_{t+1,j}
&=
w_{t,j}
\exp\left(
-\lrY{[\chat{t}]_j}
\right).
\end{align}
Summing these updated weights over all coordinates and normalizing by
$W_t$ gives
\begin{align}
\frac{{W_{t+1}}}{W_t}
&=
\sum_{j=1}^{\mc}
\frac{w_{t,j}}{W_t}
\exp\left(
-\lrY{[\chat{t}]_j}
\right)=
\sum_{j=1}^{\mc}
[\phat{t}]_j
\exp\left(
-\lrY[\chat{t}]_j
\right).
\end{align}
Since $\lrY B\leq 1$, we have
$\abs{\lrY[\chat{t}]_j}\leq 1$ for every $j$.  Applying
$e^{-x}\leq 1-x+x^2$ therefore gives
\begin{align}
\frac{W_{t+1}}{W_t}
&\leq
1
-
\lrY
\inner{\chat{t}}{\phat{t}}
+
\lrY^2
\sum_{j=1}^{\mc}
[\phat{t}]_j
[\chat{t}]_j^2=
1
-
\lrY
\inner{\chat{t}}{\phat{t}}
+
\lrY^2 Y_t,
\end{align}
where we define $Y_t:=\sum_{j=1}^{\mc}[\phat{t}]_j[\chat{t}]_j^2$.
Using $1+x\leq e^x$ gives
\begin{align}
\frac{W_{t+1}}{W_t}
&\leq
\exp\left(
-\lrY
\inner{\chat{t}}{\phat{t}}
+
\lrY^2 Y_t
\right).
\end{align}
Taking logarithms and summing over $t\in[\TT]$ yields
\begin{align}
\log\frac{W_{\TT+1}}{W_1}
&\leq
-\lrY
\sum_{t=1}^{\TT}
\inner{\chat{t}}{\phat{t}}
+
\lrY^2
\sum_{t=1}^{\TT}
Y_t.
\label{eqn:hedge-potential-upper}
\end{align}
On the other hand, for every $j\in[\mc]$,
\begin{align}
W_{\TT+1}
&\geq
w_{\TT+1,j}=
\exp\left(
-\lrY
\sum_{t=1}^{\TT}
[\chat{t}]_j
\right).
\end{align}
Since $W_1=\mc$, choosing the coordinate with minimum cumulative loss gives
\begin{align}
\log\frac{W_{\TT+1}}{W_1}
&\geq
-\lrY
\min_{j\in[\mc]}
\sum_{t=1}^{\TT}
[\chat{t}]_j
-
\log\mc.
\label{eqn:hedge-potential-lower}
\end{align}
Combining
\Cref{eqn:hedge-potential-upper,eqn:hedge-potential-lower}
and rearranging gives
\begin{align}
\sum_{t=1}^{\TT}
\inner{\chat{t}}{\phat{t}}
-
\min_{j\in[\mc]}
\sum_{t=1}^{\TT}
[\chat{t}]_j
\leq
\frac{\log(2\mc)}{\lrY}
+
\lrY
\sum_{t=1}^{\TT}
Y_t.
\label{eqn:generic-second-order-hedge}
\end{align}

By the bounded range and second moment assumptions, $0 \leq Y_t \leq B^2$ and $\E\left[ Y_t \,\middle|\, \mathcal G_{t-1} \right] \leq \sigma^2$, respectively.
Since $Y_t^2\leq B^2Y_t$, the centered increment
$Y_t-\E[Y_t\mid\mathcal G_{t-1}]$ has conditional variance at most
$B^2\sigma^2$.  Applying Freedman's inequality,
\Cref{lem:freedman}, gives, except with probability at most
$\delta_{\mathrm F}/2$,
\begin{align}
\sum_{t=1}^{\TT}Y_t
&\leq
\sigma^2\TT
+
C\left[
B\sigma\sqrt{\TT\Lambda_{\mathrm F}}
+
B^2\Lambda_{\mathrm F}
\right].
\label{eqn:generic-quadratic-process}
\end{align}
Because $\lrY B\leq1$, substituting this estimate into
\Cref{eqn:generic-second-order-hedge} yields
\begin{align}
\sum_{t=1}^{\TT}
\inner{\chat{t}}{\phat{t}}
-
\min_{j\in[\mc]}
\sum_{t=1}^{\TT}[\chat{t}]_j
&\leq
C\left[
\frac{\log(2\mc)}{\lrY}
+
\lrY\sigma^2\TT
+
\sigma\sqrt{\TT\Lambda_{\mathrm F}}
+
B\Lambda_{\mathrm F}
\right].
\label{eqn:generic-observed-regret}
\end{align}

\paragraph{Centered feedback noise.}
We next control the difference between the stochastic feedback and its
conditional mean.  Recall that
$\bxi_t=\chat{t}-\bmu_t$, where
$\bmu_t=\E[\chat{t}\mid\mathcal G_{t-1}]$.  By construction,
\begin{align}
\E\left[
\bxi_t
\mid
\mathcal G_{t-1}
\right]
&=
\bzero.
\label{eqn:generic-centered-feedback}
\end{align}
Moreover, clipping implies
$\abs{[\bxi_t]_j}\leq 2B$ for every $j\in[\mc]$.  Since subtracting the
conditional mean can only decrease the conditional second moment,
\begin{align}
\E\left[
[\bxi_t]_j^2
\mid
\mathcal G_{t-1}
\right]
&\leq
\sigma^2.
\label{eqn:generic-centered-second-moment}
\end{align}

We first control the fluctuation in the feedback averaged according to
$\phat{t}$. Since $\phat{t}$ is determined by the history before round
$t$, it is $\mathcal G_{t-1}$-measurable. Hence,
\begin{align}
\E\left[
-\inner{\bxi_t}{\phat{t}}
\mid
\mathcal G_{t-1}
\right]
&=
-\sum_{j=1}^{\mc}
[\phat{t}]_j
\E\left[
[\bxi_t]_j
\mid
\mathcal G_{t-1}
\right]
=0.
\end{align}
Thus
$(-\inner{\bxi_t}{\phat{t}})_{t=1}^{\TT}$ is a
martingale-difference sequence with respect to $(\mathcal G_t)$.
Because $\phat{t}$ is a probability distribution,
\begin{align}
\abs{\inner{\bxi_t}{\phat{t}}}
&\leq
\sum_{j=1}^{\mc}
[\phat{t}]_j
\abs{[\bxi_t]_j}
\leq
2B.
\end{align}
Similarly, convexity of $x\mapsto x^2$ gives
\begin{align}
\inner{\bxi_t}{\phat{t}}^2
&\leq
\sum_{j=1}^{\mc}
[\phat{t}]_j
[\bxi_t]_j^2.
\end{align}
Taking conditional expectations and using
\Cref{eqn:generic-centered-second-moment},
\begin{align}
\E\left[
\inner{\bxi_t}{\phat{t}}^2
\mid
\mathcal G_{t-1}
\right]
&\leq
\sum_{j=1}^{\mc}
[\phat{t}]_j
\E\left[
[\bxi_t]_j^2
\mid
\mathcal G_{t-1}
\right]\leq
\sigma^2.
\end{align}
Consequently, the sum of the conditional second moments satisfies
\begin{align}
\sum_{t=1}^{\TT}
\E\left[
\inner{\bxi_t}{\phat{t}}^2
\mid
\mathcal G_{t-1}
\right]
&\leq
\TT\sigma^2.
\end{align}

We also require concentration for the noise in each individual coordinate,
because the regret bound compares the weighted feedback under $\phat{t}$
with the best fixed coordinate $j\in[\mc]$ over all rounds.  Fix
$j\in[\mc]$.  Since
\begin{align}
[\bxi_t]_j
&=
[\chat{t}]_j
-\E\left[
[\chat{t}]_j
\mid
\mathcal G_{t-1}
\right],
\end{align}
we have
\begin{align}
\E\left[
[\bxi_t]_j
\mid
\mathcal G_{t-1}
\right]
&=
0.
\end{align}
Moreover, $[\bxi_t]_j$ is $\mathcal G_t$-measurable, since
$\chat{t}$ is $\mathcal G_t$-measurable.  Hence
$([\bxi_t]_j)_{t=1}^{\TT}$ is a martingale-difference sequence with
respect to $(\mathcal G_t)$. Clipping gives $\abs{[\chat{t}]_j}\leq B$, and therefore also $\abs{
\E\left[
[\chat{t}]_j
\mid
\mathcal G_{t-1}
\right]
}
\leq
B$.
Consequently,
\begin{align}
\abs{[\bxi_t]_j}
&\leq
2B.
\end{align}
Finally, by \Cref{eqn:generic-centered-second-moment},
\begin{align}
\sum_{t=1}^{\TT}
\E\left[
[\bxi_t]_j^2
\mid
\mathcal G_{t-1}
\right]
&\leq
\TT\sigma^2.
\end{align}

With these established bounds, we now set
\begin{align}
\eta
&:=
\frac{\delta_{\mathrm F}}{2(\mc+1)}.
\end{align}
Applying Freedman's inequality, \Cref{lem:freedman}, to
$(-\inner{\bxi_t}{\phat{t}})_{t=1}^{\TT}$ with increment bound $2B$,
sum of conditional second moments at most $\TT\sigma^2$, and failure
probability $\eta$, gives
\begin{align}
\Prb\left[
-\sum_{t=1}^{\TT}
\inner{\bxi_t}{\phat{t}}
>
C_0\left(
\sigma\sqrt{\TT\log\frac{1}{\eta}}
+
B\log\frac{1}{\eta}
\right)
\right]
&\leq
\eta,
\end{align}
for a universal constant $C_0$.  Likewise, for every fixed
$j\in[\mc]$, applying \Cref{lem:freedman} to
$([\bxi_t]_j)_{t=1}^{\TT}$ gives
\begin{align}
\Prb\left[
\sum_{t=1}^{\TT}
[\bxi_t]_j
>
C_0\left(
\sigma\sqrt{\TT\log\frac{1}{\eta}}
+
B\log\frac{1}{\eta}
\right)
\right]
&\leq
\eta.
\end{align}
There are $\mc+1$ such concentration events, one for the weighted
feedback and one for each coordinate.  A union bound therefore implies
that all of these inequalities hold simultaneously except with
probability at most
\begin{align}
(\mc+1)\eta
&=
\frac{\delta_{\mathrm F}}{2}.
\end{align}
On this event, taking the maximum over $j\in[\mc]$ in the coordinate
bounds and adding the weighted-feedback bound yields
\begin{align}
-\sum_{t=1}^{\TT}
\inner{\bxi_t}{\phat{t}}
+
\max_{j\in[\mc]}
\sum_{t=1}^{\TT}
[\bxi_t]_j
&\leq
C\left[
\sigma\sqrt{\TT\Lambda_{\mathrm F}}
+
B\Lambda_{\mathrm F}
\right],
\label{eqn:generic-feedback-noise}
\end{align}
where $C$ is a universal constant and
$\Lambda_{\mathrm F}$ upper-bounds
\begin{align}
\log\frac{1}{\eta}
&=
\log\frac{2(\mc+1)}{\delta_{\mathrm F}}.
\end{align}

This argument does not require the coordinates of $\bxi_t$ to be
independent.  Freedman's inequality is applied separately to each scalar
martingale-difference sequence, and the union bound is used only to make
the resulting bounds hold simultaneously.

\paragraph{Clipping bias.}
Finally, the bias caused by clipping requires no concentration.  Since
$\norm{\bbias_t}_\infty\leq b$ and $\phat{t}$ is a probability
distribution, we have that $[\bbias_t]_j \leq b$ and $-\inner{\bbias_t}{\phat{t}} \leq b$, for every $t$ and $j$.  Hence,
\begin{align}
-
\sum_{t=1}^{\TT}
\inner{\bbias_t}{\phat{t}}
+
\max_{j\in[\mc]}
\sum_{t=1}^{\TT}[\bbias_t]_j
&\leq
2b\TT.
\label{eqn:generic-feedback-bias-bound}
\end{align}
Combining
\Cref{eqn:generic-regret-bridge,eqn:generic-observed-regret,eqn:generic-feedback-noise,eqn:generic-feedback-bias-bound}
gives \Cref{eqn:stochastic-hedge-bound}.  The two probabilistic bounds
each fail with probability at most $\delta_{\mathrm F}/2$, so their
intersection has probability at least $1-\delta_{\mathrm F}$ by a union
bound.

\end{proof}

\subsection{The Constraint-Player Guarantee}
\label{subsec:constraint-player-guarantee}

We now bound the constraint-player contribution $\VecSide$ from
\Cref{eqn:realized-gap-decomposition}.  The main complication is that the
algorithm never works with the exact payoff vector $\ctilde{t}$ directly.
Instead, it maintains a distribution $\phat{t}$ over the constraints,
samples a single constraint $\jt{t}\sim\phat{t}$ to play, and updates
$\phat{t}$ using the clipped stochastic feedback $\chat{t}$ produced by
the shared-coordinate estimator.

Accordingly, the analysis separates into two key sources of error.  The first
comes from replacing the mixed action $\phat{t}$ by the sampled
constraint $\jt{t}$ on each round.  The second comes from replacing the
true payoff vector $\ctilde{t}$ by the stochastic feedback $\chat{t}$ in
the exponential-weights update.  We control these two effects separately
and then combine them to obtain the final constraint-side bound.

We first isolate the error introduced by sampling a single constraint.

\begin{lemma}[Pure-Constraint Sampling]
\label{lem:pure-constraint-sampling}
For every $\delta_{\mathrm S}\in(0,1)$, with probability at least
$1-\delta_{\mathrm S}$,
\begin{align}
\Ry\sum_{t=1}^{\TT}
\left(
[\ctilde{t}]_{\jt{t}}
-
\inner{\ctilde{t}}{\phat{t}}
\right)
&\leq
2\Rx\Ry\om
\sqrt{
2\TT\log\frac{1}{\delta_{\mathrm S}}
}.
\label{eqn:pure-constraint-sampling-bound}
\end{align}
\end{lemma}

\begin{proof}
Define
\begin{align}
\xi_t^Y
&:=
\Ry\left(
[\ctilde{t}]_{\jt{t}}
-
\inner{\ctilde{t}}{\phat{t}}
\right).
\end{align}
Conditional on the pre-round history $\Fil{t-1}$ and the random direction
$\ut{t}$, the computed action $\Xtilde{t}$ and hence its payoff vector
$\ctilde{t}$ are fixed.  Because $\jt{t}$ and $\ut{t}$ are independent
conditional on $\Fil{t-1}$, the constraint $\jt{t}$ remains a fresh draw
from $\phat{t}$.  Therefore
\begin{align}
\E\left[
[\ctilde{t}]_{\jt{t}}
\,\middle|\,
\Fil{t-1},\ut{t}
\right]
&=
\inner{\ctilde{t}}{\phat{t}},
\end{align}
and hence
\begin{align}
\E\left[
\xi_t^Y
\,\middle|\,
\Fil{t-1},\ut{t}
\right]
&=
0.
\end{align}
The tower property gives
$\E[\xi_t^Y\mid\Fil{t-1}]=0$.  Thus sampling a single constraint changes
the realized payoff on an individual round but introduces no systematic
bias.

For every $j\in[\mc]$, since
$\norm{\A{j}}\leq\om$ and
$\Xtilde{t}\in\Rx\cS_{\dn}$,
\begin{align}
\abs{[\ctilde{t}]_j}
&=
\abs{\inner{\A{j}}{\Xtilde{t}}}
\leq
\Rx\om.
\end{align}
Because $\phat{t}$ is a probability
distribution, the same bound also holds for
$\abs{\inner{\ctilde{t}}{\phat{t}}}$.  Consequently,
\begin{align}
\abs{\xi_t^Y}
&\leq
2\Rx\Ry\om.
\end{align}
Together with the conditional-centering relation established above, this
shows that $(\xi_t^Y)_{t=1}^{\TT}$ is a martingale-difference sequence
with increments bounded in magnitude by $2\Rx\Ry\om$.  Applying
Azuma--Hoeffding, \Cref{lem:azuma-hoeffding}, therefore gives, for every
$\delta_S\in(0,1)$,
\begin{align}
\Prb\left[
\sum_{t=1}^{\TT}
\xi_t^Y
>
2\Rx\Ry\om
\sqrt{
2\TT\log\frac{1}{\delta_S}
}
\right]
&\leq
\delta_S.
\end{align}
Equivalently, with probability at least $1-\delta_S$,
\begin{align}
\sum_{t=1}^{\TT}
\xi_t^Y
&\leq
2\Rx\Ry\om
\sqrt{
2\TT\log\frac{1}{\delta_S}
},
\end{align}
which proves the claim.
\end{proof}

The preceding results now control both discrepancies introduced by the
constraint-side implementation.  The stochastic exponential-weights
bound controls the use of the clipped estimator $\chat{t}$ in place of
the true payoff vector $\ctilde{t}$, while the pure-constraint
concentration bound controls the additional fluctuation from sampling
$\jt{t}\sim\phat{t}$ rather than using the full distribution
$\phat{t}$.  It remains to substitute the estimator bounds and choose the
clipping threshold and learning rate so that the resulting error terms
are balanced.  This yields the following constraint-side guarantee.

\begin{theorem}[Constraint-Side Guarantee]
\label{thm:constraint-side-guarantee}
As in \Cref{alg:sampled-rank-one}, let
$\Lam:=\log\left(16\mc\dn/\del\right)$ and
$\Bloc:=\Theta\left(\Rx\om\sqrt{\TT/\Lam}\right)$, and choose
$\lrY:=\Theta\left(\min\left\{1/\Bloc,\sqrt{\log(2\mc)/(\Rx^2\om^2\TT)}\right\}\right)$,
with the universal constant sufficiently small such that
$\lrY\Bloc\leq1$. Then, with probability at least $1-\del/3$,
\begin{align}
\VecSide
&\leq
O\left(
\Rx\Ry\om\sqrt{\TT\Lam}
\right).
\label{eqn:constraint-side-bound}
\end{align}
\end{theorem}

\begin{proof}
Recall that
\begin{align}
\VecSide
&:=
\Ry\left[
\sum_{t=1}^{\TT}[\ctilde{t}]_{\jt{t}}
-
\min_{j\in[\mc]}
\sum_{t=1}^{\TT}[\ctilde{t}]_j
\right].
\end{align}
The first term is the payoff of the pure constraints actually sampled by
the algorithm, while exponential weights controls the payoff of its mixed
actions $\phat{t}$.  Adding and subtracting the latter gives
\begin{align}
\VecSide
&=
\underbrace{
\Ry\sum_{t=1}^{\TT}
\left(
[\ctilde{t}]_{\jt{t}}
-
\inner{\ctilde{t}}{\phat{t}}
\right)
}_{\text{Pure-Action Sampling Error}}
+
\underbrace{
\Ry\left[
\sum_{t=1}^{\TT}
\inner{\ctilde{t}}{\phat{t}}
-
\min_{j\in[\mc]}
\sum_{t=1}^{\TT}[\ctilde{t}]_j
\right]
}_{\text{Mixed-Action Regret}}.
\label{eqn:constraint-side-ledger}
\end{align}
The first contribution is controlled by
\Cref{lem:pure-constraint-sampling}.  It remains to control the
mixed-action regret when the learner sees the stochastic feedback
$\chat{t}$ rather than the exact payoff vector $\ctilde{t}$.

Conditional on the half-step history $\Fil{t-1/2}$, the sampled
constraint $\jt{t}$ and random vector $\ut{t}$ have already been
revealed.  Hence the computed state $\psitilde{t}$ and the true payoff
vector $\ctilde{t}$ are fixed, and the only remaining randomness in
$\chat{t}$ comes from the coordinate sample $\isamp{t}$.  We may
therefore apply
\Cref{lem:energy-estimator-moments,lem:clipping-properties}
conditionally on $\Fil{t-1/2}$. Specifically, taking
$\bA=\A{j}$,
$\bpsi=\psitilde{t}$,
$\bX=\Xtilde{t}$, and clipping threshold $B=\Bloc$, these results give,
for every $j\in[\mc]$,
\begin{align}
\abs{[\chat{t}]_j}
&\leq
\Bloc,
\label{eqn:constraint-feedback-range}\\
\E\left[
[\chat{t}]_j^2
\,\middle|\,
\Fil{t-1/2}
\right]
&\leq
\Rx^2\om^2,
\label{eqn:constraint-feedback-second-moment}\\
\abs{\E\left[
[\chat{t}]_j
\,\middle|\,
\Fil{t-1/2}
\right]
-
[\ctilde{t}]_j}
&\leq
\frac{\Rx^2\om^2}{\Bloc}.
\label{eqn:constraint-feedback-bias}
\end{align}
The first bound is the worst-case range imposed by clipping.  The second
is the sharper second-moment control inherited from the shared-coordinate
estimator, while the third quantifies the bias introduced by clipping.

Set $\sigma:= \Rx\om$ and $b:=\sigma^2/\Bloc$.
To match the notation of \Cref{lem:stochastic-hedge}, use the shifted
filtration $\mathcal G_t:=\Fil{t+1/2}$ for $0\leq t<\TT$, and $\mathcal G_{\TT}:=\Fil{\TT}$. The feedback $\chat{t}$ is $\mathcal G_t$-measurable, and $\ctilde{t}$ and $\phat{t}$ are $\mathcal G_{t-1}$-measurable.
Then
$\mathcal G_{t-1}=\Fil{t-1/2}$, so
\Cref{eqn:constraint-feedback-range},
\Cref{eqn:constraint-feedback-second-moment}, and 
\Cref{eqn:constraint-feedback-bias}
are exactly the hypotheses of \Cref{lem:stochastic-hedge} with $\bc_t=\ctilde{t}$ and $B=\Bloc$. Applying \Cref{lem:stochastic-hedge} with
$\delta_{\mathrm F}=\del/6$ therefore gives, except with probability at
most $\del/6$,
\begin{align}
\sum_{t=1}^{\TT}
\inner{\ctilde{t}}{\phat{t}}
-
\min_{j\in[\mc]}
\sum_{t=1}^{\TT}[\ctilde{t}]_j
&\leq
C\left[
\frac{\log(2\mc)}{\lrY}
+
\lrY\sigma^2\TT
+
\sigma\sqrt{\TT\Lam}
+
\Bloc\Lam
+
\frac{\sigma^2\TT}{\Bloc}
\right],
\label{eqn:mixed-constraint-regret-bound}
\end{align}
where $\Lambda_{\mathrm F}=\log(24(\mc+1)/\del)=O(\Lam)$.

It remains to substitute the parameter choices.  Since
\begin{align}
\Bloc
&=
\Theta\left(
\sigma\sqrt{\frac{\TT}{\Lam}}
\right),
\end{align}
the range and clipping-bias contributions balance, as
\begin{align}
\Bloc\Lam
+
\frac{\sigma^2\TT}{\Bloc}
&=
O\left(
\sigma\sqrt{\TT\Lam}
\right).
\label{eqn:constraint-clipping-balance}
\end{align}
For the learning-rate terms, the choice of $\lrY$ implies
\begin{align}
\frac{\log(2\mc)}{\lrY}
&=
O\left(
\Bloc\log(2\mc)
+
\sigma\sqrt{\TT\log(2\mc)}
\right)=
O\left(
\sigma\sqrt{\TT\Lam}
\right),
\end{align}
where we used $\log(2\mc)=O(\Lam)$ and the definition of $\Bloc$.
Likewise,
\begin{align}
\lrY\sigma^2\TT
&\leq
O\left(
\min\left\{
\frac{\sigma^2\TT}{\Bloc},
\sigma\sqrt{\TT\log(2\mc)}
\right\}
\right)=
O\left(
\sigma\sqrt{\TT\Lam}
\right).
\end{align}
Consequently,
\begin{align}
\sum_{t=1}^{\TT}
\inner{\ctilde{t}}{\phat{t}}
-
\min_{j\in[\mc]}
\sum_{t=1}^{\TT}[\ctilde{t}]_j\leq
O\left(
\Rx\om\sqrt{\TT\Lam}
\right).
\label{eqn:mixed-constraint-regret-specialized}
\end{align}

Finally, applying \Cref{lem:pure-constraint-sampling} with
$\delta_{\mathrm S}=\del/6$ gives, except with probability at most
$\del/6$,
\begin{align}
\Ry\sum_{t=1}^{\TT}
\left(
[\ctilde{t}]_{\jt{t}}
-
\inner{\ctilde{t}}{\phat{t}}
\right)
&\leq
O\left(
\Rx\Ry\om\sqrt{\TT\Lam}
\right).
\end{align}
The mixed-action regret has the same order after multiplication by
$\Ry$.  A union bound over the two events therefore shows that both terms
in \Cref{eqn:constraint-side-ledger} satisfy the desired bounds
simultaneously with probability at least $1-\del/3$.  Hence, as claimed,
\begin{align}
\VecSide
&\leq
O\left(
\Rx\Ry\om\sqrt{\TT\Lam}
\right).
\end{align}
\end{proof}

\subsection{Proof of the Main Theorem}
\label{subsec:main-theorem-proof}

We now combine the matrix- and constraint-side guarantees to prove the
main convergence and runtime bounds. The proof has three steps. First,
we substitute the two regret bounds into the realized gap decomposition
to obtain the convergence rate after $\TT$ rounds. Second, we choose
$\TT$ so that this gap is at most $\eps$. Finally, having fixed the
required number of rounds, we account for the arithmetic cost of
generating the corresponding trajectory, including the stochastic
feedback computation and the Lánczos approximation used on the matrix
side.

\begin{proof}[Proof of \Cref{thm:main}]
We first combine the matrix- and constraint-side guarantees in the
gap decomposition.  Apply
\Cref{thm:matrix-side-guarantee} with
$\del_X=\del/3$ and $\del_{\mathrm L}=\del/3$.  This theorem already
incorporates both the rank-one sampling fluctuation and the Lánczos
approximation error, and therefore gives the required matrix-side bound
except with probability at most $2\del/3$.  Separately,
\Cref{thm:constraint-side-guarantee} gives the constraint-side bound
except with probability at most $\del/3$.  A union bound therefore shows
that both guarantees hold simultaneously with probability at least
$1-\del$.  We condition on this event for the remainder of the accuracy
analysis.

For the matrix side, \Cref{thm:matrix-side-guarantee} gives
\begin{align}
\MatSide
&\leq
\frac{\Rx\log(4\dn)}{\lrX}
+
\frac{3\Rx\lrX\Ry^2\om^2\TT}{2}
+
2\Rx\Ry\om
\sqrt{
2\TT\log\frac{3}{\del}
}
+
\Rx\Ry\om.
\label{eqn:main-proof-matrix-bound}
\end{align}
With the learning rate
\begin{align}
\lrX
&=
\frac{1}{\Ry\om}\sqrt{
\frac{2\log(4\dn)}
{3\TT}
},
\end{align}
the first two terms balance, and hence
\begin{align}
\MatSide
&\leq
O\left(
\Rx\Ry\om
\sqrt{
\TT
\log\frac{4\dn}{\del}
}
+
\Rx\Ry\om
\right).
\label{eqn:main-proof-matrix-simplified}
\end{align}
The first term contains the matrix-learning and rank-one sampling
contributions, while the final additive term accounts for the cumulative
Lánczos approximation error.
On the constraint side,
\Cref{thm:constraint-side-guarantee} gives
\begin{align}
\VecSide
&\leq
O\left(
\Rx\Ry\om
\sqrt{
\TT\Lam
}
\right),
\label{eqn:main-proof-vector-bound}
\end{align}
where
$\Lam=\log(16\mc\dn/\del)$.  Since $\Lam$ also dominates the logarithmic
factor in \Cref{eqn:main-proof-matrix-simplified}, the two bounds combine
to give
\begin{align}
\MatSide+\VecSide
&=
O\left(
\Rx\Ry\om
\sqrt{
\TT\Lam
}
+
\Rx\Ry\om
\right).
\end{align}

By the realized-gap decomposition,
\Cref{eqn:realized-gap-decomposition}, the output gap is obtained by
dividing this cumulative error by the number of rounds.  Furthermore, since $T \geq 1$ and $\Lambda = \log(16mn/\delta) > 1$, we have that $1/T \leq \sqrt{\Lambda/T}$. Therefore,
\begin{align}
\Gap(\Xout,\zout)
&\leq
C
\Rx\Ry\om
\left(
\sqrt{
\frac{\Lam}{\TT}
}
+
\frac{1}{\TT}
\right) \leq
C'
\Rx\Ry\om
\sqrt{
\frac{\Lam}{\TT}
},
\label{eqn:main-gap-expanded}
\end{align}
for universal constants $C,C'$.  This proves \Cref{eqn:main-gap}. It remains to choose the number of rounds.  To ensure
$\Gap(\Xout,\zout)\leq\eps$, it is sufficient that
\begin{align}
\TT
&\geq
(C')^2
\left(
\frac{\Rx\Ry\om}{\eps}
\right)^2
\Lam=
(C')^2
\gam^2
\Lam.
\label{eqn:round-bound-derivation}
\end{align}
Thus, $\TT=\widetilde O(\gam^2)$ rounds suffice.

\paragraph{Runtime.}
We now bound the arithmetic cost of generating the $\TT$ rounds under
the sparse-access model of \Cref{def:sparse-access}.  On each round, the
algorithm must compute the stochastic feedback used by the constraint
update and apply Lánczos to construct the matrix-side rank-one response.
We will now bound each of these costs.

\vspace{0.1in}
\noindent \emph{Feedback and state-vector operations.}
On round $t$, Lánczos produces the practical state vector
$\psitilde{t}\in\R^{\dn}$ explicitly.  The corresponding rank-one matrix
$\rtilde{t}=\psitilde{t}\psitilde{t}^{\top}$ need not be materialized
and is represented implicitly by $\psitilde{t}$. 

As described in \Cref{sec:squared-magnitude-sampling}, once
$\psitilde{t}$ has been constructed explicitly, we can sample a
coordinate $\isamp{t}$ in $O(\dn)$ time according to $\Prb[
\isamp{t}=i
\mid
\psitilde{t}
]=
|[\psitilde{t}]_i|^2.$
Since $\psitilde{t}$ is stored explicitly, we also retain $O(1)$ query
access to each of its coordinates. For a fixed constraint $j$, the
sampled entry required by the shared-coordinate estimator is
\begin{align}
(\A{j}\psitilde{t})_{\isamp{t}}
&=
\sum_{\ell:\,[\A{j}]_{\isamp{t},\ell}\neq0}
[\A{j}]_{\isamp{t},\ell}
[\psitilde{t}]_{\ell}.
\end{align}
Since each row of $\A{j}$ has at most $\sr$ nonzero entries, this
quantity can be evaluated using $O(\sr)$ calls to the sparse-access
oracles.  Reusing the same sampled coordinate $\isamp{t}$ for all
$\mc$ constraints therefore computes the feedback vector $\chat{t}$ in
$O(\mc\sr)$ time.

Updating and normalizing the $\mc$ exponential weights and sampling
$\jt{t}\sim\phat{t}$ requires $O(\mc)$ additional work, while generating
and normalizing the fresh unit vectors costs $O(\dn)$.  As described in
\Cref{sec:squared-magnitude-sampling}, sampling $\isamp{t}$ from the
squared-magnitude distribution of $\psitilde{t}$ also costs $O(\dn)$
time.  Since $\sr\geq1$, the total work outside Lánczos is therefore
$O(\dn+\mc\sr)$ per round.

\vspace{0.1in}
\noindent \emph{Hamiltonian-vector products.}
At the beginning of round $t$, the accumulated Hamiltonian is
\begin{align}
\Ht{t}
&=
\Ry
\sum_{\tau=1}^{t-1}
\A{\jt{\tau}}.
\label{eqn:hamiltonian-sampled-sum}
\end{align}
Grouping repeated samples shows that at most
${\min\{t-1,\mc\}}$ distinct sparse constraint matrices appear in this
sum.  Computing a single Hamiltonian-vector product $\Ht{t}\bv$ by
applying these matrices separately therefore takes $O\left(\dn\sr\min\{t,\mc\}\right)$ time.
Alternatively, the accumulated Hamiltonian may be maintained explicitly
as a dense matrix and applied in $O(\dn^2)$ time.  Taking the cheaper of
these two implementations, one multiplication by $\Ht{t}$ results in an overall runtime of
\begin{align}
O\left(
\dn\sr
\min\left\{
t,
\mc,
\frac{\dn}{\sr}
\right\}
\right).
\label{eqn:hamiltonian-action-cost}
\end{align}
The three entries in the minimum correspond respectively to the number
of preceding rounds, the total number of distinct constraints, and the
point at which dense matrix-vector multiplication becomes cheaper than
summing sparse products.
The dense alternative need only be initialized at the first round
$t$ for which $\min\{t-1,\mc\}\geq\dn/\sr$.
At that point, forming the sum from its distinct sampled matrices costs
$O(\dn^2)$, including initialization. Later updates cost $O(\dn\sr)$
each.  Before this crossover we use the sparse representation.
The one-time $O(\dn^2)$ cost and the subsequent update costs are absorbed
by \Cref{eqn:sampled-runtime}, since the dense branch can be selected
only when $\TT,\mc\geq\dn/\sr$, and its matrix--vector products already
cost $O(\dn^2)$ each.

\vspace{0.1in}
\noindent \emph{Lánczos cost.}
By \Cref{lem:lanczos-projector}, the Lánczos degree on round $t$ is
determined by the size of the accumulated exponential argument and the
required approximation accuracy.  Under the chosen matrix learning rate,
$\lrX\Ry\om\TT=\widetilde O(\sqrt{\TT})$, and the degree bound from that
lemma therefore gives
\begin{align}
k_t
&=
\widetilde O\left(
\min\left\{
\dn,
\TT^{1/4}
\right\}
\right)
\end{align}
uniformly over $t\in[\TT]$.  Summing over all rounds,
\begin{align}
\sum_{t=1}^{\TT} k_t
&=
\widetilde O\left(
\TT^{5/4}
\right).
\label{eqn:lanczos-degree}
\end{align}
The exponent $1/4$ arises because the norm of the accumulated
exponential argument grows as $\widetilde O(\sqrt{\TT})$, while the
Lánczos degree scales as the square root of this norm.

On round $t$, a $k_t$-step exact-arithmetic Lánczos computation requires
$k_t$ Hamiltonian-vector products, together with
$O(k_t\dn+k_t^2)$ additional arithmetic operations for the Krylov basis
and the tridiagonal compression.  By
\Cref{eqn:hamiltonian-action-cost}, its cost is therefore
\begin{align}
O\left(
k_t
\dn\sr
\min\left\{
t,
\mc,
\frac{\dn}{\sr}
\right\}
+
k_t\dn
+
k_t^2
\right).
\label{eqn:lanczos-one-round-cost}
\end{align}
Moreover, \Cref{lem:lanczos-projector} caps the Lánczos degree by the
ambient dimension, so $k_t\leq\dn$ for every $t$.  Hence,
\begin{align}
\sum_{t=1}^{\TT}
k_t^2
&\leq
\dn
\sum_{t=1}^{\TT}
k_t.
\end{align}
Summing over all rounds and leveraging
\Cref{eqn:lanczos-degree}, gives an overall runtime of
\begin{align}
\widetilde O\left(
\dn\sr
\min\left\{
\TT,
\mc,
\frac{\dn}{\sr}
\right\}
\TT^{5/4}
+
(\dn+\mc\sr)\TT
\right).
\label{eqn:sampled-runtime}
\end{align}
Here the first term is the total Lánczos cost, dominated by the
Hamiltonian-vector products, while the second is the accumulated
feedback and state-vector cost. Finally, substituting the round complexity
$\TT=\widetilde O(\gam^2)$ derived above yields
\begin{align}
\widetilde O\left(
\dn\sr
\min\left\{
\gam^2,
\mc,
\frac{\dn}{\sr}
\right\}
\gam^{5/2}
+
(\dn+\mc\sr)\gam^2
\right),
\label{eqn:main-runtime}
\end{align}
which is the runtime claimed in
\Cref{eqn:main-theorem-runtime}.
Note that the primal output $\Xout$ is represented by the $\TT$
rank-one factors generated over the trajectory, while the dual output
$\zout$ is represented by the empirical distribution of the sampled
constraints.  These are precisely the output representations claimed in
the theorem statement, thereby completing the proof.
\end{proof}

The preceding theorem gives the algorithmic core of our result, by solving
the spectrahedron--simplex game associated with a fixed target value in
sublinear sparse-oracle time.  To recover a solution of the original SDP,
it remains only to combine this game solver with the target-value reduction
from \Cref{subsec:game-to-sdp-solution}.  That reduction performs a
logarithmic search over objective thresholds and, for each returned game
solution, evaluates the one-sided upper certificate from
\Cref{lem:approximate-upper-certificate}.  The next corollary shows that
these additional steps preserve the running time of
\Cref{thm:main} up to polylogarithmic factors, yielding the corresponding
end-to-end SDP solver.

\begin{corollary}[End-to-End Two-Sided Stochastic SDP Solver]
\label{cor:stochastic-sdp-solver}
Consider an SDP of the form
\Cref{eqn:intro-standard-sdp-primal,eqn:intro-standard-sdp-dual} under the sparse-access model of
\Cref{def:sparse-access}.  Assume strong duality, supplied primal and
dual radii $\Rx,\Ry\geq1$, and that the input matrices are
$\sr$-row-sparse with operator norm at most one. Then, for every $\eps,\del\in(0,1)$, there is a randomized
exact-arithmetic sparse-oracle algorithm that, with probability at least
$1-\del$, returns a matrix
$\widehat{\bX}\succeq\bzero$ with
$\Tr(\widehat{\bX})\leq\Rx$ satisfying
\begin{align}
    \inner{\bC}{\widehat{\bX}}
    &\geq
    \operatorname{OPT}-\eps,
    \\
    \inner{\A{j}}{\widehat{\bX}}
    &\leq
    b_j+\frac{\eps}{\Ry}
    \qquad
    \forall~j\in[\mc].
\end{align}
The primal output is represented implicitly by
$\widetilde O(\gam^2)$ rank-one factors, where
$\gam:=\Rx\Ry/\eps$, and the running time is
\begin{align}
    \widetilde O\left(
        \dn\sr
        \min\left\{
            \gam^2,
            \mc,
            \frac{\dn}{\sr}
        \right\}
        \gam^{5/2}
        +
        (\dn+\mc\sr)\gam^2
    \right).
    \label{eqn:end-to-end-stochastic-runtime}
\end{align}
\end{corollary}

\begin{proof}
Apply the optimization-to-feasibility reduction of
\Cref{cor:target-value-binary-search}, with
$\eta=\eps/2$ and
$\xi=\zeta=\eps/(8(1+\Ry))$.  For each
queried target value $g$, the resulting feasibility game has matrix radius
$\Rx$, simplex radius one, and supplied payoff-width bound $\om_D=2$. Since it must be
solved to saddle-point gap $\xi$, its effective inverse-accuracy parameter is
\begin{align}
    \gam_D
    &:=
    \frac{\Rx\om_D}{\xi}
    =
    \Theta\left(
        \frac{\Rx\Ry}{\eps}
    \right)
    =
    \Theta(\gam),
\end{align}
where we used $\Ry\geq1$.  Therefore, by \Cref{thm:main}, each
target-value game can be solved in time
\begin{align}
    \widetilde O\left(
        \dn\sr
        \min\left\{
            \gam^2,
            \mc,
            \frac{\dn}{\sr}
        \right\}
        \gam^{5/2}
        +
        (\dn+\mc\sr)\gam^2
    \right).
    \label{eqn:target-game-solver-runtime}
\end{align}

It remains to account for the computable upper certificate required by
\Cref{cor:target-value-binary-search}.  The constraint strategy returned
by the stochastic game solver is the empirical distribution
\begin{align}
    \bp^{\mathrm{out}}
    &=
    \frac{1}{\TT}
    \sum_{t=1}^{\TT}
    \be_{\jt{t}},
\end{align}
so
\begin{align}
    \bM_g
    &:=
    \sum_{j=0}^{\mc}
    p_j^{\mathrm{out}}\bD_j(g)
\end{align}
contains at most $\min\{\TT,\mc+1\}$ distinct payoff matrices.
The empirical distribution $\bp^{\mathrm{out}}$ is supported only on
constraints that were actually sampled, so after combining repeated
indices the mixture $\bM_g$ contains at most
$\min\{\TT,\mc+1\}$ distinct payoff matrices.  Applying these sparse
matrices separately to a vector therefore costs
$O(\dn\sr\min\{\TT,\mc\})$ time. 
Alternatively, if the game solver has already switched to a dense
representation of the accumulated Hamiltonian, then no additional matrix
needs to be formed.  In this case,
$\bM_g=\Ht{\TT+1}/\TT$ is already available explicitly and can be applied
to a vector in $O(\dn^2)$ time.  The one-time $O(\dn^2)$ cost of forming
this dense representation is already included in the runtime of the game
solve.  If the game solver never forms a dense representation, we instead
use the sparse implementation above.

Taking the cheaper of these two implementations gives
\begin{align}
    \operatorname{Time}(\bM_g\bv)
    &=
    O\left(
        \dn\sr
        \min\left\{
            \TT,
            \mc,
            \frac{\dn}{\sr}
        \right\}
    \right).
\end{align}
The additional coordinate introduced by the fixed-trace embedding changes
the dimension only by one, while the identity components of the matrices
$\bD_j(g)$ combine into a single scalar multiple of the identity and can
be applied in $O(\dn)$ time.  By
\Cref{lem:approximate-upper-certificate}, computing
$\widetilde u_g$ therefore costs
\begin{align}
    \widetilde O\left(
        \dn\sr
        \min\left\{
            \TT,
            \mc,
            \frac{\dn}{\sr}
        \right\}
        \sqrt{
            \frac{\Rx\om_D}{\zeta}
        }
    \right).
\end{align}
For the parameter choices above, this is lower order than the cost of the
corresponding game solve.

Let
$K:=1+\max\{0,\lceil\log_2(2\Rx/\eta)\rceil\}=O(\log(2+\Rx/\eps))$
bound the total number of queried targets in
\Cref{cor:target-value-binary-search}, including the initial solve at
$g=-\Rx$, where $\eta=\Theta(\eps)$.  Each query invokes the stochastic
game solver once and computes one upper certificate.  Assign failure probability
$\del/(2K)$ to each of these two randomized computations at every query.
Each invocation uses fresh randomness, so its failure guarantee holds
conditional on the preceding search history.  A union bound over at
most $2K$ failure events therefore shows that, with probability
at least
\begin{align}
    1
    -
    2K\frac{\del}{2K}
    &=
    1-\del,
\end{align}
all game solves and upper-certificate computations succeed
simultaneously.  On this event,
\Cref{cor:target-value-binary-search} gives objective
error at most $\eta+\xi+\zeta\leq\eps$ and constraint violations at most
$\xi+\zeta\leq\eps/\Ry$. Replacing the failure probability in each invocation by
$\del/(2K)$ changes its running time only by logarithmic factors and multiplying the per-step cost by
$K=O(\log(2+\Rx/\eps))$ introduces only additional
polylogarithmic factors.  This therefore proves the claimed end-to-end
running time and success probability.
\end{proof}

The runtimes in \Cref{thm:main} and \Cref{cor:stochastic-sdp-solver} use the implicit
output representation produced by the stochastic solver.  For each
target-value game, the matrix output is represented by its $\TT$ rank-one
factors, while the constraint-player output is represented by the
histogram of the sampled constraints.  Consequently, the final primal
matrix returned by \Cref{cor:stochastic-sdp-solver} can likewise be
represented by the rank-one factors from the retained target-value solve,
rather than explicitly materialized as a dense matrix.

This implicit representation is necessary for a sublinear input/output
guarantee. In particular, explicitly writing a generic $\dn\times\dn$ primal matrix
already requires $\Omega(\dn^2)$ time.  Explicitly forming the primal
matrix from its $\TT$ rank-one factors would add
$O(\TT\dn^2)$ arithmetic operations. Similarly, if the constraint matrices are supplied explicitly as sparse
adjacency lists and the cost of reading the full input is charged upfront,
then the input size and preprocessing cost are
$\Theta(\mc\dn\sr)$.  This term is absent from our sparse-oracle runtime,
which charges only for the portions of the input accessed along the
realized trajectory.

\begingroup
\small
\bibliographystyle{alpha}
\bibliography{references} 
\endgroup

\newpage 
\appendix

\section{Hermitian Inputs via ``Realification''}
\label{sec:realification}

The two-sided stochastic analysis of \Cref{sec:algorithm} is stated for
real symmetric matrices and real random unit vectors, matching the setting
of the matrix-regret theorem of Carmon--Duchi--Sidford--Tian
\cite{CDST19}.  This real-symmetry assumption is not a restriction on the
SDPs to which the algorithm applies.  In this section, we show explicitly
how a Hermitian SDP can be converted into an equivalent real symmetric
SDP in twice the matrix dimension.  The transformation preserves
feasibility, objective values, operator norms, and the primal--dual gap.
We may therefore run the stochastic solver on the realified instance and
map its primal output back to a complex Hermitian matrix.

The construction is the standard identification of a complex vector
$\bz=\bx+\mathrm{i}\by\in\C^{\dn}$ with the real vector
$(\bx,\by)\in\R^{2\dn}$.  For an arbitrary complex matrix
$\bA\in\C^{\dn\times\dn}$, define
\begin{align}
\Rfy(\bA)
&:=
\begin{pmatrix}
\operatorname{Re}\bA & -\operatorname{Im}\bA \\
\operatorname{Im}\bA & \operatorname{Re}\bA
\end{pmatrix}
\in
\R^{2\dn\times2\dn}.
\label{eqn:realification-map}
\end{align}
Under the identification above,
\begin{align}
\Rfy(\bA)
\begin{pmatrix}
\bx\\
\by
\end{pmatrix}
&=
\begin{pmatrix}
\operatorname{Re}(\bA\bz)\\
\operatorname{Im}(\bA\bz)
\end{pmatrix}.
\label{eqn:realification-action}
\end{align}
In particular, if $\bA$ is Hermitian, then $\Rfy(\bA)$ is real
symmetric.

We will also need to map a real primal solution back to the original
complex space.  For a real symmetric block matrix
\begin{align}
\bY
&=
\begin{pmatrix}
\bY_{11} & \bY_{12} \\
\bY_{21} & \bY_{22}
\end{pmatrix}
\in
\R^{2\dn\times2\dn},
\end{align}
define
\begin{align}
\Cfy(\bY)
&:=
\bY_{11}+\bY_{22}
+
\mathrm{i}(\bY_{21}-\bY_{12}).
\label{eqn:complexification-map}
\end{align}

\paragraph{Realification of the SDP.}
To make the reduction explicit, consider the standard Hermitian SDP
\begin{align}
\operatorname{OPT}_{\C}
&=
\max_{\bX\succeq0}
\left\{
\inner{\bC}{\bX}:
\inner{\A{j}}{\bX}\leq b_j,
~
\forall~j\in[\mc]
\right\},
\label{eqn:hermitian-sdp-primal}
\end{align}
with dual
\begin{align}
\operatorname{OPT}_{\C}
&=
\min_{\by\in\R_{\geq0}^{\mc}}
\left\{
\mathbf b^\top\by:
\sum_{j=1}^{\mc}y_j\A{j}
\succeq
\bC
\right\}.
\label{eqn:hermitian-sdp-dual}
\end{align}
Here $\bC,\A{1},\ldots,\A{\mc}$ are Hermitian.  Its realification is
obtained simply by replacing every Hermitian data matrix by its block
realification:
\begin{align}
\operatorname{OPT}_{\R}
&=
\max_{\bY\succeq0}
\left\{
\inner{\Rfy(\bC)}{\bY}:
\inner{\Rfy(\A{j})}{\bY}\leq b_j,~\forall~j\in[\mc]
\right\},
\label{eqn:realified-sdp-primal}
\end{align}
with dual
\begin{align}
\operatorname{OPT}_{\R}
&=
\min_{\by\in\R_{\geq0}^{\mc}}
\left\{
\mathbf b^\top\by:
\sum_{j=1}^{\mc}
y_j\Rfy(\A{j})
\succeq
\Rfy(\bC)
\right\}.
\label{eqn:realified-sdp-dual}
\end{align}
Thus the thresholds $b_j$ and the dual variable $\by$ are unchanged.
Only the matrix data and matrix dimension are modified.

The primal correspondence includes a factor of $1/2$.  Given a
Hermitian primal matrix $\bX$, its real representative is
\begin{align}
\bY
&:=
\frac12\Rfy(\bX).
\label{eqn:realification-primal-embedding}
\end{align}
As shown below, this normalization ensures $\Tr(\bY)=\Tr(\bX)$ and $\inner{\Rfy(\bA)}{\bY}
=\inner{\bA}{\bX}$,
for every Hermitian $\bA$.  Hence the trace radius, all constraint
values, and the objective value are preserved.  In the reverse
direction, a real feasible matrix $\bY$ maps to
$\bX=\Cfy(\bY)$ with the same trace and the same Hermitian payoffs.

The dual correspondence requires no rescaling.  By linearity of
realification,
\begin{align}
\sum_{j=1}^{\mc}
y_j\Rfy(\A{j})
-\Rfy(\bC)
&=
\Rfy\left(
\sum_{j=1}^{\mc}
y_j\A{j}
-\bC
\right).
\label{eqn:realification-dual-slack}
\end{align}
Since realification preserves positive semidefiniteness, $\by$ is
feasible for \Cref{eqn:hermitian-sdp-dual} if and only if it is feasible
for \Cref{eqn:realified-sdp-dual}, and its objective
$\mathbf b^\top\by$ is unchanged.  Consequently, the original and
realified SDPs have the same optimum.

This reduction is compatible with the standard SDP-to-game reduction of
\Cref{sec:standard-sdp-to-game}.  One may first replace
$\bC,\A{1},\ldots,\A{\mc}$ by their realifications and then apply the
same lifting and normalization procedure to the resulting real SDP.
Equivalently, at the level of the normalized game, each Hermitian game
matrix is replaced by its realification, while the constraint-player
variable and its radius are unchanged.  We next record the resulting
pointwise correspondence between the two games.

The block map $\Rfy$ is the standard real representation used for complex
semidefinite programs---see, for example,
\cite[Section~2.4, Eq.~(4)]{WaldspurgerEtAl2015}.  Let
$\Gap_{\C}$ denote the saddle-point gap for the game with Hermitian
constraint matrices $\A{1},\ldots,\A{\mc}$, and let $\Gap_{\R}$ denote
the corresponding gap after replacing these matrices by
$\Rfy(\A{1}),\ldots,\Rfy(\A{\mc})$.

\begin{lemma}[SDP Gap Preservation Under ``Realification'']
\label{lem:realification-gap}
Let $\A{1},\ldots,\A{\mc}\in\C^{\dn\times\dn}$ be Hermitian.  Then the
following properties hold:

\begin{enumerate}
\item \textbf{Feasible-set and payoff correspondence.}
For every real symmetric $\bY\in\Rx\cS_{2\dn}$, the matrix
$\bX:=\Cfy(\bY)$ belongs to $\Rx\cS_{\dn}$ and satisfies, for every $j\in[\mc]$,
\begin{align}
\inner{\A{j}}{\bX}
&=
\inner{\Rfy(\A{j})}{\bY}.
\label{eqn:realification-payoff}
\end{align}
Conversely, every
$\bX\in\Rx\cS_{\dn}$ has the real embedding
$\bY=\frac12\Rfy(\bX)\in\Rx\cS_{2\dn}$, with
$\Cfy(\bY)=\bX$.

\item \textbf{Spectral correspondence.}
For every real vector $\by\in\R^{\mc}$,
\begin{align}
\lambda_{\max}\left(
\sum_{j=1}^{\mc}
y_j\Rfy(\A{j})
\right)
&=
\lambda_{\max}\left(
\sum_{j=1}^{\mc}
y_j\A{j}
\right).
\label{eqn:realification-eigenvalue}
\end{align}

\item \textbf{Gap preservation.}
The real and Hermitian games have the same saddle-point value and,
for every real symmetric $\bY\in\Rx\cS_{2\dn}$ and
$\by\in\Ry\Delta_{\mc}$,
\begin{align}
\Gap_{\R}(\bY,\by)
&=
\Gap_{\C}(\Cfy(\bY),\by).
\label{eqn:realification-gap}
\end{align}
\end{enumerate}
\end{lemma}

\begin{proof}
The proof mirrors the three assertions of the lemma.  We first establish
the two-way correspondence between the real and Hermitian feasible sets
and show that it preserves every constraint payoff.  We then prove that
realification preserves the spectrum of a Hermitian matrix, which gives
the equality of the matrix-player best-response values.  Combining these
two correspondences yields the equality of the saddle-point values and,
more strongly, the pointwise equality of the primal--dual gaps.

\paragraph{Feasible-set and payoff correspondence.}
Fix a real symmetric $\bY\in\Rx\cS_{2\dn}$.  Since $\bY\succeq0$, write a rank-one
decomposition
\begin{align}
\bY
&=
\sum_{\ell}
\bw_{\ell}\bw_{\ell}^{\top},
\qquad \text{where} \quad 
\bw_{\ell}=
\begin{pmatrix}
\bx_{\ell}\\
\by_{\ell}
\end{pmatrix}
\in\R^{2\dn}.
\end{align}
For each $\ell$, define
$\bz_{\ell}:=\bx_{\ell}+\mathrm{i}\by_{\ell}\in\C^{\dn}$.
A direct calculation from the definition of $\Cfy$ gives
\begin{align}
\Cfy\left(
\bw_{\ell}\bw_{\ell}^{\top}
\right)
&=
\bz_{\ell}\bz_{\ell}^{\dagger}.
\label{eqn:rank-one-complexification}
\end{align}
Therefore, by linearity,
\begin{align}
\Cfy(\bY)
&=
\sum_{\ell}
\bz_{\ell}\bz_{\ell}^{\dagger}
\succeq0.
\label{eqn:complexification-positive}
\end{align}
Moreover,
\begin{align}
\Tr(\Cfy(\bY))
&=
\sum_{\ell}
\norm{\bz_{\ell}}_2^2
=\sum_{\ell}
\norm{\bw_{\ell}}_2^2
=
\Tr(\bY).
\label{eqn:complexification-trace}
\end{align}
Since $\bY\in\Rx\cS_{2\dn}$, it follows that
$\bX:=\Cfy(\bY)\in\Rx\cS_{\dn}$.

It remains to show that this map preserves the constraint payoffs.  Let
$\bA$ be any Hermitian matrix.  By
\Cref{eqn:realification-action}, each rank-one factor satisfies
\begin{align}
\bw_{\ell}^{\top}
\Rfy(\bA)
\bw_{\ell}
&=
\bz_{\ell}^{\dagger}
\bA
\bz_{\ell}.
\label{eqn:rank-one-realification-payoff}
\end{align}
Using the rank-one decomposition of $\bY$,
\begin{align}
\inner{\Rfy(\bA)}{\bY}
&=
\sum_{\ell}
\inner{\Rfy(\bA)}
{\bw_{\ell}\bw_{\ell}^{\top}}
=
\sum_{\ell}
\bz_{\ell}^{\dagger}
\bA
\bz_{\ell}
=
\inner{\bA}{
\sum_{\ell}
\bz_{\ell}\bz_{\ell}^{\dagger}
}
=
\inner{\bA}{\Cfy(\bY)}.
\label{eqn:complexification-payoff}
\end{align}
Taking $\bA=\A{j}$ proves
\Cref{eqn:realification-payoff}.

For the converse direction, fix $\bX\in\Rx\cS_{\dn}$.  Since
$\bX\succeq0$, also $\Rfy(\bX)\succeq0$, and the block definitions give $\Tr(\Rfy(\bX))
=2\Tr(\bX)$ and $\Cfy(\Rfy(\bX))=2\bX$.
Hence,
$\bY:=\frac12\Rfy(\bX)$ lies in $\Rx\cS_{2\dn}$ and satisfies
$\Cfy(\bY)=\bX$.  This completes the proof of the feasible-set and payoff
correspondence.

\paragraph{Spectral correspondence.}
We next prove \Cref{eqn:realification-eigenvalue}.  Let
$\bA$ be Hermitian and let
$\bz=\bx+\mathrm{i}\by$ be an eigenvector with eigenvalue $\lambda$.
Since $\lambda\in\R$, \Cref{eqn:realification-action} gives
\begin{align}
\Rfy(\bA)
\begin{pmatrix}
\bx\\
\by
\end{pmatrix}
&=
\lambda
\begin{pmatrix}
\bx\\
\by
\end{pmatrix}.
\end{align}
Applying the same identity to
$\mathrm{i}\bz=-\by+\mathrm{i}\bx$ gives
\begin{align}
\Rfy(\bA)
\begin{pmatrix}
-\by\\
\bx
\end{pmatrix}
&=
\lambda
\begin{pmatrix}
-\by\\
\bx
\end{pmatrix}.
\end{align}
Thus every eigenvalue of $\bA$ appears twice in $\Rfy(\bA)$, and these
$2\dn$ eigenvalues exhaust its spectrum.  In particular, $\norm{\Rfy(\bA)}=\norm{\bA}$ and 
\begin{align}
\lambda_{\max}(\Rfy(\bA))
&=
\lambda_{\max}(\bA).
\label{eqn:realification-spectral-identities}
\end{align}
Now fix $\by\in\R^{\mc}$.  By linearity of $\Rfy$,
\begin{align}
\sum_{j=1}^{\mc}
y_j\Rfy(\A{j})
&=
\Rfy\left(
\sum_{j=1}^{\mc}
y_j\A{j}
\right).
\end{align}
Since $\sum_j y_j\A{j}$ is Hermitian,
\Cref{eqn:realification-spectral-identities} immediately gives
\Cref{eqn:realification-eigenvalue}.

\paragraph{Gap preservation.}
We finally combine the two correspondences above.  Fix
a real symmetric $\bY\in\Rx\cS_{2\dn}$ and $\by\in\Ry\Delta_{\mc}$, and set
$\bX:=\Cfy(\bY)$.  By \Cref{eqn:realification-eigenvalue},
\begin{align}
\lambda_{\max}\left(
\sum_{j=1}^{\mc}
y_j\Rfy(\A{j})
\right)
&=
\lambda_{\max}\left(
\sum_{j=1}^{\mc}
y_j\A{j}
\right),
\end{align}
while \Cref{eqn:realification-payoff} gives
\begin{align}
\min_{j\in[\mc]}
\inner{\Rfy(\A{j})}{\bY}
&=
\min_{j\in[\mc]}
\inner{\A{j}}{\bX}.
\end{align}
Substituting these two identities into the definition of the gap yields
\begin{align}
\Gap_{\R}(\bY,\by)
&=
\Gap_{\C}(\bX,\by)
=\Gap_{\C}(\Cfy(\bY),\by),
\end{align}
which proves \Cref{eqn:realification-gap}.

Finally, the feasible-set and payoff correspondence shows that, for
every fixed constraint-player strategy, optimizing over the real matrix
actions gives exactly the same value as optimizing over the Hermitian
matrix actions.  The two games therefore have the same saddle-point
value.  This completes the proof.
\end{proof}

Realification also respects products and real-coefficient analytic matrix
functions.  In particular, matrix multiplication obeys
\begin{align}
\Rfy(\bA\bBmat)
&=
\Rfy(\bA)\Rfy(\bBmat)
\end{align}
and matrix exponentiation obeys 
\begin{align}
    \Rfy(e^{\bH})
&=
e^{\Rfy(\bH)}.
\end{align}
Thus all matrix operations appearing in the stochastic solver may be
carried out directly on the ``realified'' matrices.  Importantly, the
realification argument does not require us to identify the real unit vectors used by the Carmon--Duchi--Sidford--Tian \cite{CDST19} learner in dimension
$2\dn$ with the complex Haar direction used elsewhere in the paper. Rather, the
solver is simply analyzed and run on the real symmetric instance.

The overhead of the reduction is small.  The realified solver operates in
dimension $2\dn$.  Moreover, $\norm{\Rfy(\A{j})}=\norm{\A{j}}$
and if each Hermitian $\A{j}$ has at most $\sr$ nonzero complex entries
per row, then each row of $\Rfy(\A{j})$ has at most $2\sr$ nonzero real
entries.  Hence, realification changes the iteration and arithmetic bounds
only by constant factors.
A realified row is produced by scanning the corresponding original
complex row and emitting its real and imaginary components, omitting
zeros.  This takes $O(\sr)$ time, so the row-scan and matrix--vector
access costs used above are preserved.  Also, each stored real factor
$\bw=(\bx,\by)$ maps to the single complex factor
$\bz=\bx+\mathrm{i}\by$ in $O(\dn)$ time by
\Cref{eqn:rank-one-complexification}, preserving the implicit output
size.

Concretely, to apply the stochastic SDP solver to a Hermitian instance,
we replace every Hermitian matrix in the instance by its realification,
run the real symmetric algorithm in the doubled dimension, and map its
primal output $\bY^{\mathrm{out}}$ back to
$\Cfy(\bY^{\mathrm{out}})$.  The constraint-player output is unchanged.
By \Cref{lem:realification-gap}, this mapping preserves the
primal--dual gap exactly.
\end{document}